\documentclass[a4paper]{article}
\usepackage{enumitem}
\usepackage{calc}

\usepackage{lmodern} 

\usepackage{graphicx} 

\usepackage{amsmath}
\usepackage{amssymb,tikz}
\usepackage{pict2e}
\usepackage{amsthm}
\usepackage{amscd}
\usepackage{mathrsfs}
\usepackage{todonotes}
\usetikzlibrary{calc} 

\usepackage{yfonts}

\usepackage{marvosym}
\usepackage[percent]{overpic}

\newtheorem{theorem}{Theorem} [section]
\newtheorem{proposition}[theorem]{Proposition}	
\newtheorem{corollary}[theorem]{Corollary}	
\newtheorem{lemma}[theorem]{Lemma}		

\newtheorem{assumptions}[theorem]{Assumptions}

\newtheorem{remark}[theorem]{Remark}
\theoremstyle{definition}

\newcommand{\C}{\mathbb{C}}
\newcommand{\R}{\mathbb{R}}

\newcommand{\re}{\text{\upshape Re\,}}
\newcommand{\im}{\text{\upshape Im\,}}

\newcommand{\Ai}{{\rm Ai}}

\def\XXint#1#2#3{{\setbox0=\hbox{$#1{#2#3}{\int}$}
\vcenter{\hbox{$#2#3$}}\kern-.5\wd0}}

\usepackage{tikz}
\usetikzlibrary{arrows}
\usetikzlibrary{decorations.pathmorphing}
\usetikzlibrary{decorations.markings}
\usetikzlibrary{patterns}
\usetikzlibrary{automata}
\usetikzlibrary{positioning}
\usepackage{tikz-cd}
\tikzset{->-/.style={decoration={
				markings,
				mark=at position #1 with {\arrow{latex}}},postaction={decorate}}}
	
	\tikzset{-<-/.style={decoration={
				markings,
				mark=at position #1 with {\arrowreversed{latex}}},postaction={decorate}}}

\usetikzlibrary{shapes.misc}\tikzset{cross/.style={cross out, draw, 
         minimum size=2*(#1-\pgflinewidth), 
         inner sep=0pt, outer sep=0pt}}

\usepackage{pgfplots}
	
\numberwithin{equation}{section}

\def\bigO{{\cal O}}

\usepackage[colorlinks=true]{hyperref}
\hypersetup{urlcolor=blue, citecolor=red, linkcolor=blue}

\begin{document}
\title{Global rigidity and exponential moments \\ for soft and hard edge point processes}
\author{Christophe Charlier, Tom Claeys}

\maketitle

\begin{abstract}
We establish sharp global rigidity upper bounds for universal determinantal point processes describing edge eigenvalues of random matrices. For this, we first obtain a general result which can be applied to general (not necessarily determinantal) point processes which have a smallest (or largest) point: it allows to deduce global rigidity upper bounds from the exponential moments of the counting function of the process. By combining this with known exponential moment asymptotics for the Airy and Bessel point processes, we improve on the best known upper bounds for the global rigidity of the Airy point process, and we obtain new global rigidity results for the Bessel point process. 

Secondly, we obtain exponential moment asymptotics for the Wright's generalized Bessel process and the Meijer-$\mathrm{G}$ process, up to and including the constant term. As a direct consequence, we obtain new results for the expectation and variance of the associated counting functions. Furthermore, by combining these asymptotics with our general rigidity theorem, we obtain new global rigidity upper bounds for these point processes.
\end{abstract}

\noindent
{\small{\sc AMS Subject Classification (2010)}: 41A60, 60B20, 33B15, 33E20, 35Q15.}

\noindent
{\small{\sc Keywords}: Rigidity, Exponential moments, Muttalib--Borodin ensembles, Product random matrices, Random matrix theory, Asymptotic analysis, Large gap probability, Riemann--Hilbert problems.}


\section{Introduction and statement of results}
\label{Section:intro}

An important question in recent years in random matrix theory has been to understand how much the ordered eigenvalues of a random matrix can deviate from their typical locations. It has been observed \cite{Johansson98, Gustavsson, ArguinBeliusBourgade, ErdosSchleinYau, ErdosYauYin, CFLW} that the individual eigenvalues fluctuate on scales that are only slightly bigger than the microscopic scale. This property is loosely called the \textit{rigidity} of random matrix eigenvalues.
To make this more precise, let us consider the Circular Unitary Ensemble which consists of $n\times n$ unitary Haar distributed matrices. The eigenvalues of such a random matrix lie on the unit circle in the complex plane, and if we denote the eigenangles as $0 < \theta_1\leq \ldots\leq \theta_n \leq 2\pi$, we can expect that $\theta_j$ will for typical configurations lie close to $\frac{2\pi j}{n}$
because of the rotational invariance of the probability distribution of the eigenvalues. Indeed, it was shown in \cite[Theorem 1.5]{ArguinBeliusBourgade} (see also \cite{PaquetteZeitouni})
that
\[\lim_{n\to\infty}\mathbb P_{\rm CUE}\left((2-\epsilon)\frac{\log n}{n}<\max_{j=1,\ldots, n}\left|\theta_j-\frac{2\pi j}{n}\right|<(2+\epsilon)\frac{\log n}{n}\right)=1\]
for any $\epsilon>0$. We call this an \textit{optimal global rigidity} result because the lower and upper bounds of the maximal eigenvalue deviation differ only by a multiplicative factor which can be chosen arbitrarily close to $1$.
Similar optimal global rigidity results have been obtained in circular $\beta$-ensembles \cite{ChhaibiMadauleNajnudel, Lambert}, in unitary invariant random matrix ensembles \cite{CFLW}, and also for the sine $\beta$-process \cite{HolcombPaquette, Lambert}. In the two-dimensional setting of the Ginibre ensemble, results of a similar nature were obtained in \cite{Lambert2}.

One of the most important features of random matrix eigenvalues is their universal nature: their asymptotic behavior on microscopic scales is similar for large classes of random matrix models.
For instance, in many matrix models of Hermitian $n\times n$ matrices, like the GUE, Wigner matrices, and unitary invariant matrices, the microscopic large $n$ behavior of \textit{bulk} eigenvalues is described by the \textit{sine point process} (see e.g.\ \cite{ErdosYau} and references therein), whereas the microscopic behavior of \textit{edge} eigenvalues is described by the \textit{Airy point process} \cite{DeiftGioev, Deiftetal, BourgadeErdosYau, Forrester, PraehoferSpohn, Soshnikov, TracyWidom}. 
For ensembles of positive-definite Hermitian matrices, the situation is somewhat more complicated.
In the Wishart-Laguerre ensemble and its unitary invariant generalizations, the \textit{Bessel point process} typically describes the microscopic behavior of the smallest eigenvalues close to the hard edge $0$ \cite{Forrester, KuijlaarsVanlessen, TracyWidomBessel}.
However, in a generalization of the Wishart-Laguerre ensemble known as the Muttalib-Borodin Laguerre ensemble \cite{Borodin, Muttalib}, the local behavior of eigenvalues near the hard edge is described by a different determinantal point process known as the \textit{Wright's generalized Bessel point process} \cite{Borodin, KuijMolag}. Another generalization of the Bessel process, known as the \textit{Meijer-$\mathrm{G}$ point process}, arises at the hard edge of Wishart-type products of Ginibre or truncated unitary matrices \cite{AkemannIpsenKieburg, AkemannKieburgWei, KuijlaarsZhang, KuijlaarsStivigny, KieburgKuijlaarsStivigny}, and in Cauchy multi-matrix ensembles \cite{BertolaBothner, BertolaCauchy2MM}. 

\medskip

In this paper, we will establish upper bounds for the global rigidity of the Airy point process, the Bessel point process, and its (determinantal) generalizations arising near the hard edge in Muttalib-Borodin ensembles and in product random matrix ensembles. We do this by combining asymptotics for exponential moments of the counting measures of these point processes, which are Fredholm determinants of certain integral kernel operators, with a global rigidity estimate which can be applied to general point processes which almost surely have a smallest (or largest) point. 
In the case of the Airy and Bessel  point processes, asymptotics for the exponential moments are known, see \cite{BothnerBuckingham} for Airy and \cite{Charlier} for Bessel, and they allow us to improve on the best known upper bounds for the Airy point process (see \cite{Zhong} and \cite[Theorem 1.6]{CorwinGhosal}), and to deduce completely new global rigidity results for the Bessel point process.

\medskip

Another main contribution of this paper consists of exponential moment asymptotics for Wright's generalized Bessel and Meijer-$G$ point processes. We emphasize that we explicitly compute the multiplicative constant in these asymptotics, which is in general very challenging, see e.g.\ \cite{Krasovsky,CLM2019,CLM2019G}. As consequences of the exponential moment asymptotics, we obtain asymptotics for the average and variance of the counting functions of these processes, and an upper bound for their global rigidity.

\paragraph{General rigidity theorem.} Suppose that $X$ is a locally finite random point process on the real line which has almost surely a smallest particle, and denote the ordered random points in the process by $x_1\leq x_2\leq \cdots$. We write $N(s)$ for the random variable that counts the number of points $\leq s$. We will work under the following assumptions, which, as we will see later, are fairly easy to verify in practice.
\begin{assumptions}\label{assumptions}
There exist constants $\mathrm{C}, a >0$, $s_0\in\mathbb R$, $M > \sqrt{2/a}$ and continuous functions $\mu,\sigma:[s_0,+\infty)\to [0,+\infty)$ such that the following holds:
{\begin{enumerate}
\item[(1)] We have
\begin{equation}\label{expmomentbound}
\mathbb{E} \big[e^{-\gamma N(s)}\big]\leq \mathrm{C} \, e^{-\gamma \mu(s)+\frac{\gamma^{2}}{2}\sigma^2(s)},
\end{equation}
for any $\gamma\in[-M,M]$ and for any $s>s_0$.
\item[(2)] The functions $\mu$ and $\sigma$ are strictly increasing and differentiable, and they satisfy 
\begin{align*}
\lim_{s\to + \infty} \mu(s) = + \infty \qquad \mbox{ and } \qquad \lim_{s\to + \infty} \sigma(s) = + \infty.
\end{align*}
Moreover, $s\mapsto s\mu'(s)$ is weakly increasing and $\displaystyle \lim_{s\to+\infty}\frac{s\mu'(s)}{\sigma^2(s)}=+\infty$.
\item[(3)] The function $\sigma^2\circ\mu^{-1}:[\mu(s_0),+\infty)\to [0,+\infty)$ is strictly concave, and
\begin{align*}
& (\sigma^2\circ\mu^{-1})(s)\sim a\log s \qquad \mbox{as} \quad s\to +\infty.
\end{align*}
\end{enumerate}
}
\end{assumptions}
In the above assumptions, $\mathrm{C}$ and $s_0$ are auxiliary constants whose values are unimportant, but on the other hand $a, \mu, \sigma$ will turn out to encode information about fundamental quantities of the point process under consideration, like the mean and variance of the counting functions.

\begin{theorem}\label{thm:rigidity} {\bf (Rigidity)}
Suppose that $X$ is a locally finite point process on the real line which almost surely has a smallest particle and which is such that Assumptions \ref{assumptions} hold. Let us write $x_k$ for the $k$-th smallest particle of the process $X$, $k \geq 1$. Then, there are constants $c>0$ and $s_{0}>0$ such that for any small enough $\epsilon>0$ and for all $s \geq s_{0}$, 
\begin{align}\label{estimate of main thm}
\mathbb P\left( \sup_{k \geq \mu(2s)}\frac{|\mu(x_k)-k|}{\sigma^2(\mu^{-1}(k))} >  \sqrt{\frac{2}{a}(1+\epsilon)} \right)\leq \frac{c \, \mu(s)^{-\frac{\epsilon}{2}}}{\epsilon}.
\end{align}
In particular, for any $\epsilon>0$,
\begin{align*}
\lim_{k_0\to\infty}\mathbb P\left( \sup_{k \geq k_{0}}\frac{|\mu(x_k)-k|}{\sigma^2(\mu^{-1}(k))}\leq \sqrt{\frac{2}{a}(1+\epsilon)} \right)=1.
\end{align*}
\end{theorem}

\begin{remark}\label{rem:close to optimal}
The above result derives an upper bound for the global rigidity via the asymptotics for the first exponential moment of the counting function. Estimates for the first exponential moment however do not allow to obtain a sharp lower bound for the global rigidity. For this, one would need more delicate information, like estimates for higher exponential moments, about more complicated averages in the point process, or about convergence of the counting function to a {Gaussian multiplicative chaos} measure, see e.g.\ \cite{ArguinBeliusBourgade, BWW, CFLW, LOS18}. In the point processes arising in random matrix theory for which optimal lower bounds for the global rigidity are available, see e.g.\ \cite{ArguinBeliusBourgade, ChhaibiMadauleNajnudel, CFLW, HolcombPaquette, PaquetteZeitouni}, it turns out that the upper bounds obtained via the first exponential moment are sharp, and therefore we believe that the upper bound in Theorem \ref{thm:rigidity} is, at least for the concrete examples considered below related to random matrix theory, close to optimal.
\end{remark}

\paragraph{Outline of the proof of Theorem \ref{thm:rigidity}.}
We will prove Theorem \ref{thm:rigidity} in Section \ref{Section: rigidity} using elementary probabilistic estimates.
The most delicate step in the proof consists of establishing a probabilistic bound for the supremum of the normalized counting function of the point process under consideration. For this, we need to use a discretization argument, a union bound, and Markov's inequality together with the exponential moment asymptotics from Assumptions \ref{assumptions}. Next, we prove that the bound on the supremum of the normalized counting function implies rigidity of the points, and we quantify 
the relevant probabilities to obtain Theorem \ref{thm:rigidity}. This method is similar to that of \cite[Section 4]{CFLW}.

\paragraph{Global rigidity for the Airy point process.}
The Airy point process is a determinantal point process on $\mathbb{R}$ whose correlation kernel is given by
\begin{align}\label{Airykernel}
\mathbb{K}^\Ai(x, y) = \frac{ \Ai(x) \Ai'(y) - \Ai'(x) \Ai(y) }{ x - y }, \qquad x,y \in \mathbb{R},
\end{align}
where $\Ai$ denotes the Airy function. This point process describes the largest eigenvalues in a large class of random matrix ensembles, and it has almost surely a largest particle. Upper bounds for the fluctuations of the points have been obtained recently in \cite{Zhong} and \cite[Theorem 1.6]{CorwinGhosal}. A sharper upper bound can be obtained by combining the exponential moment estimates from \cite{BothnerBuckingham} with Theorem \ref{thm:rigidity}.

\vspace{0.2cm}
The Airy point process satisfies Assumptions \ref{assumptions} only after considering the opposite points $x_j=-a_j$, where $a_1> a_2> \cdots$ are the random points in the Airy point process. We write $N^{\rm Ai}(s)$ for the number of points $x_j$ smaller than or equal to $s$.
It was proved in \cite{BothnerBuckingham} (see also \cite{BogClaeysIts, CharlierClaeys})
that
\begin{align*}
\mathbb E\big[e^{-2\pi\nu N^{\rm Ai}(s)}\big]=8^{\nu^{2}}G(1+i\nu)G(1-i\nu)e^{-2\pi\nu\mu(s)+2\pi^2\nu^2\sigma^2(s)}(1+\bigO(s^{-3/2}))
\end{align*}
as $s\to +\infty$ uniformly for $\nu$ in compact subsets of $\mathbb{R}$, where $G$ is Barnes' $G$ function, and where
\begin{align}\label{mu sigma2 airy}
\mu(s)=\frac{2}{3\pi}s^{3/2},\qquad \sigma^2(s)=\frac{3}{4\pi^2}\log s.
\end{align}
It is straightforward to verify from this that the Airy point process satisfies Assumptions \ref{assumptions} with 
\begin{align*}
M=10, \quad \gamma = 2\pi \nu, \quad \mathrm{C} = 2 \max_{\nu \in [-\frac{M}{2\pi},\frac{M}{2\pi}]} 8^{\nu^{2}}G(1+i\nu)G(1-i\nu), \quad a=\frac{1}{2\pi^2},
\end{align*}
and with $s_{0}$ a sufficiently large constant. Applying Theorem \ref{thm:rigidity}, we obtain the following result.

\begin{theorem}\label{thm:rigidity Airy}{\bf (Rigidity for the Airy process)}
Let $-x_1>-x_2>\ldots$ be the points in the Airy point process. There exists a constant $c>0$ such that
\begin{align*}
\mathbb P\left( \sup_{k\geq \mu(s)} \frac{|\frac{2}{3\pi}x_k^{3/2}-k |}{\log k} > \frac{\sqrt{1+\epsilon}}{\pi}\right) \leq \frac{c \, s^{-\frac{3 \epsilon}{4}}}{\epsilon},
\end{align*}
as $s\to +\infty$, uniformly for $\epsilon>0$ small. In particular, for any $\epsilon > 0$ we have
\begin{align*}
\lim_{k_0\to\infty}\mathbb P\left( \sup_{k \geq k_{0}} \frac{|\frac{2}{3\pi}x_k^{3/2}-k|}{\log k}\leq \frac{1}{\pi} + \epsilon\right)=1.
\end{align*}
\end{theorem}
\begin{remark}
This result implies that for any $\epsilon > 0$, the probability that
\begin{align}\label{Airy inequalities}
\left( \frac{3\pi}{2} \right)^{2/3} \left( k- \Big( \frac{1}{\pi}+\epsilon \Big)\log k \right)^{2/3} \leq x_{k} \leq \left( \frac{3\pi}{2} \right)^{2/3} \left( k+ \Big( \frac{1}{\pi}+\epsilon \Big)\log k \right)^{2/3} \quad \mbox{for all } k \geq k_{0}
\end{align}
tends to $1$ as $k_{0} \to +\infty$. Figure \ref{fig:AiryBessel rigidity} illustrates this and supports our belief that Theorem \ref{thm:rigidity Airy} is close to optimal (see also Remark \ref{rem:close to optimal}).
\end{remark}
\begin{figure}
\begin{center}
\begin{tikzpicture}
\node at (0,0) {\includegraphics[scale=0.3]{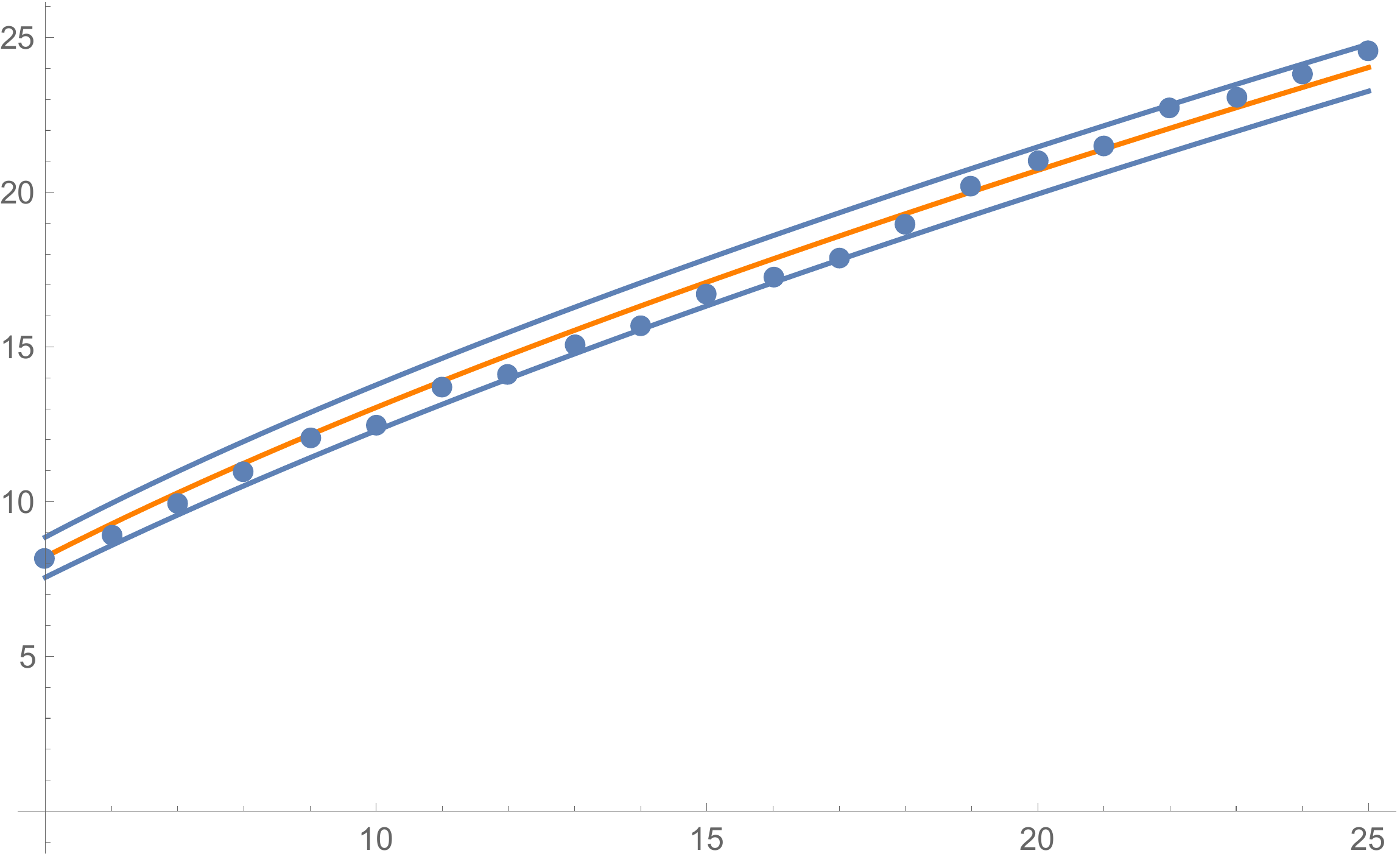}};
\node at (0,-2.47) {Airy};
\end{tikzpicture}
\begin{tikzpicture}
\node at (0,0) {\includegraphics[scale=0.35]{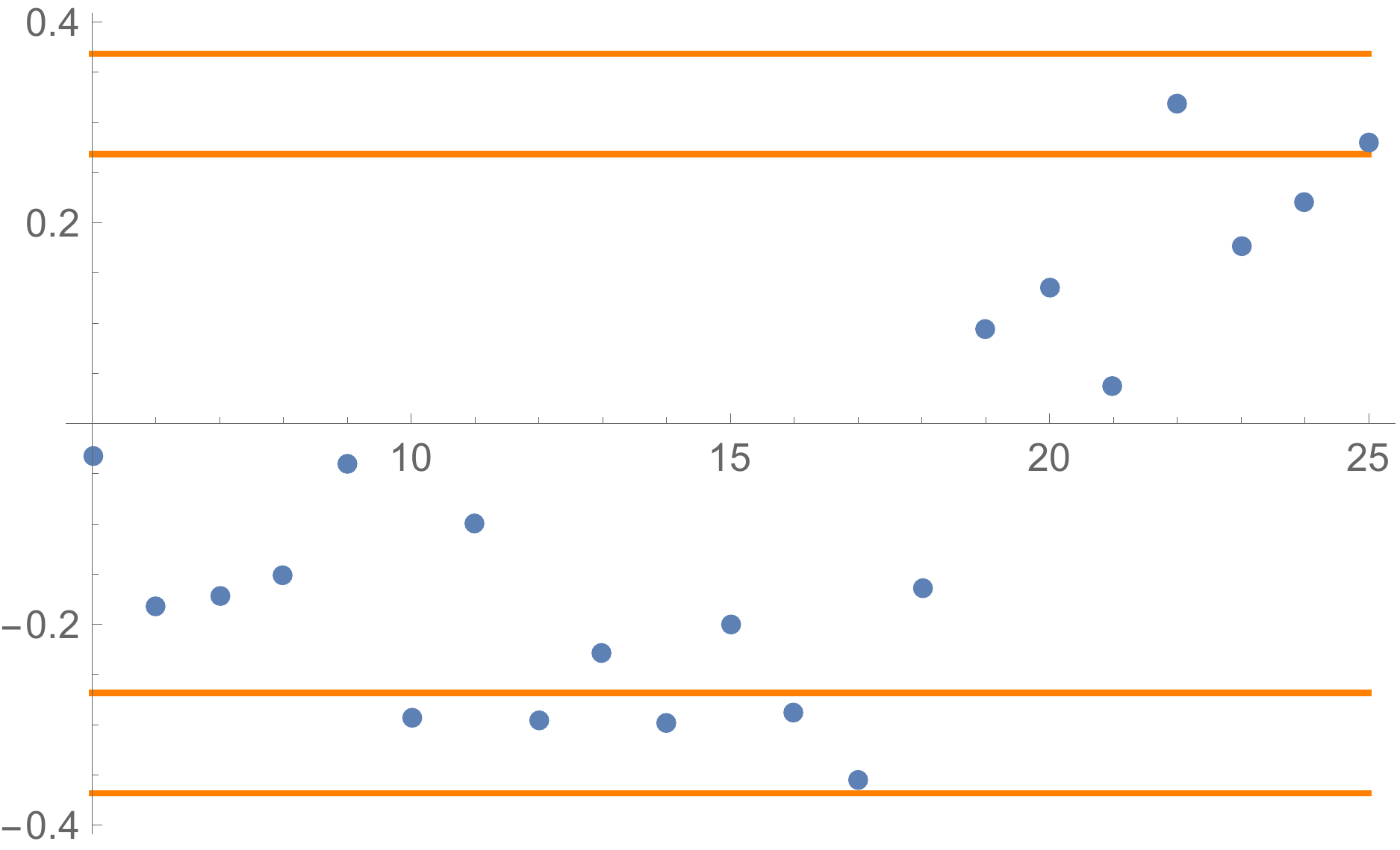}};
\node at (0,-2.5) {Normalized Airy};
\end{tikzpicture}
\end{center}
\caption{\label{fig:AiryBessel rigidity}\textit{Global ridigity for the Airy point process.
At the left, the blue dots represent the random points and have coordinates $(k,x_{k})$, the blue curves are the upper and lower bounds in \eqref{Airy inequalities} (with $\epsilon=0.05$), and the orange curve is parametrized by $\big(t,(\frac{3\pi}{2}t)^{2/3}\big)$. At the right, the blue dots represent the normalized random points with coordinates
$\Big(k,\frac{\frac{2}{3\pi}x_k^{3/2}-k}{\log k}\Big)$. The orange lines indicate the heights $\pm \frac{1}{\pi}\pm\epsilon$ with $\epsilon=0.05$. We observe the presence of points in the bands between the orange lines, indicating that Theorem \ref{thm:rigidity Airy} can be expected to be sharp. The points shown in the figure are not exactly the points in the Airy point process: they are sampled as re-scaled extreme eigenvalues of a large GUE matrix, which approximate the points in the Airy point process.
} }
\end{figure}

\paragraph{Global rigidity for the Bessel point process.}
The Bessel point process is another canonical point process from the theory of random matrices. It models the behavior of the eigenvalues near hard edges in a large class of random matrix ensembles, with the Laguerre-Wishart ensemble as most prominent example \cite{Forrester}. The Bessel point process is a determinantal point process on $(0,+\infty)$ whose correlation kernel is given by
\begin{align}\label{Besselkernel}
\mathbb K^{\rm Be}_\alpha(x,y)=\frac{\sqrt{y}J_\alpha(\sqrt{x})J_\alpha'(\sqrt{y})-\sqrt{x}J_\alpha(\sqrt{y})
J_\alpha'(\sqrt{x})}{2(x-y)},  \qquad x,y > 0,
\end{align}
where $\alpha>-1$ and $J_\alpha$ is the Bessel function of the first kind of order $\alpha$. To the best of our knowledge, there are no global rigidity upper bounds available in the literature for the Bessel process, but the corresponding exponential moment asymptotics have been obtained in \cite{Charlier}, and they allow us to apply Theorem \ref{thm:rigidity}. Let us write $N^{\rm Be}(s)$ for the number of points $x_j$ smaller than or equal to $s$ in the Bessel process. By \cite[eq (1.11)-(1.12)]{Charlier}, we have
\begin{align}\label{moment asymp Bessel}
\mathbb{E}\big[e^{-2\pi \nu N^{\mathrm{Be}}(s)}\big] = 4^{\nu^{2}}e^{\pi \nu \alpha} G(1+i\nu)G(1-i\nu) e^{-2\pi \nu \mu(s) + 2 \pi^{2} \nu^{2} \sigma^{2}(s)}(1+\bigO(s^{-1/2}\log s)),
\end{align}
as $s \to + \infty$ uniformly for $\nu$ in compact subsets of $\mathbb{R}$, with
\begin{align}\label{mu sigma2 Bessel}
\mu(s) = \frac{\sqrt{s}}{\pi}, \qquad \sigma^{2}(s) = \frac{1}{4\pi^{2}}\log s.
\end{align}
We verify from \eqref{moment asymp Bessel} that the Bessel point process satisfies Assumptions \ref{assumptions} with 
\begin{align*}
M=10, \quad \gamma = 2\pi \nu, \quad \mathrm{C} = 2 \max_{\nu \in [-\frac{M}{2\pi},\frac{M}{2\pi}]} 4^{\nu^{2}}e^{\pi \nu \alpha}G(1+i\nu)G(1-i\nu), \quad a=\frac{1}{2\pi^2},
\end{align*}
and with $s_{0}$ a sufficiently large constant. Applying Theorem \ref{thm:rigidity}, we obtain the following result.

\begin{theorem}\label{thm:rigidity Bessel}{\bf (Rigidity for the Bessel point process)}
Let $x_1<x_2<\ldots$ be the points in the Bessel point process.
There exists a constant $c>0$ such that
\begin{align*}
\mathbb P\left( \sup_{k\geq \mu(s)} \frac{|\frac{1}{\pi}x_k^{1/2}-k |}{\log k} > \frac{\sqrt{1+\epsilon}}{\pi}\right) \leq \frac{c \, s^{-\frac{\epsilon}{4}}}{\epsilon},
\end{align*}
as $s\to +\infty$, uniformly for $\epsilon>0$ small. In particular, for any $\epsilon > 0$ we have
\begin{align*}
\lim_{k_0\to\infty}\mathbb P\left( \sup_{k \geq k_{0}} \frac{|\frac{1}{\pi}x_k^{1/2}-k|}{\log k}\leq \frac{1}{\pi} + \epsilon\right)=1.
\end{align*}
\end{theorem}
\begin{remark}
The above implies that for any $\epsilon > 0$, the probability that
\begin{align}\label{Bessel inequalities}
\pi^{2} \left( k- \Big( \frac{1}{\pi}+\epsilon \Big)\log k \right)^{2} \leq x_{k} \leq \pi^{2} \left( k+ \Big( \frac{1}{\pi}+\epsilon \Big)\log k \right)^{2} \quad \mbox{for all } k \geq k_{0}
\end{align}
tends to $1$ as $k_{0} \to +\infty$. Figure \ref{fig:Bessel rigidity} illustrates this. 
\end{remark}
\begin{figure}
\begin{center}
\begin{tikzpicture}
\node at (0,0) {\includegraphics[scale=0.3]{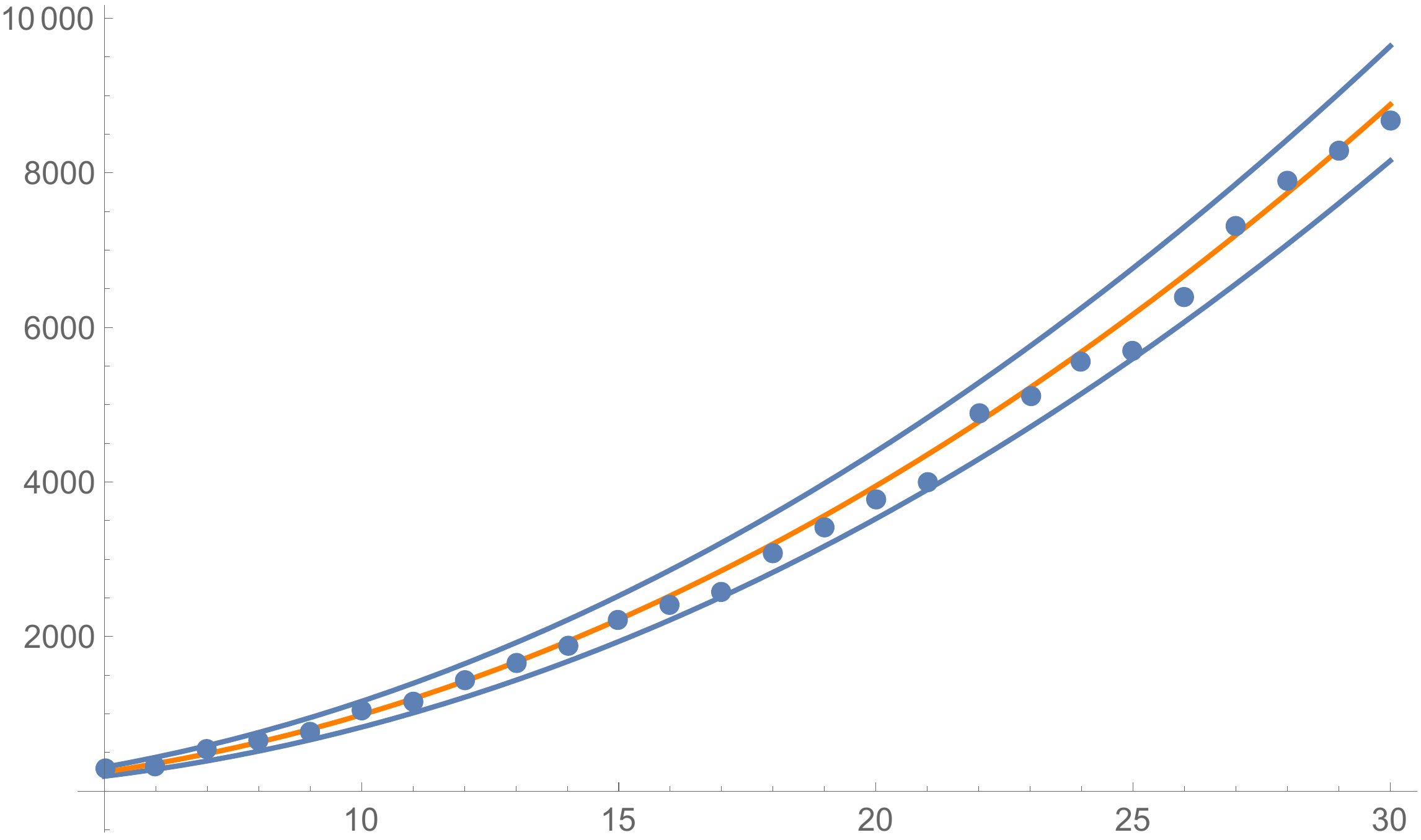}};
\node at (0,-2.47) {Bessel};
\end{tikzpicture}
\begin{tikzpicture}
\node at (0,0) {\includegraphics[scale=0.35]{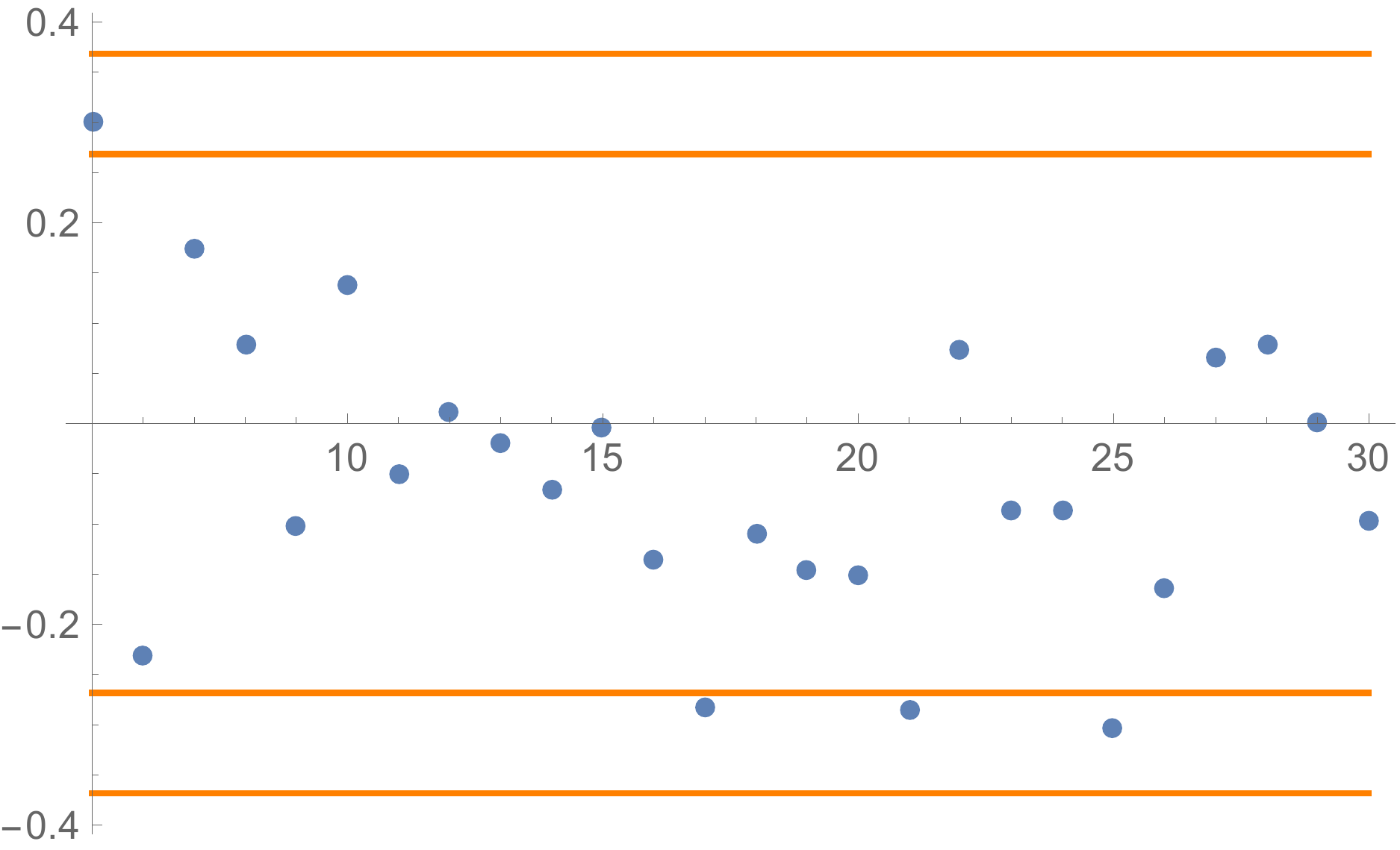}};
\node at (0,-2.5) {Normalized Bessel};
\end{tikzpicture}
\end{center}
\caption{\label{fig:Bessel rigidity}\textit{Global ridigity for the Bessel point process.
At the left, the blue dots represent the random points and have coordinates $(k,x_{k})$, the blue curves are the upper and lower bounds in \eqref{Bessel inequalities} (with $\epsilon=0.05$), and the orange curve is parametrized by $(t,\pi^2t^2)$. At the right, the blue dots represent the normalized random points with coordinates
$\Big(k,\frac{\frac{1}{\pi}x_k^{1/2}-k}{\log k}\Big)$. The orange lines indicate the heights $\pm \frac{1}{\pi}\pm\epsilon$ with $\epsilon=0.05$. We observe the presence of points in the bands between the orange lines, indicating that Theorem \ref{thm:rigidity Bessel} can be expected to be sharp. The points shown in the figure are not exactly the points in the Bessel point process: they are sampled as re-scaled extreme eigenvalues of a large Laguerre/Wishart random matrix, which approximate the points in the Bessel point process.}
} 
\end{figure}

\paragraph{Exponential moments and rigidity for the Wright's generalized Bessel process.}
The Wright's generalized Bessel process appears as the limiting point process at the hard edge of Muttalib-Borodin ensembles \cite{Borodin, ClaeysRomano, ForWang, KuijMolag, LiuWangZhang, Zhang, Zhang2}. This is a determinantal point process on $(0,+\infty)$ which depends on parameters $\theta > 0$ and $\alpha > -1$. The associated kernel is given by
\begin{align}\label{def Wright kernel}
\mathbb K^{\mathrm{Wr}}(x,y)=\theta \, (xy)^{\frac{\alpha}{2}} \int_0^1 J_{\frac{\alpha+1}{\theta},\frac{1}{\theta}}(xt) J_{\alpha+1,\theta}\big((yt)^\theta  \big)t^\alpha dt, \qquad x,y > 0,
\end{align}
where $J_{\alpha,\theta}$ is Wright's generalized Bessel function,
\begin{align*}
J_{a,b}(x)=\sum_{m=0}^\infty \frac{(-x)^m}{m!\Gamma(a+bm)}.
\end{align*}
If $\theta = 1$, this point process reduces (up to a rescaling) to the Bessel point process:
\begin{align}\label{Bessel and Wright kernel}
\mathbb K^{\mathrm{Wr}}(x,y) \Big|_{\theta = 1} = 4 \, \mathbb{K}^{\mathrm{Be}}(4x,4y), \qquad x,y > 0.
\end{align}
We obtain asymptotics for the exponential moments in this point process.
\begin{theorem}\label{thm:expmoments Wrights}
Let $\nu\in\mathbb R$ and let $N^{\mathrm{Wr}}(s)$ denote the number of points smaller than or equal to $s$ in the Wright's generalized Bessel process. As $s \to + \infty$, we have
\begin{align}\label{asymp gap thm explicit  Wr}
\mathbb{E} \big[ e^{-2\pi\nu N^{\mathrm{Wr}}(s)} \big]= C \, \exp \bigg(-2\pi\nu \mu(s) + 2 \pi^{2} \nu^{2} \sigma^{2}(s) + \bigO \big(s^{-\frac{\theta}{1+\theta}}\big) \bigg),
\end{align}
where the functions $\mu$ and $\sigma^{2}$ are given by
\begin{align}\label{mu and sigma2 wrights}
\mu(s) = \frac{1+\theta}{\pi} \, \theta^{- \frac{\theta}{1+\theta}} \cos \left( \frac{\pi}{2}\frac{1-\theta}{1+\theta} \right)s^{\frac{\theta}{1+\theta}}, \qquad \mbox{ and } \qquad  \sigma^{2}(s) = \frac{\theta}{2\pi^{2}(1+\theta)} \log s,
\end{align}
and the values of $C$ by
\begin{align}\label{constant for Wrights}
& C = \exp \bigg( \frac{\pi \nu(1-\theta + 2 \alpha)}{1+\theta} \bigg) \bigg[ 4 (1+\theta)\, \theta^{- \frac{\theta}{1+\theta}} \sin^{2} \left( \frac{\pi \, \theta}{1+\theta} \right) \bigg]^{\nu^{2}}G(1+i\nu)G(1-i\nu),
\end{align}
where $G$ is Barnes' $G$-function. Furthermore, the error term in \eqref{asymp gap thm explicit  Wr} is uniform for $\nu$ in compact subsets of $\mathbb R$.
\end{theorem}
\begin{remark}
By setting $\theta = 1$ in \eqref{asymp gap thm explicit Wr} and then applying the rescaling $s \mapsto \frac{s}{4}$, we recover the asymptotics \eqref{moment asymp Bessel} for the Bessel point process, which is consistent with \eqref{Bessel and Wright kernel}. In fact, we even slightly improved the error term: from \eqref{asymp gap thm explicit Wr} with $\theta = 1$, it follows that the error term $\bigO(s^{-1/2}\log s)$ in \eqref{moment asymp Bessel} is $\bigO(s^{-1/2})$.
\end{remark}
\begin{remark}
It follows from \cite[page 4]{Borodin} that the left-hand side of \eqref{asymp gap thm explicit  Wr} is invariant under the following changes of the parameters:
\begin{align}\label{symmetry}
s \mapsto s^{\theta}, \qquad \theta \mapsto \frac{1}{\theta}, \quad \mbox{ and } \quad \alpha \mapsto \alpha^{\star}:=\frac{1+\alpha}{\theta}-1.
\end{align} 
It follows that the constant $C$ and the functions $\mu$ and $\sigma^{2}$ must obey the following symmetry relations for any $\theta > 0$ and $\alpha > -1$:
\begin{align}\nonumber
& 
\mu(s,\theta,\alpha) = \mu(s^{\theta},\tfrac{1}{\theta},\alpha^{\star}), \quad 
\sigma^{2}(s,\theta,\alpha) = \sigma^{2}(s^{\theta},\tfrac{1}{\theta},\alpha^{\star}), \quad C(\theta,\alpha) = C(\tfrac{1}{\theta},\alpha^{\star}),
\end{align}
where we have indicated the dependence of the quantities on $\theta$ and $\alpha$ explicitly. These identities can be verified directly from \eqref{constant for Wrights} and provide a consistency check of our results.
\end{remark}
\begin{remark}\label{remark:existence}
It is not entirely obvious that the kernel \eqref{def Wright kernel} defines a point process. 
To see this, we note first that the kernel \eqref{def Wright kernel} arises as the large $n$ limit of the correlation kernel $K_n$ in the Muttalib-Borodin Laguerre ensemble with $n$ particles (see \cite{Borodin}). Next, from \cite{Lenard} and \cite[Theorem 1]{Soshnikov}, we know that a kernel defines a point process if and only if it generates locally integrable correlation functions which are symmetric under permutations of variables and satisfy a certain positivity condition. Since $K_n$ must satisfy the symmetry and positivity conditions, and since these conditions are closed under taking limits, we can conclude that \eqref{def Wright kernel} also defines a point process. The uniqueness of the point process follows from the fact that the process is characterized by its Laplace transform $\mathbb E e^{-\sum_{k=1}^\infty f(x_k)}$ for continuous compactly supported functions $f$, where $x_1, x_2, \ldots$ are the points in the process.
For a determinantal point process with a kernel $K$ which is trace-class on any compact, the Laplace transform is characterized by $K$ since
\[\mathbb E e^{-\sum_{k=1}^\infty f(x_k)} = \det\left(1-(1-e^{-f})K\right).\]
The Fredholm determinant at the right is defined since the trace-norm of $(1-e^{-f})K$ is bounded by $\|1-e^{-f}\|_\infty$  times the trace-norm of $K$ restricted to the support of $f$. Hence the process defined by $K$ is unique.
\end{remark}
Theorem \ref{thm:expmoments Wrights} has the following immediate consequence.
\begin{corollary}\label{coro:counting function asymptotics Wrights}
As $s \to + \infty$, we have
\begin{align}
& \mathbb{E}[N^{\mathrm{Wr}}(s)] = \mu(s)- 
\frac{1-\theta + 2 \alpha}{2(1+\theta)}+ \bigO(s^{-\frac{\theta}{1+\theta}}), \label{asymp E N Wr} \\
& {\rm Var}[N^{\mathrm{Wr}}(s)] = \sigma^{2}(s) + 
\frac{1}{2\pi^{2}}\log \left[4 (1+\theta)\, \theta^{- \frac{\theta}{1+\theta}} \sin^{2} \left( \frac{\pi \, \theta}{1+\theta} \right) \right] +
\frac{1+\gamma_E}{2\pi^{2}}+
 \bigO(s^{-\frac{\theta}{1+\theta}}), \label{asymp Var N Wr}
\end{align}
where $\gamma_E\approx 0.5772$ is Euler's constant and the functions $\mu$ and $\sigma^{2}$ are given by \eqref{mu and sigma2 wrights}.
\end{corollary}
\begin{proof}
The asymptotics \eqref{asymp gap thm explicit Wr} are valid uniformly for $\nu$ in compact subsets of $\mathbb{R}$. Hence, we obtain \eqref{asymp E N Wr}--\eqref{asymp Var N Wr} by first expanding \eqref{asymp gap thm explicit Wr} as $\nu \to 0$, and then identifying the power of $\nu$ using
\begin{align*}
\mathbb{E} \big[ e^{-2\pi\nu N^{\mathrm{Wr}}(s)} \big] = 1 - 2 \pi \nu \mathbb{E}[N^{\mathrm{Wr}}(s)] + 2 \pi^{2} \nu^{2} \mathbb{E}[(N^{\mathrm{Wr}}(s))^{2}] + \bigO(\nu^{3}), \qquad \mbox{as } \nu \to 0.
\end{align*}
\end{proof}
\begin{remark}
Setting $\theta = 1$ in \eqref{asymp E N Wr}--\eqref{asymp Var N Wr} and then rescaling $s \mapsto \frac{s}{4}$, we recover the asymptotic formulas \cite[eq (1.14)]{Charlier} with improved error terms.
\end{remark}
We verify from Theorem \ref{thm:expmoments Wrights} that the Wright's generalized Bessel process satisfies Assumptions \ref{assumptions} with $s_{0}$ a sufficiently large constant, and
\begin{align*}
M=10, \quad \gamma = 2\pi \nu, \quad \mathrm{C} = 2 \max_{\nu \in [-\frac{M}{2\pi},\frac{M}{2\pi}]} C(\nu), \quad a=\frac{1}{2\pi^2},
\end{align*}
where $C = C(\nu)$ is given by \eqref{constant for Wrights}. We obtain the following global rigidity result by combining Theorem \ref{thm:rigidity} with Theorem \ref{thm:expmoments Wrights}.

\begin{theorem}{\bf (Rigidity for the Wright's generalized Bessel process)}\label{thm:rigidityWright}
Let $x_1<x_2<\ldots$ be the points in the Wright's generalized Bessel point process. There exists a constant $c>0$ such that
\begin{align*}
\mathbb P\left( \sup_{k\geq \mu(s)} \frac{\big|\frac{1+\theta}{\pi} \, \theta^{- \frac{\theta}{1+\theta}} \cos \left( \frac{\pi}{2}\frac{1-\theta}{1+\theta} \right)x_{k}^{\frac{\theta}{1+\theta}}-k \big|}{\log k} > \frac{\sqrt{1+\epsilon}}{\pi}\right) \leq \frac{c \, s^{-\frac{\theta \, \epsilon}{2(1+\theta)}}}{\epsilon},
\end{align*}
as $s\to +\infty$, uniformly for $\epsilon>0$ small. In particular, for any $\epsilon > 0$ we have
\begin{align*}
\lim_{k_0\to\infty}\mathbb P\left( \sup_{k \geq k_{0}} \frac{\big|\frac{1+\theta}{\pi} \, \theta^{- \frac{\theta}{1+\theta}} \cos \left( \frac{\pi}{2}\frac{1-\theta}{1+\theta} \right)x_{k}^{\frac{\theta}{1+\theta}}-k \big|}{\log k}\leq \frac{1}{\pi} + \epsilon\right)=1.
\end{align*}
\end{theorem}
\begin{remark}
This result implies that for any $\epsilon > 0$, the probability that
\begin{align}\label{eq:boundsMB}
\left[ \frac{\pi}{1+\theta}\frac{\theta^{\frac{\theta}{1+\theta}}}{\cos ( \frac{\pi}{2}\frac{1-\theta}{1+\theta} )}  \left( k- \Big( \frac{1}{\pi}+\epsilon \Big)\log k \right)\right]^{\frac{1+\theta}{\theta}} \leq x_{k} \leq \left[ \frac{\pi}{1+\theta}\frac{\theta^{\frac{\theta}{1+\theta}}}{\cos ( \frac{\pi}{2}\frac{1-\theta}{1+\theta} )}  \left( k+ \Big( \frac{1}{\pi}+\epsilon \Big)\log k \right)\right]^{\frac{1+\theta}{\theta}}
\end{align}
for all $k \geq k_{0}$, tends to $1$ as $k_{0} \to +\infty$. We verify this numerically in Figure \ref{fig:WrightMeigerG rigidity} (left) for $\epsilon = 0.05$ and different values of $\theta$. 
\end{remark}
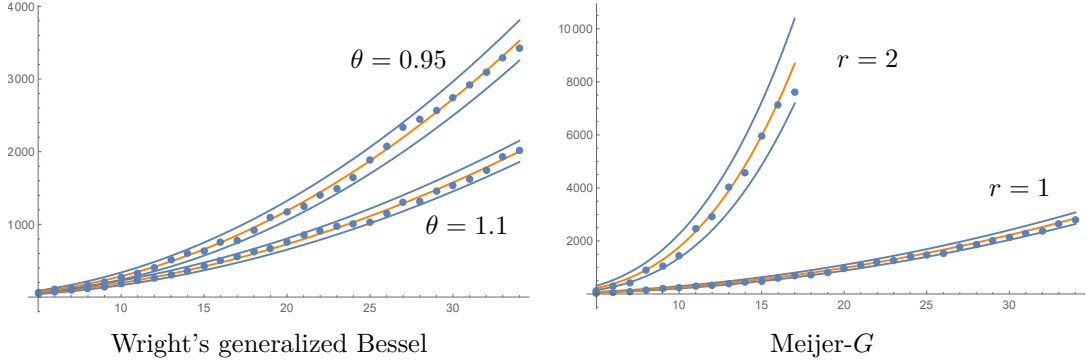
\begin{figure}
\begin{center}
\begin{tikzpicture}
\node at (0,0) {\includegraphics[scale=0.3]{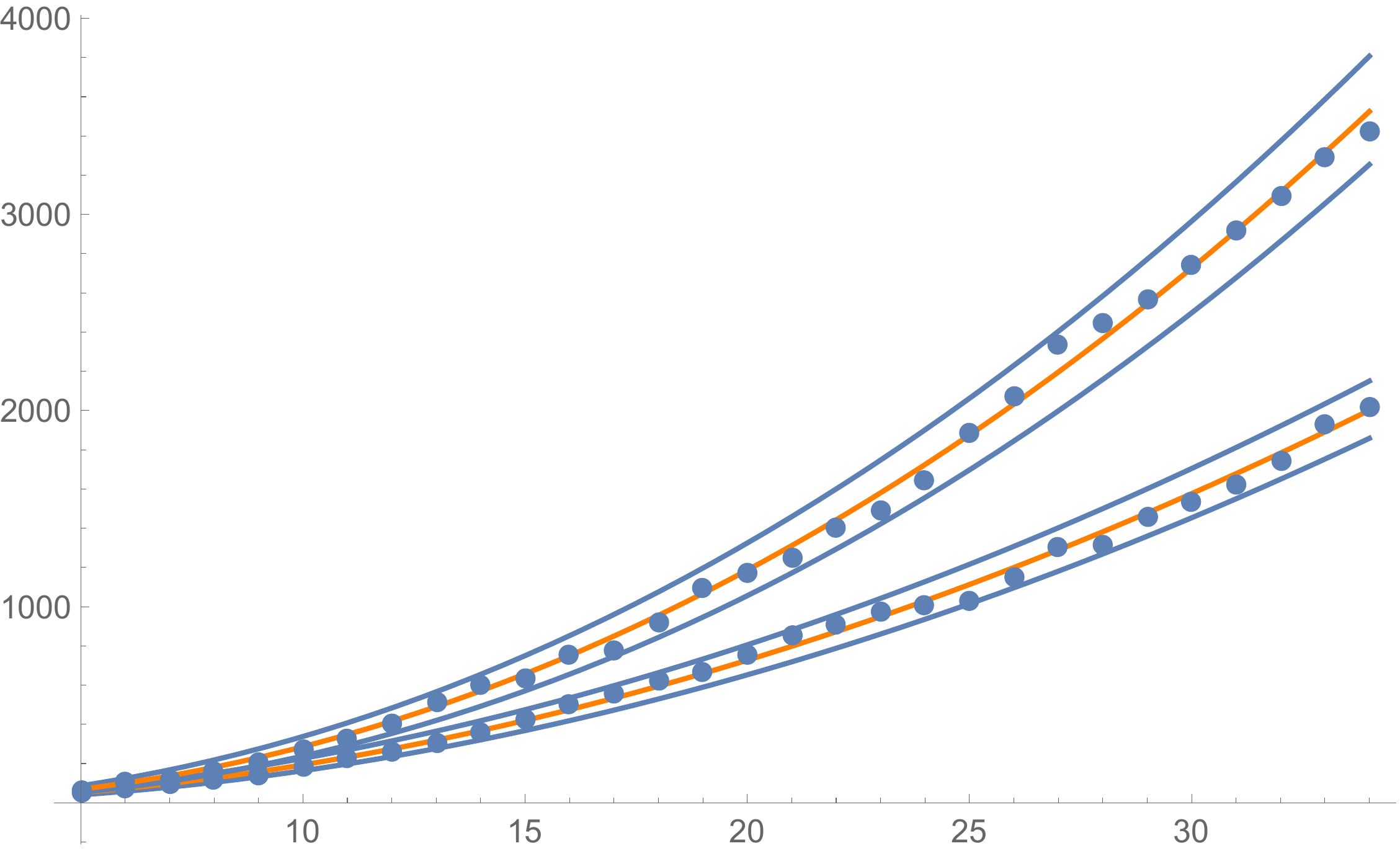}};
\node at (0,-2.5) {Wright's generalized Bessel};
\node at (1.7,1.3) {$\theta=0.95$};
\node at (2.6,-0.9) {$\theta=1.1$};
\end{tikzpicture}
\begin{tikzpicture}
\node at (0,0) {\includegraphics[scale=0.3]{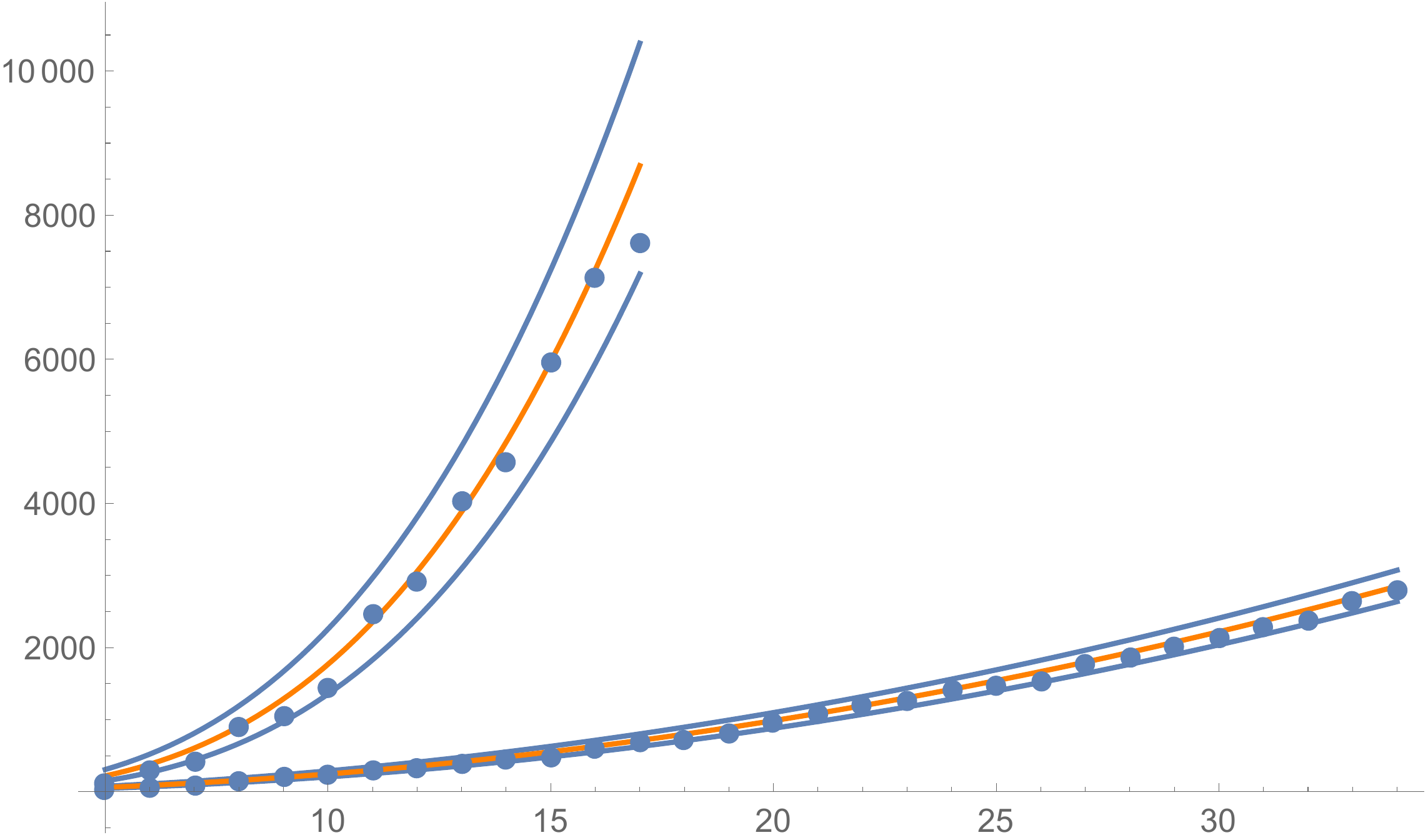}};
\node at (0.6,1.3) {$r=2$};
\node at (2.6,-0.4) {$r=1$};
\node at (0,-2.5) {Meijer-$G$};
\end{tikzpicture}
\end{center}
\caption{\label{fig:WrightMeigerG rigidity}\textit{Global ridigity for the Wright's generalized Bessel and Meijer-$G$ point processes. The blue dots represent the random points for different values of the parameters, the blue curves represent the upper and lower bounds in \eqref{eq:boundsMB} and \eqref{eq:boundsMG}  with $\epsilon=0.05$, and the orange curves correspond to $\epsilon=-\frac{1}{\pi}$. The points were sampled as random matrix eigenvalues which are known to approximate the random points. At the left, we rely for this on a result of Cheliotis \cite[Theorem 4]{Cheliotis}, and at the right, on a result of Kuijlaars and Zhang \cite[Theorem 5.3]{KuijlaarsZhang}.}} 
\end{figure}
\paragraph{Exponential moments and rigidity for the Meijer-$\mathrm{G}$ process.}
The Meijer-$\mathrm{G}$ process is the limiting point process at the hard edge of Wishart-type products of Ginibre matrices or truncated unitary matrices and appears also in Cauchy multi-matrix models \cite{BertolaBothner, BertolaCauchy2MM, KieburgKuijlaarsStivigny, KuijlaarsZhang}. It is a determinantal point process on $(0,+\infty)$ which depends on parameters $r,q \in \mathbb{N}:=\{0,1,2,\ldots\}$, $r > q \geq 0$, $\nu_{1},\ldots,\nu_{r} \in \mathbb{N}$ and $\mu_{1},\ldots,\mu_{q} \in \mathbb{N}_{>0}$ such that $\mu_{k} > \nu_{k}$, $k=1,\ldots,q$. Its kernel can be expressed in terms of the Meijer-$\mathrm{G}$ function:
\begin{align}\label{def Meijer G kernel}
\mathbb K^{\mathrm{Me}}(x,y)=\int_0^1
\mathrm{G}_{q,r+1}^{1,q}\left( \left. {-\mu_1,\dots, -\mu_q \atop 0,-\nu_1,\dots, -\nu_r} \right|tx\right)
\mathrm{G}_{q,r+1}^{r,0}\left( \left. {\mu_1,\dots, \mu_q \atop \nu_1,\dots, \nu_r,0} \right|ty\right)dt, \qquad x,y > 0.
\end{align}
We obtain exponential moment asymptotics for this point process.
\begin{theorem}\label{thm:expmoments Meijer}
Let $\nu\in\mathbb R$ and let $N^{\mathrm{Me}}$ be the counting function of the Meijer-$\mathrm{G}$ process. As $s \to + \infty$, we have
\begin{align}\label{asymp gap thm explicit  Me}
\mathbb{E} \big[ e^{-2\pi\nu N^{\mathrm{Me}}(s)} \big]= C \, \exp \bigg(-2\pi\nu \mu(s) + 2 \pi^{2} \nu^{2} \sigma^{2}(s) + \bigO \big(s^{-\frac{1}{1+r-q}}\big) \bigg),
\end{align}
where
the functions $\mu$ and $\sigma^{2}$ are given  by
\begin{align}\label{mu and sigma2 meijer}
\mu(s) = \frac{1+r-q}{\pi} \cos \left( \frac{\pi}{2}\frac{r-q-1}{r-q+1} \right)s^{\frac{1}{1+r-q}}, \qquad \mbox{ and } \qquad  \sigma^{2}(s) = \frac{1}{2\pi^{2}(1+r-q)} \log s,
\end{align}
and the values of $C$ by
\begin{align*}
& C = \exp \bigg( \frac{2\pi \nu}{1+r-q} \bigg[ \sum_{j=1}^{r}\nu_{j} - \sum_{j=1}^{q} \mu_{j} \bigg] \bigg) \bigg[ 4(1+r-q)\sin^{2} \left( \frac{\pi}{1+r-q} \right) \bigg]^{\nu^{2}}G(1+i\nu)G(1-i\nu),
\end{align*}
where $G$ is Barnes' $G$-function. Furthermore, the error term in \eqref{asymp gap thm explicit  Wr} is uniform for $\nu$ in compact subsets of $\mathbb R$.
\end{theorem}
\begin{remark}
If $q=0$ and if the parameters $\nu_{1}, \ldots,\nu_{r}$ form an arithmetic progression, then the kernel $\mathbb{K}^{\mathrm{Me}}$ defines the same point process (up to rescaling) as the Wright's generalized Bessel point process (for a rational $\theta$), see \cite[Theorem 5.1]{KuijlaarsStivigny}. More precisely, if $r \geq 1$ is an integer, $\alpha > -1$ and
\begin{align}
& \theta = \frac{1}{r}, & & \nu_{j} = \alpha + \frac{j-1}{r}, \qquad j = 1,\ldots,r, \label{cond parameters for relation first and third kernel}
\end{align}
then the kernels $\mathbb{K}^{\mathrm{Me}}$ and $\mathbb{K}^{\mathrm{Wr}}$ are related by
\begin{align}\label{Mejer G Wright relation kernel}
\left( \frac{x}{y} \right)^{\frac{\alpha}{2}}\mathbb{K}^{\mathrm{Me}}(x,y) = r^{r} \mathbb{K}^{\mathrm{Wr}}(r^{r}x,r^{r}y).
\end{align}
Therefore, if the parameters satisfy \eqref{cond parameters for relation first and third kernel}, we obtain the following relations:
\begin{align}\label{Wright to Meijer identities}
&\mu^{\mathrm{Me}}(s) = \mu^{\mathrm{Wr}}(r^{r}s), \qquad \sigma^{\mathrm{Me}}(s)^{2} = \sigma^{\mathrm{Wr}}(s)^{2}, \qquad C^{\mathrm{Me}} = r^{\frac{r \nu^{2}\theta}{\theta+1}}C^{\mathrm{Wr}},
\end{align}
where the quantities with superscript $\mathrm{Wr}$ and $\mathrm{Me}$ are given in Theorems \ref{thm:expmoments Wrights} and \ref{thm:expmoments Meijer}, respectively. All the identities in \eqref{Wright to Meijer identities} can be verified by a direct computation; this provides another consistency check of our results.
\end{remark}
\begin{remark}\label{remark:existence2}
The existence and uniqueness of a point process with correlation kernel \eqref{def Meijer G kernel} can be shown in a similar way as outlined in Remark \ref{remark:existence} for the Wright's generalized Bessel process.
\end{remark}
\begin{corollary}\label{thm:counting function asymptotics Meijer}
As $s \to + \infty$, we have
\begin{align}
& \mathbb{E}[N^{\mathrm{Me}}(s)] = \mu(s)- 
\frac{1}{1+r-q}\bigg[\sum_{j=1}^{r} \nu_{j} - \sum_{j=1}^{q} \mu_{j} \bigg]+ \bigO(s^{-\frac{1}{1+r-q}}), \label{asymp E N Me} \\
& {\rm Var}[N^{\mathrm{Me}}(s)] = \sigma^{2}(s) + 
\frac{1}{2\pi^{2}}\log \left[4(1+r-q)\sin^{2} \left( \frac{\pi}{1+r-q} \right)\right] +
\frac{1+\gamma_E}{2\pi^{2}}+
 \bigO(s^{-\frac{1}{1+r-q}}), \label{asymp Var N Me}
\end{align}
where $\gamma_E$ is Euler's constant and the functions $\mu$ and $\sigma^{2}$ are given by \eqref{mu and sigma2 meijer}.
\end{corollary}
\begin{proof}
The proof is similar to the proof of Corollary \ref{coro:counting function asymptotics Wrights}.
\end{proof}
We verify from Theorem \ref{thm:expmoments Meijer} that the Meijer-$G$ process satisfies Assumptions \ref{assumptions} with $s_{0}$ a sufficiently large constant, and
\begin{align*}
M=10, \quad \gamma = 2\pi \nu, \quad \mathrm{C} = 2 \max_{\nu \in [-\frac{M}{2\pi},\frac{M}{2\pi}]} C(\nu), \quad a=\frac{1}{2\pi^2},
\end{align*}
where $C = C(\nu)$ is given in Theorem \ref{thm:expmoments Meijer}. We obtain the following rigidity result by combining Theorem \ref{thm:rigidity} with Theorem \ref{thm:expmoments Meijer}.

\begin{theorem}{\bf (Rigidity for the Meijer-$G$ process)}\label{thm:rigidityMeijer}
Let $x_1<x_2<\ldots$ be the points in the Meijer-$G$ point process. There exists a constant $c>0$ such that
\begin{align*}
\mathbb P\left( \sup_{k\geq \mu(s)} \frac{\big|\frac{1+r-q}{\pi} \cos \left( \frac{\pi}{2}\frac{r-q-1}{r-q+1} \right)x_{k}^{\frac{1}{1+r-q}}-k\big|}{\log k} > \frac{\sqrt{1+\epsilon}}{\pi}\right) \leq \frac{c \, s^{-\frac{\epsilon}{2(1+r-q)}}}{\epsilon},
\end{align*}
as $s\to +\infty$, uniformly for $\epsilon>0$ small. In particular, for all $\epsilon > 0$ we have
\begin{align*}
\lim_{k_0\to\infty}\mathbb P\left( \sup_{k \geq k_{0}} \frac{\big|\frac{1+r-q}{\pi} \cos \left( \frac{\pi}{2}\frac{r-q-1}{r-q+1} \right)x_{k}^{\frac{1}{1+r-q}}-k \big|}{\log k}\leq \frac{1}{\pi} + \epsilon\right)=1.
\end{align*}
\end{theorem}
\begin{remark}
The above means that for any $\epsilon > 0$, the probability that
\begin{multline}\label{eq:boundsMG}
\left[ \frac{\pi}{1+r-q}\frac{1}{\cos ( \frac{\pi}{2}\frac{r-q-1}{r-q+1} )}  \left( k- \Big( \frac{1}{\pi}+\epsilon \Big)\log k \right)\right]^{1+r-q}  \leq x_{k} \\
 \leq \left[ \frac{\pi}{1+r-q}\frac{1}{\cos ( \frac{\pi}{2}\frac{r-q-1}{r-q+1} )}  \left( k+ \Big( \frac{1}{\pi}+\epsilon \Big)\log k \right)\right]^{1+r-q} \qquad \mbox{for all } k \geq k_{0}
\end{multline}
tends to $1$ as $k_{0} \to +\infty$. This has been verified numerically for $q=0$ and different values of $r$ by generating products of $r$ independent Ginibre matrices \cite{AkemannKieburgWei,KuijlaarsZhang} and is illustrated in Figure \ref{fig:WrightMeigerG rigidity} (right) for $\epsilon = 0.05$. 
\end{remark}

\paragraph{Outline of the proofs of Theorem \ref{thm:expmoments Wrights} and Theorem \ref{thm:expmoments Meijer}.} It is well-known \cite{BorodinPoint, Johansson, Soshnikov} that the left-hand-sides of \eqref{asymp gap thm explicit  Wr} and \eqref{asymp gap thm explicit  Me} are (for any determinantal point process generated by one of the kernels \eqref{def Wright kernel} or \eqref{def Meijer G kernel}, recall Remark \ref{remark:existence} and Remark \ref{remark:existence2}) equal to the Fredholm determinants
\begin{align}\label{Fredholm det Wr and Me}
\det \left( 1- (1-t)\mathbb{K}^{\mathrm{Wr}} \Big|_{[0,s]}\right) \qquad \mbox{ and } \qquad \det \left( 1- (1-t)\mathbb{K}^{\mathrm{Me}} \Big|_{[0,s]}\right),\qquad t=e^{-2\pi\nu}
\end{align}
respectively. The kernels $\mathbb{K}^{\mathrm{Wr}}$ and $\mathbb{K}^{\mathrm{Me}}$ are known to be \textit{integrable} in the sense of Its-Izergin-Korepin-Slavnov (IIKS) \cite{IIKS} only for particular values of the parameters. For example, for $\theta = p/q$, $p,q \in \mathbb{N}_{>0}$, $\mathbb{K}^{\mathrm{Wr}}$ is integrable of size $p+q$ \cite{Zhang}, and there are associated Riemann-Hilbert (RH) problems of size $(p+q)\times(p+q)$. We expect the analysis of these RH problems to be rather complicated (except in the simplest case when $p=q=1$). Furthermore, for irrational values of $\theta$, $\mathbb{K}^{\mathrm{Wr}}$ is not known to be integrable at all. To circumvent this problem, we follow the ideas from \cite{ClaeysGirSti} and conjugate the operators appearing in \eqref{Fredholm det Wr and Me} with a Mellin transform. This allows us in Section \ref{Section: RH and diff identities} to derive a differential identity in $s$, i.e. to express, \textit{for all} values of the parameters, the derivatives
\begin{align}\label{partial s Fredholm det Wr and Me}
\partial_{s} \log \det \left( 1- (1-t)\mathbb{K}^{\mathrm{Wr}} \Big|_{[0,s]}\right) \qquad \mbox{ and } \qquad \partial_{s} \log \det \left( 1- (1-t)\mathbb{K}^{\mathrm{Me}} \Big|_{[0,s]}\right),
\end{align}
in terms of the solution, denoted $Y$, to a $2 \times 2$ RH problem. We then perform, in Section \ref{Section: steepest descent}, an asymptotic analysis of this RH problem by means of the Deift/Zhou \cite{DeiftZhou} steepest descent method. The local analysis requires the use of parabolic cylinder functions. By integrating in $s$ the derivatives \eqref{partial s Fredholm det Wr and Me}, we obtain
\begin{multline}\label{integration in s}
\log \det \left( 1- (1-t)\mathbb{K} \Big|_{[0,s]}\right) = \log \det \left( 1- (1-t)\mathbb{K} \Big|_{[0,M]}\right) \\
 + \int_{M}^{s} \partial_{\tilde{s}} \log \det \left( 1- (1-t)\mathbb{K} \Big|_{[0,\tilde{s}]}\right)d\tilde{s}, \qquad \mathbb{K} = \mathbb{K}^{\mathrm{Wr}}, \mathbb{K}^{\mathrm{Me}},
\end{multline}
for a certain constant $M$. By substituting the large $s$ asymptotics of \eqref{partial s Fredholm det Wr and Me} in the integrand of \eqref{integration in s}, we determine the functions $\mu$ and $\sigma^{2}$ of Theorems \ref{thm:expmoments Wrights} and \ref{thm:expmoments Meijer} in Section \ref{Section: Differential identity in s}. However, the quantity
\begin{align*}
\log \det \left( 1- (1-t)\mathbb{K} \Big|_{[0,M]}\right)
\end{align*}
is an unknown constant, so this method does not allow for the evaluation of $C$ (the constants of order $1$) of Theorems \ref{thm:expmoments Wrights} and \ref{thm:expmoments Meijer}. Such constants are notoriously difficult to compute explicitly \cite{Krasovsky}, and require the use of other, more complicated, differential identities. To obtain $C$, we will use a differential identity in $t$, i.e. we will express 
\begin{align}\label{partial t Fredholm det Wr and Me}
\partial_{t} \log \det \left( 1- (1-t)\mathbb{K} \Big|_{[0,s]}\right) \qquad \mathbb{K} = \mathbb{K}^{\mathrm{Wr}}, \mathbb{K}^{\mathrm{Me}},
\end{align}
in terms of $Y$ in Section \ref{Section: RH and diff identities}. Large $s$ asymptotics for the derivatives \eqref{partial t Fredholm det Wr and Me} appears to be rather complicated to obtain. In particular, it requires the explicit evaluation of certain (regularized) integrals involving parabolic cylinder functions. The key observation is that these integrals do not depend on any other parameters than $t$. Then we evaluate explicitly these complicated integrals by using the known expansion \eqref{moment asymp Bessel} from \cite{Charlier}. By integrating \eqref{partial t Fredholm det Wr and Me} in $t$, we have
\begin{multline}\label{integration in t}
\log \det \left( 1- (1-t)\mathbb{K} \Big|_{[0,s]}\right) = \log \det \left( 1- (1-t)\mathbb{K} \Big|_{[0,s]}\right)\bigg|_{t=1} \\
 + \int_{1}^{t} \partial_{\tilde{t}} \log \det \left( 1- (1-\tilde{t})\mathbb{K} \Big|_{[0,s]}\right)d\tilde{t}, \qquad \mathbb{K} = \mathbb{K}^{\mathrm{Wr}}, \mathbb{K}^{\mathrm{Me}}.
\end{multline}
By substituting the large $s$ asymptotics of \eqref{partial t Fredholm det Wr and Me} in the integrand of \eqref{integration in t} in Section \ref{Section: Differential identity in t}, we obtain $C$ (and moreover we recover the same functions $\mu,\sigma^{2}$ as obtained via the differential identity in $s$), since
\begin{align*}
\log \det \left( 1- (1-t)\mathbb{K} \Big|_{[0,s]}\right)\bigg|_{t=1} = 0.
\end{align*}
In conclusion, each of the two differential identities has its advantages and disadvantages: the differential identity in $s$ leads to an easier analysis, but does not allow for the evaluation of $C$, while the differential identity in $t$ is significantly more involved but allows to compute $C$. Another advantage of the differential identity in $s$ is that it allows to give the optimal estimates $\bigO(s^{-\frac{\theta}{1+\theta}})$ and $\bigO \big(s^{-\frac{1}{1+r-q}}\big)$ for the error terms of Theorems \ref{thm:expmoments Wrights} and \ref{thm:expmoments Meijer}, while with the differential identity in $t$, we are only able to prove that the errors are of order $\bigO(s^{-\frac{\theta}{2(1+\theta)}})$ and $\bigO \big(s^{-\frac{1}{2(1+r-q)}}\big)$.

\section{Proof of Theorem \ref{thm:rigidity}}\label{Section: rigidity}
In this section, we suppose that $X$ is a locally finite random point process on the real line which has a smallest particle almost surely, with counting function $N(s)$, and which is such that Assumptions \ref{assumptions} hold for certain constants $\mathrm{C}, a>0$, $s_0\in\mathbb R$, $M>\sqrt{2/a}$, and for certain functions $\mu, \sigma$.

\medskip

We start by establishing a bound for the tail of the probability distribution of the extremum of the normalized counting function.
\begin{lemma}\label{lemma: A r eps}
There exist $c>0$ and $s_{0}>0$ such that for any $\epsilon>0$ sufficiently small and $s>s_{0}$,
\begin{align}\label{prob statement lemma 2.1}
\mathbb P\left(\sup_{x> s}\left|\frac{N(x)-\mu(x)}{\sigma^2(x)}\right| >  \sqrt{\frac{2}{a}(1+\epsilon)} \right)\leq \frac{c\, \mu(s)^{-\epsilon}}{2\epsilon}.
\end{align}
In particular, for any $\epsilon > 0$, 
\begin{align*}
\lim_{s \to + \infty} \mathbb{P}\left(\sup_{x>s}\left|\frac{N(x)-\mu(x)}{\sigma^2(x)}\right| \leq   \sqrt{\frac{2}{a}(1+\epsilon)}\right) = 1.
\end{align*}
\end{lemma}
\begin{proof}
Let us define $\kappa_{k}= \mu^{-1}(k)$. We start by noting that for $x\in[\kappa_{k-1},\kappa_k]$, $k\in\mathbb N$, we have by monotonicity of $\mu$ and of the counting function $N$ that
\begin{equation*}
N(x)-\mu(x)\leq N(\kappa_k)-\mu(\kappa_{k-1})
=N(\kappa_k)-\mu(\kappa_{k})+1,
\end{equation*}
and since $\sigma$ is increasing, we also have $\sigma^2(x)\geq \sigma^2(\kappa_{k-1})$.
For large enough $s$, it follows that 
\begin{align*}
\sup_{x>s}\frac{N(x)-\mu(x)}{\sigma^2(x)} \leq \sup_{k:\kappa_k>s}
\frac{N(\kappa_k)-\mu(\kappa_{k})+1}{\sigma^2(\kappa_{k-1})}.
\end{align*}
Hence, by a union bound, for any $\gamma > 0$ we have
\begin{multline}\label{lol2}
\mathbb P\left(\sup_{x>s}\frac{N(x)-\mu(x)}{\sigma^2(x)}>\gamma\right)\leq \sum_{k:\kappa_k>s}\mathbb P\left(\frac{N(\kappa_k)-\mu(\kappa_{k})+1}{\sigma^2(\kappa_{k-1})}>\gamma\right) \\ = \sum_{k:\kappa_k>s} \mathbb{P}\left( e^{\gamma N(\kappa_{k})} > e^{\gamma \mu(\kappa_{k})-\gamma + \gamma^{2} \sigma^{2}(\kappa_{k-1})} \right) \leq \sum_{k:\kappa_k>s}\mathbb E\left(e^{\gamma N(\kappa_k)}\right)e^{-\gamma\mu(\kappa_k)+\gamma-\gamma^2 \sigma^2(\kappa_{k-1})},
\end{multline}
where the last inequality is obtained by applying Markov's inequality on the positive random variable $e^{\gamma N(\kappa_{k})}$. Using \eqref{expmomentbound} in \eqref{lol2}, we obtain
\begin{align*}
\mathbb P\left(\sup_{x>s}\frac{N(x)-\mu(x)}{\sigma^2(x)}>\gamma\right)\leq 
\mathrm{C} \, 
e^\gamma\sum_{k:\kappa_k>s} e^{-\frac{\gamma^2}{2} \sigma^2(\kappa_{k})}e^{\gamma^2(\sigma^{2}(\kappa_{k})-\sigma^{2}(\kappa_{k-1}))}.
\end{align*}
Because $(\sigma^2\circ\mu^{-1})$ is strictly concave and behaves as $(\sigma^2\circ\mu^{-1})(k) \sim a \log k$ as $k \to + \infty$, we have that
\[e^{\gamma^2(\sigma^{2}(\kappa_{k})-\sigma^{2}(\kappa_{k-1}))}=e^{\gamma^2[\sigma^{2}(\mu^{-1}(k))-\sigma^{2}(\mu^{-1}(k-1))]}\] decreases with $k$ and is uniformly bounded in $k$ by a constant which we denote as $C'$ and which we can choose independently of $\gamma\in[0,M]$.
Using also the fact that $\sigma^2$ and $\mu$ are increasing, we obtain
\begin{align}
\mathbb P\left(\sup_{x>s}\frac{N(x)-\mu(x)}{\sigma^2(x)}>\gamma\right)&\leq 
\mathrm{C} \, C' e^\gamma \sum_{k:\kappa_k>s} e^{-\frac{\gamma^2}{2} \sigma^2(\mu^{-1}(k))} \nonumber \\
&\leq \mathrm{C} \, C'e^\gamma\left(e^{-\frac{\gamma^2}{2} \sigma^2(s)}+ \int_{\mu(s)}^\infty e^{-\frac{\gamma^2}{2}(\sigma^2\circ\mu^{-1})(x)}dx\right). \label{bound1 sup}
\end{align}
Similarly, we obtain
\begin{multline*}
\mathbb P\left(\sup_{x>s}\frac{\mu(x)-N(x)}{\sigma^2(x)}>\gamma\right)\leq \sum_{k:\kappa_k>s}\mathbb P\left(\frac{\mu(\kappa_{k-1})-N(\kappa_{k-1})+1}{\sigma^2(\kappa_{k-1})}>\gamma\right) \\ = \sum_{k:\kappa_{k+1}>s} \mathbb{P}\left( e^{-\gamma N(\kappa_{k})} > e^{-\gamma \mu(\kappa_{k})-\gamma + \gamma^{2} \sigma^{2}(\kappa_{k})} \right) \leq \sum_{k:\kappa_{k+1}>s}\mathbb E\left(e^{-\gamma N(\kappa_k)}\right)e^{\gamma\mu(\kappa_k)+\gamma-\gamma^2 \sigma^2(\kappa_{k})}.
\end{multline*}
Using again \eqref{expmomentbound} and the fact that $\sigma$ and $\mu$ are increasing, we then get
\begin{align}
\mathbb P\left(\sup_{x>s}\frac{\mu(x)-N(x)}{\sigma^2(x)}>\gamma\right)&\leq 
\mathrm{C} \, e^\gamma \sum_{k:\kappa_{k+1}>s} e^{-\frac{\gamma^2}{2} \sigma^2(\mu^{-1}(k))} \nonumber \\
&\leq \mathrm{C} \, e^\gamma\left(2 \, e^{-\frac{\gamma^2}{2} (\sigma^2(\mu^{-1}(\mu(s)-1))}+ \int_{\mu(s)}^\infty e^{-\frac{\gamma^2}{2}(\sigma^2\circ\mu^{-1})(x)}dx\right). \label{bound2 sup}
\end{align}
By combining \eqref{bound1 sup} and \eqref{bound2 sup}, we obtain
\begin{align*}
\mathbb P\left(\sup_{x>s}\left|\frac{N(x)-\mu(x)}{\sigma^2(x)}\right|>\gamma\right)\leq \mathrm{C}(C'+2)e^\gamma\left(e^{-\frac{\gamma^2}{2} (\sigma^2\circ\mu^{-1})(\mu(s)-1)}+ \int_{\mu(s)}^\infty e^{-\frac{\gamma^2}{2}(\sigma^2\circ\mu^{-1})(x)}dx\right).
\end{align*}
It follows from criteria (2) and (3) of Assumptions \ref{assumptions} that the right hand side converges to $0$ as $s\to +\infty$, provided that $\gamma>\sqrt{2/a}$.
More precisely, for $\gamma \in (\sqrt{2/a},M]$ the right hand side is smaller than
\[2\mathrm{C}(C'+2)e^M \left((\mu(s)-1)^{-\frac{a \gamma^2}{2}}+\int_{\mu(s)}^{\infty}x^{-\frac{a \gamma^2}{2}}dx \right)\leq 3\mathrm{C}(C'+2)e^M\frac{\mu(s)^{1-\frac{a \gamma^2}{2}}}{\frac{a \gamma^2}{2}-1}
\]
for all sufficiently large $s$. In conclusion, for any $\gamma \in (\sqrt{2/a},M]$, we have
\begin{align*}
\mathbb P\left(\sup_{x>s}\left|\frac{N(x)-\mu(x)}{\sigma^2(x)}\right|>\gamma\right) \leq 3\mathrm{C}(C'+2)e^M\frac{\mu(s)^{1-\frac{a \gamma^2}{2}}}{\frac{a \gamma^2}{2}-1}.
\end{align*} 
We obtain the claim after setting $c = 6\mathrm{C}(C'+2)e^M$ and $\gamma = \sqrt{\frac{2}{a}(1+\epsilon)}$.
\end{proof}
Next, assuming a bound for the extremum of the normalized counting function, we derive the global rigidity of the points in the process.
\begin{lemma}\label{lemma: xk not far away from kappa k}
Let $\epsilon > 0$. For all sufficiently large $s$, if the event
\begin{align}\label{event holds true}
\sup_{x> s}\left|\frac{N(x)-\mu(x)}{\sigma^2(x)}\right|=\sup_{x\geq s}\left|\frac{N(x)-\mu(x)}{\sigma^2(x)}\right| \leq   \sqrt{\frac{2}{a}(1+\epsilon)}
\end{align}
holds true, then we have
\begin{align}\label{upper and lower bound for mu xk}
\sup_{k \geq \mu(2s)} \frac{|\mu(x_k) - k|}{(\sigma^2\circ\mu^{-1})(k)} \leq \sqrt{\frac{2}{a}(1+2\epsilon)}.
\end{align}
\end{lemma}

\begin{proof}
We start by proving that
\begin{align}\label{intermediate proof}
x_{k} > s, \qquad \mbox{for all } k\geq \mu(2s),
\end{align}
for all large enough $s$. Suppose that $x_k\leq s<2s \leq \kappa_k$, where $\kappa_{k} = \mu^{-1}(k)$. Then 
\begin{align*}
\mu(2s) \leq  \mu(\kappa_k)=k=N(x_k) \leq N(s),
\end{align*}
which implies by Assumptions \ref{assumptions} that
\begin{align}\label{lol4}
\frac{N(s)-\mu(s)}{\sigma^2(s)}\geq \frac{\mu(2s)-\mu(s)}{\sigma^2(s)}\geq \frac{s\inf_{s\leq \xi\leq 2s}\mu'(\xi)}{\sigma^2(s)} \geq \frac{\inf_{s\leq \xi\leq 2s}\xi\mu'(\xi)}{2 \, \sigma^2(s)} = \frac{s\mu'(s)}{2 \, \sigma^2(s)}.
\end{align}
Again by Assumptions \ref{assumptions}, the right-hand-side of \eqref{lol4} tends to $+\infty$ as $s\to +\infty$, so there is a contradiction with \eqref{event holds true}, provided that $s$ is chosen large enough. We conclude that $x_k > s$ for all $k \geq \mu(2s)$, provided that $s$ is large enough.

We split the proof of \eqref{upper and lower bound for mu xk} into two parts. We first prove the following upper bound for $\mu(x_k)$:
\begin{align}\label{upper bound in proof}
\mu(x_k)\leq k+ \sqrt{\frac{2}{a}(1+2\epsilon)}(\sigma^2\circ\mu^{-1})(k), \qquad \mbox{for all } k\geq \mu(2s).
\end{align}
Define $m = m(k)$ as the unique integer
such that $\kappa_{k+m}<x_k\leq \kappa_{k+m+1}$. If $m < 0$, then \eqref{upper bound in proof} is immediately satisfied. Let us now treat the case $m \geq 0$. Since $k \geq \mu(2s)$, we have $x_{k}>s$ by \eqref{intermediate proof}. Therefore, we use \eqref{event holds true} together with $m \geq 0$ to conclude that
\begin{align*}
\sqrt{\frac{2}{a}(1+\epsilon)}\geq\frac{\mu(x_k)-N(x_k)}{\sigma^2(x_k)}\geq \frac{m}{(\sigma^2 \circ\mu^{-1})(k+m+1)},
\end{align*}
and it follows that 
\begin{align*}
m\leq \sqrt{\frac{2}{a}(1+\epsilon)}(\sigma^2 \circ\mu^{-1})(k+m+1)\leq \sqrt{\frac{2}{a}(1+\epsilon)}\left((\sigma^2\circ \mu^{-1})(k)+(m+1)(\sigma^2\circ \mu^{-1})'(k)\right),
\end{align*}
where we used the concavity of $\sigma^2\circ \mu^{-1}$ from Assumptions \ref{assumptions}. This inequality can be rewritten as
\begin{align*}
 \Big( 1- \sqrt{\frac{2}{a}(1+\epsilon)}(\sigma^2\circ \mu^{-1})'(k) \Big)m \leq  \sqrt{\frac{2}{a}(1+\epsilon)}\left((\sigma^2\circ \mu^{-1})(k)+(\sigma^2\circ \mu^{-1})'(k)\right).
\end{align*}
Since $\sigma \circ \mu^{-1}$ is concave, the derivative $(\sigma \circ \mu^{-1})'$ is decreasing, and thus for $k \geq k_{0}$ we have
\begin{align}\label{lol5}
(\sigma \circ \mu^{-1})(k) = (\sigma \circ \mu^{-1})(k_{0}) + \int_{k_{0}}^{k}(\sigma \circ \mu^{-1})'(\tilde{k})d\tilde{k} \geq (\sigma \circ \mu^{-1})(k_{0}) + (\sigma \circ \mu^{-1})'(k)(k-k_{0}).
\end{align}
Since $(\sigma \circ \mu^{-1})(k) \sim a \log(k)$ as $k \to + \infty$, \eqref{lol5} yields $(\sigma \circ \mu^{-1})'(k) \to 0$ as $k \to + \infty$. We deduce that, for any fixed $\delta > 0$, 
\begin{align*}
(1-\delta)m\leq (1+\delta)\sqrt{\frac{2}{a}(1+\epsilon)}(\sigma^2\circ \mu^{-1})(k)-(1-\delta), \qquad \mbox{for all } k \geq \mu(2s),
\end{align*}
provided that $s$ is large enough. We choose $\delta>0$ sufficiently small such that
\begin{align*}
\frac{1+\delta}{1-\delta}\sqrt{\frac{2}{a}(1+\epsilon)} < \sqrt{\frac{2}{a}(1+2\epsilon)}.
\end{align*}
Therefore, we achieve the inequality 
\begin{align*}
m+1\leq \sqrt{\frac{2}{a}(1+2\epsilon)}(\sigma^2\circ \mu^{-1})(k), \qquad \mbox{for all } k \geq  \mu(2s),
\end{align*}
provided that $s$ is large enough. It follows that
\begin{align*}
\mu(x_k)\leq \mu(\kappa_{k+m+1}) = k+m+1\leq k+\sqrt{\frac{2}{a}(1+2\epsilon)}(\sigma^2\circ \mu^{-1})(k).
\end{align*}
In the second part of the proof, we show the following lower bound for $\mu(x_{k})$:
\begin{align}\label{lower bound in proof}
k- \sqrt{\frac{2}{a}(1+\epsilon)}(\sigma^2\circ\mu^{-1})(k)\leq \mu(x_k), \qquad \mbox{for all } k\geq \mu(2s)
\end{align}
which is even slightly better than what is required to prove \eqref{upper and lower bound for mu xk}. Suppose that $\mu(x_k)<k-m$ with $m>0$. By combining \eqref{intermediate proof} with \eqref{event holds true}, we have
\begin{align*}
\sqrt{\frac{2}{a}(1+\epsilon)}\geq \frac{N(x_k)-\mu(x_k)}{\sigma^2(x_k)}>\frac{m}{\sigma^2(x_k)} > \frac{m}{(\sigma^{2} \circ \mu^{-1})(k)}, \qquad \mbox{for all } k \geq \mu(2s),
\end{align*}
and it follows that $m< \sqrt{\frac{2}{a}(1+\epsilon)}(\sigma^2\circ \mu^{-1})(k)$, which proves the lower bound.
\end{proof}

It now suffices to combine the above two results in order to obtain Theorem \ref{thm:rigidity}.
\begin{proof}[Proof of Theorem \ref{thm:rigidity}]
It follows from Lemma \ref{lemma: A r eps} that there exists $c>0$ such that for all $\epsilon > 0$ sufficiently small and for all $s$ sufficiently large, we have
\begin{align}\label{bayes1}
\mathbb{P}\left( \sup_{x>s} \frac{|N(x)-\mu(x)|}{\sigma^{2}(x)} \leq \sqrt{\frac{2}{a}\left( 1+\frac{\epsilon}{2} \right)} \right) \geq 1 - \frac{c \, \mu(s)^{-\frac{\epsilon}{2}}}{\epsilon}.
\end{align}
Furthermore, Lemma \ref{lemma: xk not far away from kappa k} implies that
\begin{align}\label{bayes2}
\mathbb{P}\left( \sup_{k \geq \mu(2s)}\frac{|\mu(x_k)-k|}{\sigma^2(\mu^{-1}(k))} \leq  \sqrt{\frac{2}{a}(1+\epsilon)} \; \bigg| \; \sup_{x> s} \frac{|N(x)-\mu(x)|}{\sigma^{2}(x)} \leq \sqrt{\frac{2}{a}\left( 1+\frac{\epsilon}{2} \right)} \right) = 1,
\end{align}
for all sufficiently large $s$. Theorem \ref{thm:rigidity} follows by a direct application of Bayes' formula, combining \eqref{bayes1} and \eqref{bayes2}.
\end{proof}

\section{RH problem and differential identities}\label{Section: RH and diff identities}
\paragraph{Double contour integral representation for the kernels.} For convenience, let us write
\begin{align*}
\mathbb{K}^{(1)} = \mathbb{K}^{\mathrm{Me}} \qquad \mbox{ and } \qquad  \mathbb{K}^{(2)} = \mathbb{K}^{\mathrm{Wr}},
\end{align*}
where $\mathbb{K}^{\mathrm{Me}}$ and $\mathbb{K}^{\mathrm{Wr}}$ have been defined in \eqref{def Meijer G kernel} and \eqref{def Wright kernel}, respectively. For our analysis, we will use the following double contour representation for these kernels \cite{ClaeysGirSti}:
\begin{equation}
\label{limiting kernels1}\mathbb K^{(j)}(x,y)=\frac{1}{4\pi^2}\int_\gamma d u\int_{\tilde\gamma}d v\frac{F^{(j)}(u)}{F^{(j)}(v)}\frac{x^{-u}y^{v-1}}{u-v} ,\qquad j=1,2,
\end{equation}
with
\begin{equation}\label{defF12}
F^{(1)}(z)=\frac{\Gamma(z)\prod_{k = 1}^q \Gamma\left(1+\mu_k -z\right)}{  \prod_{k=1}^r \Gamma\left(1+\nu_k -z\right)},\qquad 
F^{(2)}(z)=\frac{\Gamma(z+\frac{\alpha}{2})}{\Gamma\left(\frac{\frac{\alpha}{2}+1-z}{\theta}\right)}.
\end{equation}
For $j=1$, we require $r,q \in \mathbb{N}$, $r > q \geq 0$, $\nu_{1},\ldots,\nu_{r} \in \mathbb{N}$ and $\mu_{1},\ldots,\mu_{q} \in \mathbb{N}_{>0}$ such that $\mu_{k} > \nu_{k}$, $k=1,\ldots,q$. If $q=0$, the product in the numerator is understood as $1$.
For $j=2$, we require $\alpha>-1$ and $\theta>0$. 
The contours $\gamma, \tilde\gamma$ are both oriented upward, do not intersect each other, and intersect the real line on the interval $(0,1+\nu_{\min})$ if $j=1$, with $\nu_{\min} := \min\{\nu_{1},\ldots,\nu_{r}\}$, and on the interval $(-\frac{\alpha}{2},1+\frac{\alpha}{2})$ if $j=2$. Furthermore, $\gamma$ tends to infinity in sectors lying strictly in the left half plane, and $\tilde\gamma$ tends to infinity in sectors lying strictly in the right half plane, see Figure \ref{fig: gamma and gamma tilde}.

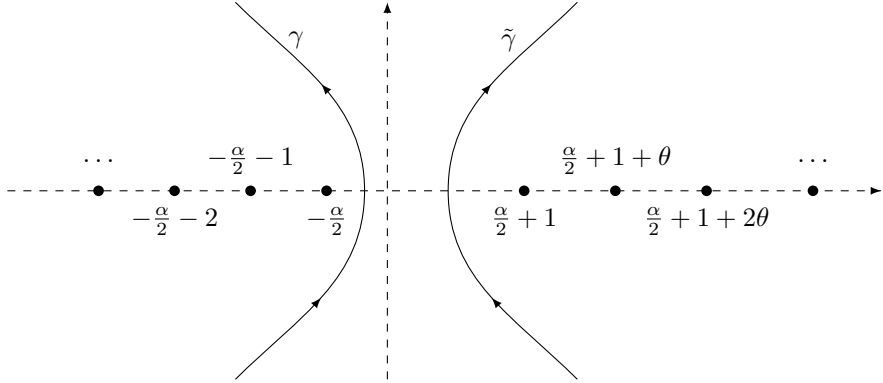
\begin{figure}
\begin{center}
\begin{tikzpicture}
\node at (0,0) {};
\draw[dashed,->-=1,black] (0,-2.5) to [out=90, in=-90] (0,2.5);
\draw[dashed,->-=1,black] (-5,0) to [out=0, in=-180] (6.5,0);

\node at (-1.2,2) {$\gamma$};
\draw[->-=0.5,black] (-2,-2.5) to [out=45, in=-90] (-0.3,0);
\draw[->-=0.5,black] (-0.3,0) to [out=90, in=-45] (-2,2.5);

\node at (1.6,2) {$\tilde{\gamma}$};
\draw[->-=0.5,black] (2.5,-2.5) to [out=135, in=-90] (0.8,0);
\draw[->-=0.5,black] (0.8,0) to [out=90, in=-135] (2.5,2.5);


\fill (-0.8,0) circle (0.07cm);
\node at (-0.8,-0.4) {$-\frac{\alpha}{2}$};
\fill (-1.8,0) circle (0.07cm);
\node at (-1.8,0.4) {$-\frac{\alpha}{2}-1$};
\fill (-2.8,0) circle (0.07cm);
\node at (-2.8,-0.4) {$-\frac{\alpha}{2}-2$};
\fill (-3.8,0) circle (0.07cm);
\node at (-3.8,0.4) {$\ldots$};

\fill (1.8,0) circle (0.07cm);
\node at (1.8,-0.4) {$\frac{\alpha}{2}+1$};
\fill (3,0) circle (0.07cm);
\node at (3,0.4) {$\frac{\alpha}{2}+1+\theta$};
\fill (4.2,0) circle (0.07cm);
\node at (4.2,-0.4) {$\frac{\alpha}{2}+1+2\theta$};
\fill (5.6,0) circle (0.07cm);
\node at (5.6,0.4) {$\ldots$};
\end{tikzpicture}
\end{center}
\caption{\label{fig: gamma and gamma tilde}\textit{The contours $\gamma$ and $\tilde{\gamma}$.}}
\end{figure}

\paragraph{Integrable kernels.} An mentioned at the end of Section \ref{Section:intro}, the kernels $\mathbb{K}^{(j)}$, $j=1,2$ are known to be integrable only for particular values of the parameters. With minor modifications of \cite[Propositions 2.1 and 2.2]{ClaeysGirSti}, we obtain the following.
\begin{proposition}
Let $t \in (0,+\infty)$. For $j=1,2$, we have
\begin{equation}\label{eq:detid}
\det \left( 1 - (1-t)\mathbb{K}^{(j)}\Big|_{[0,s]} \right) = \det \left( 1 - \mathbb{M}_{s,t}^{(j)} \right), \qquad j = 1,2,
\end{equation}
where $\mathbb{M}_{s,t}^{(j)}$ is the integral operator acting on $L^{2}(\gamma \cup \tilde{\gamma})$ with kernel
\begin{align*}
\mathbb{M}_{s,t}^{(j)}(u,v) = \frac{f(u)^{T}g(v)}{u-v},
\end{align*}
where $f$ and $g$ are given by
\begin{align}\label{def of f and g}
f(u) = \frac{1}{2\pi i} \begin{pmatrix}
\chi_{\gamma}(u) \\
s^{u} \chi_{\tilde{\gamma}}(u)
\end{pmatrix}, \qquad g(v) = \begin{pmatrix}
-\sqrt{1-t}F^{(j)}(v)^{-1}\chi_{\tilde{\gamma}}(v), \\
\sqrt{1-t}s^{-v} F^{(j)}(v) \chi_{\gamma}(v)
\end{pmatrix},
\end{align}
and $\chi_{\gamma}$ and $\chi_{\tilde{\gamma}}$ are the characteristic functions of $\gamma$ and $\tilde{\gamma}$, respectively. The determination of $\sqrt{1-t}$ in the definition of $g$ is unimportant (but the same determination must be chosen for both entries of $g$). For definiteness, we require
\begin{align}\label{choice of sqrt(1-t)}
& \sqrt{1-t} \in [0,1),  & &  \mbox{if } t \in (0,1], \\
& \sqrt{1-t} \in [0,i \, \infty), & & \mbox{if } t \in [1,+\infty).
\end{align}
\end{proposition}
\begin{proof}
The proof for $t=0$ can be found in \cite[Propositions 2.1 and 2.2]{ClaeysGirSti} and relies on a conjugation of $\left.\mathbb K^{(j)}\right|_{[0,s]}$ with a Mellin transform. The proof for arbitrary values of $t \in (0,+\infty)$ only requires minor modifications: the quantity $\mathbb{H}_{s}^{(j)}(v,z)$ of \cite[eq (2.10)]{ClaeysGirSti}\footnote{There is a factor $\frac{1}{2\pi i}$ missing in the expressions for $\mathbb{H}_{s}^{(j)}(v,z)$ and $B^{(j)}(v,u)$ of \cite[eq (2.10) and (2.19)]{ClaeysGirSti}.} needs to modified to
\begin{align*}
\frac{\sqrt{1-t}}{2\pi i} \int_{\gamma} \frac{du}{2\pi i}s^{z-u} \frac{F^{(j)}(u)}{F^{(j)}(v)(v-u)(z-u)},
\end{align*}
and the kernels $A^{(j)}$ and $B^{(j)}$ of \cite[eq (2.19)]{ClaeysGirSti} need to be modified to 
\begin{align*}
A^{(j)}(u,z) = \sqrt{1-t} \frac{s^{z-u}F^{(j)}(u)}{2\pi i (z-u)}, \qquad B^{(j)}(v,u) = \frac{\sqrt{1-t}}{2\pi i F^{(j)}(v)(v-u)},
\end{align*}
where the determination of $\sqrt{1-t}$ is unimportant, as long as it is the same for $A^{(j)}$ and $B^{(j)}$.
\end{proof}

Using a method developed by Its, Izergin, Korepin, and Slavnov \cite{IIKS}, we will establish differential identities in $s$ and $t$ for the logarithm of the Fredholm determinants \eqref{eq:detid} in terms of the following RH problem:
\subsubsection*{RH problem for $Y=Y^{(j)}$, $j=1,2$}
\begin{itemize}
\item[(a)] $Y : \mathbb{C} \setminus (\gamma \cup \tilde{\gamma}) \to \mathbb{C}^{2 \times 2}$ is analytic. 
\item[(b)] $Y(z)$ has continuous boundary values $Y_{\pm}(z)$ as $z$ approaches the contour $\gamma \cup \tilde{\gamma}$ from the left ($+$) and right ($-$), according to its orientation, and we have the jump relations
\begin{align*}
Y_{+}(z) = Y_{-}(z) J(z), \qquad z \in \gamma \cup \tilde{\gamma},
\end{align*}
with jump matrix $J=J^{(j)}$ given by
\begin{align}\label{jumps of Y}
J(z) = I - 2 \pi i f(z) g(z)^{T} = \begin{cases}
\begin{pmatrix}
1 & -\sqrt{1-t}s^{-z}F^{(j)}(z) \\
0 & 1
\end{pmatrix}, & z \in \gamma, \\
\begin{pmatrix}
1 & 0 \\
\sqrt{1-t} s^{z} F^{(j)}(z)^{-1} & 1 
\end{pmatrix}, & z \in \tilde{\gamma}.
\end{cases}
\end{align}
\item[(c)] As $z \to \infty$, there exists $Y_{1} = Y_{1}^{(j)}(s,t)$ independent of $z$ such that
\begin{align*}
Y(z) = I + \frac{Y_{1}}{z} + \bigO(z^{-2}).
\end{align*}
\end{itemize}
\begin{remark}\label{remark: symmetry}
We have some freedom in the choice of $\gamma$ and $\tilde{\gamma}$. We choose them symmetric with respect to the real line. This symmetry will be useful later to simplify computations.
\end{remark}
\begin{lemma}\label{lemma: differential identities}
For $j=1,2$, we have the following differential identities:
\begin{align}
& \partial_{s} \log \det \left( 1 - (1-t)\mathbb{K}^{(j)}\Big|_{[0,s]} \right) = \frac{Y_{1,11}}{s}, \label{diff identity with s} \\
& \partial_{t} \log \det \left( 1 - (1-t)\mathbb{K}^{(j)}\Big|_{[0,s]} \right) = \frac{-1}{2(1-t)} \int_{\gamma \cup \tilde{\gamma}} {\rm Tr} \Big[ Y^{-1}(z)Y'(z)(J(z)-I) \Big] \frac{dz}{2\pi i}. \label{diff identity with t}
\end{align}
\end{lemma}
\begin{remark}\label{remark: no boundary values indicated}
We do not mention whether we take the $+$ or $-$ boundary values of $Y$ in the integrand of \eqref{diff identity with t}. This is without ambiguity, because
\begin{multline*}
{\rm Tr} \Big[ Y_{+}^{-1}(z)Y_{+}'(z)(J(z)-I) \Big] = {\rm Tr} \Big[ Y_{-}^{-1}(z)Y_{-}'(z)(J(z)-I) \Big] \\
={\rm Tr} \Big[ Y_{+}^{-1}(z)Y_{-}'(z)(J(z)-I) \Big] = {\rm Tr} \Big[ Y_{-}^{-1}(z)Y_{+}'(z)(J(z)-I) \Big].
\end{multline*}
\end{remark}
\begin{proof}
Both \eqref{diff identity with s} and \eqref{diff identity with t} are specializations of more general results from \cite{Bertola}.
For the proof of \eqref{diff identity with s}, we refer to \cite[page 13]{ClaeysGirSti}, and for the proof of \eqref{diff identity with t}, we apply \cite[Section 5.1]{Bertola} (with $\partial = \partial_{t}$) to obtain
\begin{align*}
\partial_{t} \log \det \left( 1 - (1-t)\mathbb{K}^{(j)}\Big|_{[0,s]} \right) = \int_{\gamma \cup \tilde{\gamma}} {\rm Tr} \Big[ Y_{-}^{-1}(z)Y_{-}'(z) \partial_{t}J(z) J^{-1}(z) \Big] \frac{dz}{2\pi i}.
\end{align*}
From \eqref{jumps of Y}, it is straightforward to verify that
\begin{equation*}
\partial_{t}J(z)J(z)^{-1} = \frac{-1}{2(1-t)} (J(z)-I),
\end{equation*}
which yields \eqref{diff identity with t} and finishes the proof.
\end{proof}

\section{Steepest descent analysis}\label{Section: steepest descent}
In this section, we use the Deift/Zhou \cite{DeiftZhou} steepest descent method to perform an asymptotic analysis of $Y = Y^{(j)}$, $j=1,2$, as $s \to + \infty$ uniformly for $t$ in compact subsets of $(0,+\infty)$. The first transformation $Y \mapsto U$ in Section \ref{subsection: Y to U} is a rescaling which is identical to the one from \cite{ClaeysGirSti}. The rest of the analysis differs drastically from \cite{ClaeysGirSti}, and we highlight the main ideas for it here. As common in steepest descent analysis of RH problems, we will need to do a saddle points analysis of a phase function appearing in the jump matrix for $U$ (see Section \ref{subsection: saddles point analysis}). The opening of the lenses is done in two steps $U \mapsto \widehat{T} \mapsto T$ presented in Section \ref{subsection: U to T}, and we emphasize that this is somewhat unusual, in that it requires two different factorizations of the jump matrix, each on a different part of the jump contour. The global parametrix $P^{(\infty)}$ of Section \ref{subsection: global param} approximates $T$ everywhere in the complex plane except near the saddle points $b_{1}$ and $b_{2}$. In Sections \ref{subsection: local param at b2} and \ref{subsection: local param at b1}, we construct local parametrices $P^{(b_{k})}$ in terms of parabolic cylinder functions. The local parametrix $P^{(b_{k})}$ is defined in a small disk $\mathcal{D}_{b_{k}}$ centered at $b_{k}$ and satisfies the same jumps as $T$. The last step $T \mapsto R$ of the steepest descent analysis is completed in Section \ref{subsection: small norm}. A matrix $R$ is built in terms of $T$, $P^{(\infty)}$, $P^{(b_{1})}$, and $P^{(b_{2})}$, and we show that it satisfies a small norm RH problem. In particular, $R(z)$ is close to the identity matrix as $s \to + \infty$. We also compute the first two subleading terms of $R$ which are needed for the proofs of Theorems \ref{thm:expmoments Wrights} and \ref{thm:expmoments Meijer}. 
\subsection{First transformation $Y \mapsto U$}\label{subsection: Y to U}
We first rescale the variable of the RH problem for $Y$ in a convenient way. In the same way as in \cite[Section 3.1]{ClaeysGirSti}, we define
\begin{align}\label{def of U}
U(\zeta) = s^{\frac{\tau}{2}\sigma_{3}} Y \big( i s^{\rho} \zeta + \tau \big) s^{-\frac{\tau}{2}\sigma_{3}},
\end{align}
where $\tau=\tau^{(j)}$ and $\rho = \rho^{(j)}$, $j=1,2$, are given by
\begin{align}
& \tau^{(1)} = \frac{\nu_{\min}+1}{2}, & & \rho^{(1)} = \frac{1}{r-q+1}, \label{tau rho 1} \\
& \tau^{(2)} = \frac{1}{2}, & & \rho^{(2)} = \frac{\theta}{\theta + 1}, \label{tau rho 2}
\end{align}
and $\nu_{\min} := \min \{\nu_{1},\ldots,\nu_{r}\}$. The matrix $U$ satisfies the following RH problem.
\subsubsection*{RH problem for $U$}
\begin{itemize}
\item[(a)] $U : \mathbb{C}\setminus (\gamma_{U} \cup \gamma_{\tilde{U}}) \to \mathbb{C}^{2 \times 2}$ is analytic, where
\begin{align}\label{def of gamma U and gamma tilde U}
& \gamma_{U} = \{\zeta \in  \mathbb{C} : is^{\rho} \zeta + \tau \in \gamma \}, \qquad \mbox{ and } \qquad \tilde{\gamma}_{U} = \{\zeta \in  \mathbb{C} : is^{\rho} \zeta + \tau \in \tilde{\gamma} \}.
\end{align}
The contour $\gamma_{U}$ (resp. $\tilde{\gamma}_{U}$) lies in the upper (resp. lower) half plane and is oriented from left to right.
\item[(b)] $U$ satisfies the jumps $U_{+}(\zeta) = U_{-}(\zeta)J_{U}(\zeta)$ for $\zeta \in \gamma_{U} \cup \tilde{\gamma}_{U}$ with
\begin{align*}
J_{U}(\zeta) = \begin{cases}
\begin{pmatrix}
1 & -\sqrt{1-t} \; s^{-is^{\rho}\zeta}F(is^{\rho}\zeta + \tau) \\
0 & 1
\end{pmatrix}, & \mbox{if } \zeta \in \gamma_{U}, \\
\begin{pmatrix}
1 & 0 \\
\sqrt{1-t} \; s^{is^{\rho}\zeta}F(is^{\rho}+\tau)^{-1} & 1
\end{pmatrix}, & \mbox{if } \zeta \in \tilde{\gamma}_{U}.
\end{cases}
\end{align*}
\item[(c)] As $\zeta \to \infty$, we have
\begin{align*}
U(\zeta) = I + \frac{U_{1}}{\zeta} + \bigO(\zeta^{-2}), 
\end{align*}
where $U_{1} = U_{1}^{(j)}(s,t)$ is given by
\begin{align*}
U_{1} = \frac{1}{is^{\rho}} s^{\frac{\tau}{2}\sigma_{3}}Y_{1}s^{-\frac{\tau}{2}\sigma_{3}}.
\end{align*}
\end{itemize}
\begin{remark}\label{remark: symmetry for U}
Since $\gamma$ and $\tilde{\gamma}$ are symmetric with respect to the real line, the contours $\gamma_{U}$ and $\tilde{\gamma}_{U}$ are symmetric with respect to $i \mathbb{R}$. Furthermore, we note that the function
\begin{align*}
\zeta \mapsto f(\zeta) := s^{-is^{\rho}\zeta}F(is^{\rho}\zeta + \tau)
\end{align*}
satisfies the symmetry relation $f(\zeta) = \overline{f(-\overline{\zeta})}$, and thus we also have $J_{U}(\zeta) = \overline{J_{U}(-\overline{\zeta})}$ for $\zeta \in \gamma_{U} \cup \tilde{\gamma}_{U}$. By uniqueness of the RH solution $U$, we conclude that
\begin{align*}
U(\zeta) = \overline{U(-\overline{\zeta})}, \qquad \zeta \in \mathbb{C}\setminus (\gamma_{U} \cup \gamma_{\tilde{U}}).
\end{align*}
\end{remark}
\subsection{Saddle point analysis}\label{subsection: saddles point analysis}
We choose the branch for $\log F^{(j)}$, $j=1,2$, such that
\begin{align*}
& \log F^{(1)}(z) = \log \Gamma(z) -\sum_{k=1}^{r} \log \Gamma(1+\nu_{k}-z) + \sum_{k=1}^{q} \log \Gamma(1+\mu_{k}-z), \\
& \log F^{(2)}(z) = \log \Gamma \left( z + \frac{\alpha}{2} \right) - \log \Gamma \left( \frac{\frac{\alpha}{2}+1-z}{\theta} \right),
\end{align*}
where $z \mapsto \log \Gamma(z)$ is the log-gamma function, which has a branch cut along $(-\infty,0]$. Therefore, $z \mapsto \log F^{(1)}(z)$ has a branch cut along $(-\infty,0]  \cup [1+\nu_{\min},+\infty)$, and $z \mapsto \log F^{(2)}(z)$ has a branch cut along $(-\infty,-\frac{\alpha}{2}]\cup [1+\frac{\alpha}{2},+\infty)$. Asymptotics for $\log \big( s^{-is^{\rho}\zeta}F(is^{\rho} \zeta + \tau)\big)$ as $s \to + \infty$ and simultaneously $s^{\rho}\zeta \to \infty$, $\big|\arg(\zeta) \pm \frac{\pi}{2} \big| > \epsilon > 0$ were computed in \cite{ClaeysGirSti} and are given by
\begin{align*}
\log \big( s^{-is^{\rho}\zeta} F(is^{\rho}\zeta + \tau)\big) = & \; is^{\rho} [c_{1} \zeta \log(i \zeta) + c_{2} \zeta \log(-i\zeta) + c_{3} \zeta] \\
& \; + c_{4} \log(s) + c_{5} \log(i\zeta) + c_{6} \log(-i\zeta) + c_{7} + \frac{c_{8}}{is^{\rho}\zeta} + \bigO \left( \frac{1}{s^{2\rho}\zeta^{2}} \right),
\end{align*}
where the constants $\{c_{i}=c_{i}^{(j)}\}_{i=1}^{8}$, $j = 1,2$ are given by \cite[equations (3.10)--(3.12)]{ClaeysGirSti}\footnote{The upperscripts $j=1, 2$ in this paper correspond to the upperscripts $j=2, 3$ in \cite{ClaeysGirSti}. Also, there is a typo in \cite[equation (3.12)]{ClaeysGirSti} for the constant $c_{8}^{(3)}$. The correct value of $c_{8}^{(3)}$ can be found in \cite[equation (2.4)]{CLM2019}.}. The values of $c_{7}$ and $c_{8}$ turn out to be unimportant for us. We recall the values of the other constants here. For $j=1$, we have
\begin{align*}
& c_{1} = 1, & & c_{2} = r-q, & & c_{3} = -(r-q+1), \\
& c_{4} = \frac{\nu_{\min}}{2} - \frac{\sum_{k=1}^{r} \nu_{k} - \sum_{k=1}^{q} \mu_{k}}{r-q+1} & & c_{5} = \frac{\nu_{\min}}{2}, & & c_{6} = (r-q)\frac{\nu_{\min}}{2} - \sum_{j=1}^{r} \nu_{j} + \sum_{k=1}^{q} \mu_{k},
\end{align*}
and for $j=2$, we have
\begin{align*}
& c_{1} = 1, & & c_{2} = \frac{1}{\theta}, & & c_{3} = - \frac{\theta + 1 + \log \theta}{\theta}, \\
& c_{4} = \frac{(\theta-1)(1+\alpha)}{2(\theta + 1)}, & & c_{5} = \frac{\alpha}{2}, & & c_{6} = \frac{\theta-\alpha -1}{2 \theta}.
\end{align*}
Following \cite{ClaeysGirSti}, we define $h(\zeta) = h^{(j)}(\zeta)$, $j=1,2$, by
\begin{align}\label{def of h}
& h(\zeta) = -c_{1} \zeta \log(i \zeta) - c_{2} \zeta \log(-i \zeta) - c_{3} \zeta,
\end{align}
where the principal branch is chosen for the logarithms, and $\mathcal{G}=\mathcal{G}^{(j)}$ is defined via
\begin{align}\label{def of G}
s^{-is^{\rho}\zeta} F(is^{\rho}\zeta + \tau) = e^{-is^{\rho}h(\zeta)} \mathcal{G}(\zeta;s).
\end{align}
We have
\begin{align}\label{asymp for G}
\log \mathcal{G}(\zeta;s) = c_{4} \log s + c_{5} \log(i\zeta) + c_{6} \log(-i\zeta) + c_{7} + \frac{c_{8}}{is^{\rho}\zeta} + \bigO \left( \frac{1}{s^{2\rho}\zeta^{2}} \right)
\end{align}
as $s \to + \infty$ such that $s^{\rho}\zeta \to \infty$, $\big|\arg(\zeta) \pm \frac{\pi}{2} \big| > \epsilon > 0$. On the other hand, as $s \to +\infty$ and simultaneously $\zeta\to 0$ such that $s^{\rho}\zeta = \bigO(1)$, and such that $is^{\rho}\zeta + \tau$ is bounded away from the poles of $F$, we have $\mathcal{G}(\zeta;s) = \bigO(1)$. The jumps for $U$ can be rewritten in terms of $h$ and $\mathcal{G}$ as follows,
\begin{align}\label{JU in terms of G}
J_{U}(\zeta) = \begin{cases}
\begin{pmatrix}
1 & -\sqrt{1-t} \; e^{-is^{\rho}h(\zeta)} \mathcal{G}(\zeta;s) \\
0 & 1
\end{pmatrix}, & \mbox{if } \zeta \in \gamma_{U}, \\
\begin{pmatrix}
1 & 0 \\
\sqrt{1-t} \; e^{is^{\rho}h(\zeta)} \mathcal{G}(\zeta;s)^{-1} & 1
\end{pmatrix}, & \mbox{if } \zeta \in \tilde{\gamma}_{U}.
\end{cases}
\end{align}
\paragraph{Saddle points of $h$.} The saddle points of $h$ are the solutions to $h'(\zeta) = 0$. Using \eqref{def of h}, this equation can be written explicitly as follows:
\begin{align}\label{saddles point equation}
-(c_{1}+c_{2}+c_{3})-c_{1} \zeta \log(i\zeta) - c_{2} \zeta \log(-i\zeta) = 0.
\end{align}
A direct computation shows that this equation admits two solutions $\zeta = b_{2}$ and $\zeta = b_{1}$, where
\begin{align}\label{def of b2}
b_{2} = - \overline{b_{1}} = \exp\Big(\hspace{-0.1cm}-\hspace{-0.05cm}\frac{c_{1}+c_{2}+c_{3}}{c_{1}+c_{2}}\Big) \exp \Big( \frac{i\pi}{2}\frac{c_{2}-c_{1}}{c_{1}+c_{2}}\Big).
\end{align}
For the Meijer-$G$ point process (i.e. $j=1$), we have $c_{1}+c_{2}+c_{3} = 0$ and $c_{2} > c_{1}$, and therefore $b_{2}$ lies on the unit circle in the quadrant $Q_{1} := \{\zeta \in \mathbb{C}: \re \zeta \geq 0, \im \zeta \geq 0\}$. For the Wright's generalized Bessel process (i.e. $j = 2$), $b_{2}$ lies on the circle centered at the origin of radius  $\exp\big(\hspace{-0.1cm}-\hspace{-0.05cm}\frac{c_{1}+c_{2}+c_{3}}{c_{1}+c_{2}}\big)$; $b_{2}$ is in the quadrant $Q_{1}$ for $\theta \leq 1$, and in the quadrant $Q_{4} := \{\zeta \in \mathbb{C}: \re \zeta \geq 0, \im \zeta \leq 0\}$ for $\theta \geq 1$. Let us define 
\begin{align*}
\ell := \re (ih(b_{2})) = -(c_{1}+c_{2})\exp \left( - \frac{c_{1}+c_{2}+c_{3}}{c_{1}+c_{2}} \right) \sin \left( \frac{\pi}{2} \frac{c_{2}-c_{1}}{c_{1}+c_{2}} \right).
\end{align*}
We consider the zero set of $\re (ih)-\ell$:
\begin{align*}
\mathcal{N} = \{\zeta \in \mathbb{C}: \re(ih(\zeta)) = \ell\},
\end{align*}
which is visualized in Figure \ref{fig: real part of ih-ell}.
 
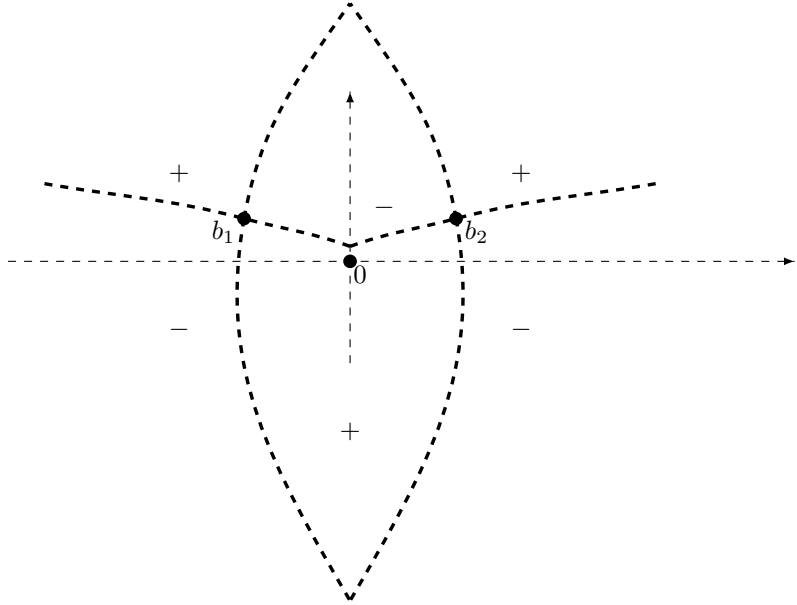
\begin{figure}
\begin{center}
\begin{tikzpicture}[scale = 0.9]
\node at (0,0) {};
\fill (0,0) circle (0.1cm);
\node at (0.15,-0.2) {$0$};
\fill (1.55,0.63) circle (0.1cm);
\node at (1.85,0.44) {$b_{2}$};
\fill (-1.55,0.63) circle (0.1cm);
\node at (-1.85,0.44) {$b_{1}$};

\draw[dashed,->-=1,black] (0,-1.5) to [out=90, in=-90] (0,2.5);
\draw[dashed,->-=1,black] (-5,0) to [out=0, in=-180] (6.5,0);

\draw[dashed,line width=0.45 mm,black] (0,0.22) to [out=20, in=180+15] (1.55,0.63) 
to [out=15, in=180+10] (4.5,1.15);
\draw[dashed,line width=0.45 mm,black] (0,0.22) to [out=180-20, in=-15] (-1.55,0.63)
to [out=180-15, in=-10] (-4.5,1.15);
\draw[dashed,line width=0.45 mm,black] (1.55,0.63) to [out=100, in=180+125] (0,3.8);
\draw[dashed,line width=0.45 mm,black] (1.55,0.63) to [out=-80, in=180-120] (0,-5);
\draw[dashed,line width=0.45 mm,black] (-1.55,0.63) to [out=180-100, in=-125] (0,3.8);
\draw[dashed,line width=0.45 mm,black] (-1.55,0.63) to [out=-100, in=180-60] (0,-5);

\node at (0,-2.5) {$+$};
\node at (-2.5,-1) {$-$};
\node at (2.5,-1) {$-$};
\node at (0.5,0.8) {$-$};
\node at (-2.5,1.3) {$+$};
\node at (2.5,1.3) {$+$};
\end{tikzpicture}
\end{center}
\caption{\label{fig: real part of ih-ell}\textit{The thick dashed curved correspond to $\re \big( ih(\zeta) \big)-\ell = 0$. These curves divide the complex plane into four regions, and we indicate the sign of $\re \big( ih(\zeta) \big)-\ell$ in each of these regions by the symbols $\pm$. The thin dashed curves represent the real and imaginary axis.}}
\end{figure}

\begin{lemma}\label{lemma: real part of h}

The set $\mathcal{N}$ consists of five simple curves $\Gamma_{j}$, $j=1,...,5$ and satisfies the symmetry $\mathcal{N} = -\overline{\mathcal{N}}$. Three of these curves, say $\Gamma_{1}$, $\Gamma_{2}$ and $\Gamma_{3}$, join $b_{2}$ with $b_{1}$. The curve $\Gamma_{4}$ starts at $b_{2}$ and leaves the right half plane in the sector $\arg \zeta \in (-\epsilon,\epsilon)$ for any $\epsilon > 0$. The last curve satisfies $\Gamma_{5} = - \overline{\Gamma_{4}}$. In particular, $\mathcal{N}$ divides the complex plane in four regions: two unbounded regions, and two bounded regions. Furthermore, the sign of $\re (ih(\zeta))-\ell$ in each of these regions is as shown in Figure \ref{fig: real part of ih-ell}.
\end{lemma}
\begin{proof}
We divide the proof in four steps. 

\hspace{-0.55cm}\textbf{Claim 1:} $\mathcal{N}$ intersects the imaginary axis at three distinct points $y_{1}<y_{2}<y_{3}$ such that $y_{1}<0$ and $y_{3}>0$.

To prove this, it suffices to inspect the graph of the function $y \mapsto \re(ih(iy))$ for $y \in \mathbb{R}$. It is a simple computation to verify that
\begin{align*}
\re (ih(iy)) = (c_{1}+c_{2})y \log |y| + c_{3} y, \qquad y  \in \mathbb{R}.
\end{align*}
This function is odd in the variable $y$, is equal to $0$ at $y=0$, tends to $+ \infty$ as $y \to + \infty$ and admits a local minimum at $y = y_{\star} := \exp \left( - \frac{c_{1}+c_{2}+c_{3}}{c_{1}+c_{2}} \right)$ where it takes the value
\begin{align}\label{re ih at ystar}
\re (ih(iy_{\star})) = -(c_{1}+c_{2})\exp \left( - \frac{c_{1}+c_{2}+c_{3}}{c_{1}+c_{2}} \right) < \ell.
\end{align}
Since $y \mapsto \re (ih(iy))$ is odd, and since
\begin{align*}
\re (ih(-iy_{\star})) = (c_{1}+c_{2})\exp \left( - \frac{c_{1}+c_{2}+c_{3}}{c_{1}+c_{2}} \right) > \ell,
\end{align*}
the equation $\re(i h(iy))=\ell$ admits three solutions $y_{1},y_{2},y_{3}$ satisfying 
\begin{align*}
y_{1} < -y_{\star}, \qquad y_{2} \in (-y_{\star},y_{\star}), \qquad y_{3} > y_{\star}.
\end{align*}

\hspace{-0.55cm}\textbf{Claim 2:} For any $\epsilon \in (0,\frac{\pi}{2})$, there exists $\rho_{\epsilon} > 0$ such that for all $\rho \geq \rho_{\epsilon}$, $\mathcal{N}$ intersects $
\{\rho \, e^{i \phi} : \phi \in (-\epsilon,\epsilon)\}$ at a single point.

This follows from a direct computation using the following expression for $\zeta = \rho \, e^{i \phi}$, $\phi \in (-\epsilon,\epsilon)$:
\begin{align*}
\re (ih(\zeta)) = \re(\zeta) \Big[ (c_{1}+c_{2})\big( \tan  \phi \log \rho + \phi \big) + c_{3} \tan \phi - \frac{\pi}{2}(c_{2}-c_{1}) \Big].
\end{align*}

\hspace{-0.55cm}\textbf{Claim 3}: There exists no closed curve $\Gamma$ lying entirely in either the left or right half plane such that $\Gamma\subset\mathcal{N}$. 

Since $h$ is analytic in $\mathbb{C}\setminus i \mathbb{R}$, $\zeta \mapsto \re (ih(\zeta))$ is harmonic in $\mathbb{C}\setminus i \mathbb{R}$. Let $\Gamma \subset \mathbb{C}\setminus i \mathbb{R}$ be a closed curve such that $\Gamma \subset \mathcal{N}$. The maximum principle for harmonic functions implies that $\re (ih(\zeta))\equiv\ell$ on the interior of $\Gamma$. Since $\re (ih(\zeta))$ is non-constant on any open disk, we conclude that there exists no such curve $\Gamma$.

\hspace{-0.55cm}\textbf{Proof of Lemma \ref{lemma: real part of h}:} 


Since $h'(b_{2})=0$, there are four curves $\{\Gamma_{j}\}_{j=1}^{4}$ emanating from $b_{2}$ that belong to $\mathcal{N}$. From \textbf{Claim 3}, none of these curves is a closed curve lying entirely in the right half plane. We conclude that these curves must leave the right half plane either on $i \mathbb{R}$ or at $\infty$. From \textbf{Claim 1} and \textbf{Claim 2}, three curves $\Gamma_{j}$, $j=1,2,3$ leave the right half plane on $i \mathbb{R}$ at $y_{1}$, $y_{2}$ and $y_{3}$, respectively, and the last curve $\Gamma_{4}$ leaves the right half plane at $\infty$ in the sector $\arg \zeta \in (-\epsilon,\epsilon)$ (for any $\epsilon > 0$ fixed). By $\re(ih(-\overline{\zeta})) = \re(ih(\zeta))$, $\mathcal{N}$ is symmetric with respect to $i \mathbb{R}$ and this proves that $\Gamma_{j}$, $j=1,2,3$, join $b_{2}$ and $b_{1}$, and that there exists $\Gamma_{5} \subset \mathcal{N}$ in the left half plane satisfying $\Gamma_{5} = - \overline{\Gamma_{4}}$. The sign of $\re(i h(\zeta)) - \ell$ in the topmost bounded region is negative by \eqref{re ih at ystar}. Since the sign of $\re(i h(\zeta)) - \ell$ changes every time a curve $\Gamma_{j}$, $j \in \{1,...,5\}$ is crossed, this determines the sign of $\re(i h(\zeta)) - \ell$ in the other regions as well.
\end{proof}

\subsection{Second transformation $U \mapsto T$}\label{subsection: U to T}
We will now define $T$ in terms of $U$ in two steps, $U \mapsto \widehat{T}$ and $\widehat{T}\mapsto T$. The transformation $U \mapsto \widehat{T}$ is similar to the one from \cite[Section 3.2]{ClaeysGirSti}. Let us define the union of two line segments $\Sigma_{5}:= [b_{1},0] \cup [0,b_{2}]$, as shown in Figure \ref{fig: Sigma 1,2,3,4}. $\widehat{T}$ consists of analytic continuations of $U$ in different regions, such that it has jumps on $\bigcup_{j=1}^{5} \Sigma_{j}$ instead of $\gamma_{U} \cup \tilde{\gamma}_{U}$, where the contours $\Sigma_{1},\ldots,\Sigma_{5}$ are shown in Figure \ref{fig: Sigma 1,2,3,4}. More precisely, denote $U_{\rm I}$ for the analytic continuation of the function $U$ as defined in the region above the contour $\gamma_U$, $U_{\rm II}$ for the analytic continuation of $U$ as defined in the region between $\gamma_U$ and $\tilde\gamma_U$, and $U_{\rm III}$ for the analytic continuation of $U$ as defined in the region below $\tilde\gamma_U$; then with the regions I', II', III' as in Figure \ref{fig: Sigma 1,2,3,4}, we define $\widehat T=U_{\rm I}$ in region I',  $\widehat T=U_{\rm II}$ in the two regions II', and $\widehat T=U_{\rm III}$ in region III'.

\begin{figure}
\begin{center}
\begin{tikzpicture}
\node at (0,0) {};
\fill (0,0) circle (0.1cm);
\node at (0.15,-0.2) {$0$};
\fill (1.55,0.63) circle (0.1cm);
\node at (1.25,0.8) {$b_{2}$};
\fill (-1.55,0.63) circle (0.1cm);
\node at (-1.25,0.8) {$b_{1}$};

\draw[dashed,->-=1,black] (0,-1.5) to [out=90, in=-90] (0,2.5);
\draw[dashed,->-=1,black] (-5,0) to [out=0, in=-180] (6.5,0);

\draw[dashed,line width=0.45 mm,black] (0,0.22) to [out=20, in=180+15] (1.55,0.63) 
to [out=15, in=180+10] (4.5,1.15);
\draw[dashed,line width=0.45 mm,black] (0,0.22) to [out=180-20, in=-15] (-1.55,0.63)
to [out=180-15, in=-10] (-4.5,1.15);
\draw[dashed,line width=0.45 mm,black] (1.55,0.63) to [out=100, in=180+125] (0,3.8);
\draw[dashed,line width=0.45 mm,black] (1.55,0.63) to [out=-80, in=180-120] (0,-5);
\draw[dashed,line width=0.45 mm,black] (-1.55,0.63) to [out=180-100, in=-125] (0,3.8);
\draw[dashed,line width=0.45 mm,black] (-1.55,0.63) to [out=-100, in=180-60] (0,-5);

\node at (1,3.5) {I'};
\node at (-4,1.8) {II'};
\node at (4,1.8) {II'};
\node at (1,-4.5) {III'};

\draw[->-=0.6,black,line width=0.45 mm] (0,0)--(1.55,0.63);
\draw[->-=0.6,black,line width=0.45 mm] (1.55,0.63)--($(1.55,0.63)+(15+45:3.2)$);
\draw[->-=0.6,black,line width=0.45 mm] (1.55,0.63)--($(1.55,0.63)+(15-45:2.9)$);
\node at (3,2.2) {$\Sigma_{2}$};
\node at (3.9,-0.4) {$\Sigma_{4}$};

\draw[-<-=0.5,black,line width=0.45 mm] (0,0)--(-1.55,0.63);
\draw[-<-=0.6,black,line width=0.45 mm] (-1.55,0.63)--($(-1.55,0.63)+(180-15-45:3.2)$);
\draw[-<-=0.6,black,line width=0.45 mm] (-1.55,0.63)--($(-1.55,0.63)+(180-15+45:2.9)$);
\node at (-3,2.2) {$\Sigma_{1}$};
\node at (-4,-0.4) {$\Sigma_{3}$};

\node at (-0.65,0) {$\Sigma_{5}$};
\end{tikzpicture}
\end{center}
\caption{\label{fig: Sigma 1,2,3,4}\textit{The jump contour $\cup_{i=1}^{5} \Sigma_{i}$ for the RH problem for $\widehat{T}$.}}
\end{figure}
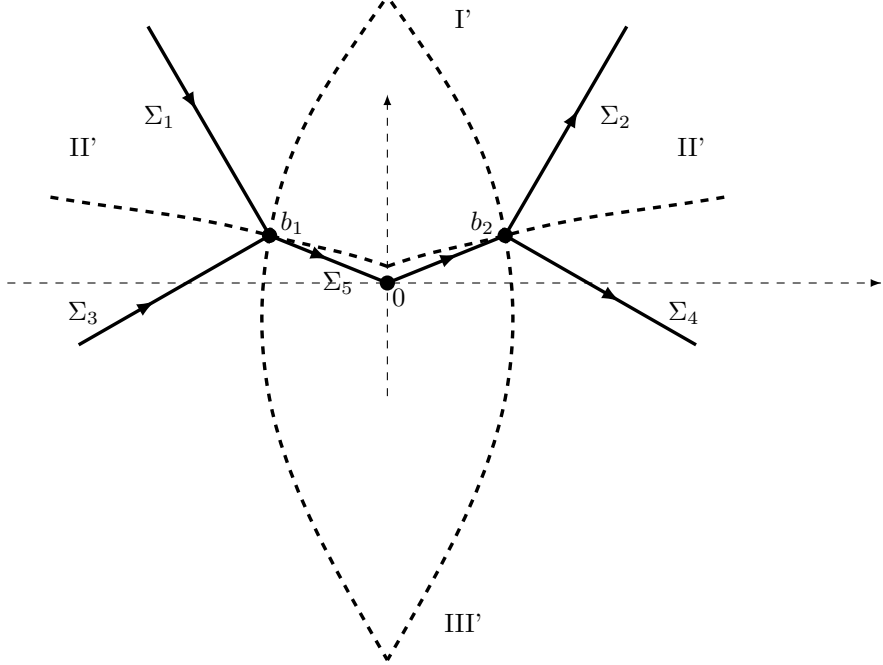
$\widehat{T}$ satisfies the same RH conditions as $U$, except for a modified jump relation on $\Sigma_5$.
\subsubsection*{RH problem for $\widehat{T}$}
\begin{itemize}
\item[(a)] $\widehat{T}$ is analytic in $\mathbb{C}\setminus \bigcup_{j=1}^{5} \Sigma_{j}$, where the contour $\bigcup_{j=1}^{5} \Sigma_{j}$ is shown in Figure \ref{fig: Sigma 1,2,3,4} and is chosen to be symmetric with respect to $i \mathbb{R}$.
\item[(b)] For $\zeta \in \bigcup_{j=1}^{5} \Sigma_{j}$, we have $\widehat{T}_{+}(\zeta) = \widehat{T}_{-}(\zeta) J_{\widehat{T}}(\zeta)$, where
\begin{align*}
J_{\widehat{T}}(\zeta) = \begin{cases}
\begin{pmatrix}
1 & -\sqrt{1-t} \, e^{-is^{\rho}h(\zeta)} \mathcal{G}(\zeta;s) \\
0 & 1
\end{pmatrix}, & \mbox{if } \zeta \in \Sigma_{1} \cup \Sigma_{2}, \\
\begin{pmatrix}
1 & 0 \\
\sqrt{1-t} \, e^{is^{\rho}h(\zeta)} \mathcal{G}(\zeta;s)^{-1} & 1
\end{pmatrix}, & \mbox{if } \zeta \in \Sigma_{3} \cup \Sigma_{4}, \\
\begin{pmatrix}
1 & -\sqrt{1-t} \, e^{-is^{\rho}h(\zeta)} \mathcal{G}(\zeta;s) \\
\sqrt{1-t} \, e^{is^{\rho}h(\zeta)} \mathcal{G}(\zeta;s)^{-1} & t
\end{pmatrix}, & \mbox{if } \zeta \in \Sigma_{5}.
\end{cases}
\end{align*}
\item[(c)] As $\zeta \to \infty$, we have
\begin{align*}
\widehat{T}(\zeta) = I + \frac{U_{1}}{\zeta} + \bigO(\zeta^{-2}).
\end{align*}
As $\zeta \to b_{1}$ and as $\zeta \to b_{2}$, we have $\widehat{T}(\zeta) = \bigO(1)$.
\end{itemize}
We note that
\begin{align}
& \begin{pmatrix}
1 & -\sqrt{1-t} \, e^{-is^{\rho}h(\zeta)} \mathcal{G}(\zeta;s) \\
\sqrt{1-t} \, e^{is^{\rho}h(\zeta)} \mathcal{G}(\zeta;s)^{-1} & t
\end{pmatrix} \nonumber \\
& = \begin{pmatrix}
1 & 0 \\
\sqrt{1-t} \, e^{is^{\rho}h(\zeta)} \mathcal{G}(\zeta;s)^{-1} & 1
\end{pmatrix} \begin{pmatrix}
1 & -\sqrt{1-t} \, e^{-is^{\rho}h(\zeta)} \mathcal{G}(\zeta;s) \\
0 & 1
\end{pmatrix} \label{first factorization} \\
& = \begin{pmatrix}
1 & - \frac{\sqrt{1-t}}{t}e^{-is^{\rho}h(\zeta)} \mathcal{G}(\zeta;s) \\
0 & 1
\end{pmatrix} \begin{pmatrix}
\frac{1}{t} & 0 \\
0 & t
\end{pmatrix} \begin{pmatrix}
1 & 0 \\
\frac{\sqrt{1-t}}{t}e^{is^{\rho}h(\zeta)} \mathcal{G}(\zeta;s)^{-1} & 1 
\end{pmatrix}. \label{second factorization}
\end{align}
We have used the factorization \eqref{first factorization} in the transformation $U \mapsto \widehat{T}$ to collapse part of the contours on $\Sigma_{5}$. In the transformation $\widehat{T} \mapsto T$, we now use the other factorization \eqref{second factorization} to open lenses on the other side of $\Sigma_{5}$. 

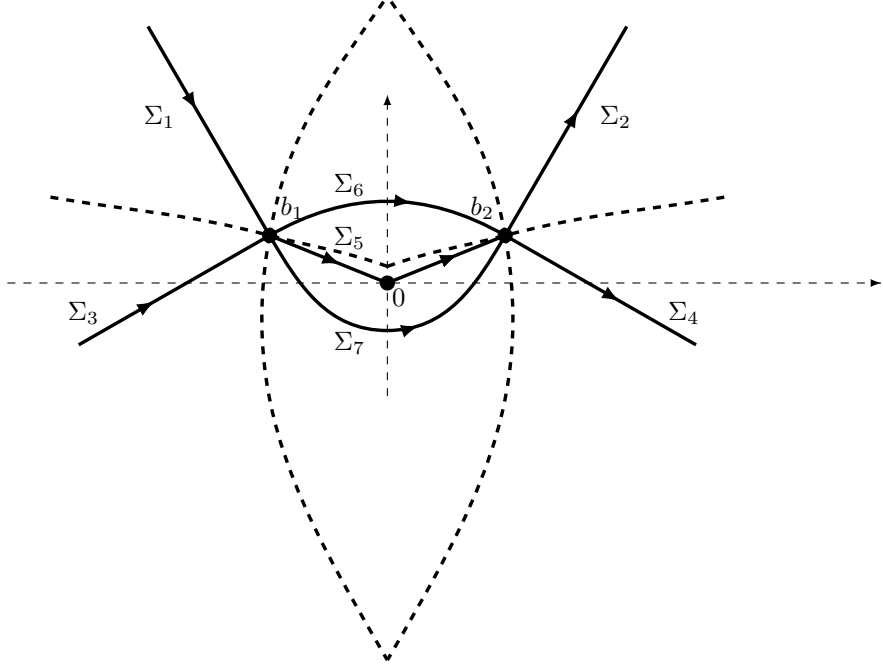
\begin{figure}
\begin{center}
\begin{tikzpicture}
\node at (0,0) {};
\fill (0,0) circle (0.1cm);
\node at (0.15,-0.2) {$0$};
\fill (1.55,0.63) circle (0.1cm);
\node at (1.25,1) {$b_{2}$};
\fill (-1.55,0.63) circle (0.1cm);
\node at (-1.25,1) {$b_{1}$};

\draw[->-=0.3, ->-=0.8,line width=0.45 mm,black] (-1.55,0.63)--(0,0)--(1.55,0.63);

\draw[dashed,->-=1,black] (0,-1.5) to [out=90, in=-90] (0,2.5);
\draw[dashed,->-=1,black] (-5,0) to [out=0, in=-180] (6.5,0);

(0,0.22) to [out=20, in=180+15] (1.55,0.63);
\draw[dashed,line width=0.45 mm,black] (0,0.22) to [out=20, in=180+15] (1.55,0.63) 
to [out=15, in=180+10] (4.5,1.15);
(0,0.22) to [out=180-20, in=-15] (-1.55,0.63);
\draw[dashed,line width=0.45 mm,black] (0,0.22) to [out=180-20, in=-15] (-1.55,0.63)
to [out=180-15, in=-10] (-4.5,1.15);
\draw[dashed,line width=0.45 mm,black] (1.55,0.63) to [out=100, in=180+125] (0,3.8);
\draw[dashed,line width=0.45 mm,black] (1.55,0.63) to [out=-80, in=180-120] (0,-5);
\draw[dashed,line width=0.45 mm,black] (-1.55,0.63) to [out=180-100, in=-125] (0,3.8);
\draw[dashed,line width=0.45 mm,black] (-1.55,0.63) to [out=-100, in=180-60] (0,-5);

\draw[->-=0.6,black,line width=0.45 mm] (1.55,0.63)--($(1.55,0.63)+(15+45:3.2)$);
\draw[->-=0.6,black,line width=0.45 mm] (1.55,0.63)--($(1.55,0.63)+(15-45:2.9)$);
\node at (3,2.2) {$\Sigma_{2}$};
\node at (3.9,-0.4) {$\Sigma_{4}$};

\draw[-<-=0.6,black,line width=0.45 mm] (-1.55,0.63)--($(-1.55,0.63)+(180-15-45:3.2)$);
\draw[-<-=0.6,black,line width=0.45 mm] (-1.55,0.63)--($(-1.55,0.63)+(180-15+45:2.9)$);
\node at (-3,2.2) {$\Sigma_{1}$};
\node at (-4,-0.4) {$\Sigma_{3}$};

\node at (-0.5,0.6) {$\Sigma_{5}$};
\node at (-0.5,1.3) {$\Sigma_{6}$};
\node at (-0.5,-0.8) {$\Sigma_{7}$};

\draw[->-=0.6,line width=0.45 mm,black]
(-1.55,0.63) to [out=-15+45, in=180+15-45] (1.55,0.63);
\draw[->-=0.6,line width=0.45 mm,black]
(-1.55,0.63) to [out=-15-45, in=180] (0,-0.63)
to [out=0, in=180+15+45] (1.55,0.63);
\end{tikzpicture}
\end{center}
\caption{\label{fig: contour for T}\textit{The jump contour for the RH problem for $T$.}}
\end{figure}

We define $T$ as follows:
\begin{align}\label{def of T}
T(\zeta) = s^{-\frac{c_{4}}{2}\sigma_{3}}e^{\frac{s^{\rho} \ell}{2}\sigma_{3}}\widehat{T}(\zeta)H(\zeta) e^{-\frac{s^{\rho} \ell}{2} \sigma_{3}}s^{\frac{c_{4}}{2}\sigma_{3}},
\end{align}
where
\begin{align*}
H(\zeta) = \begin{cases}
\begin{pmatrix}
1 & 0 \\
-\frac{\sqrt{1-t}}{t}e^{is^{\rho}h(\zeta)} \mathcal{G}(\zeta;s)^{-1} & 1 
\end{pmatrix}, & \mbox{if } \zeta \in \mbox{int}(\Sigma_{5}\cup \Sigma_{6}), \\
\begin{pmatrix}
1 & - \frac{\sqrt{1-t}}{t}e^{-is^{\rho}h(\zeta)} \mathcal{G}(\zeta;s) \\
0 & 1
\end{pmatrix}, & \mbox{if } \zeta \in \mbox{int}(\Sigma_{5}\cup \Sigma_{7}), \\
I, & \mbox{otherwise}.
\end{cases}
\end{align*}
Note that $e^{-is^{\rho}h(\zeta)} \mathcal{G}(\zeta;s)$ is analytic (in particular has no poles) in the lower half plane, while $e^{is^{\rho}h(\zeta)} \mathcal{G}(\zeta;s)^{-1}$ is analytic in the upper half plane, so that $H(\zeta)$ is analytic for $\zeta \in \mathbb{C}\setminus \big( \Sigma_{5} \cup \Sigma_{6} \cup \Sigma_{7} \big)$. 
Then $T$ satisfies the following RH problem:
\subsubsection*{RH problem for $T$}
\begin{itemize}
\item[(a)] $T : \mathbb{C}\setminus \bigcup_{j=1}^{7} \Sigma_{j} \to \mathbb{C}^{2 \times 2}$ is analytic. The contour $\bigcup_{j=1}^{7} \Sigma_{j}$ is shown in Figure \ref{fig: contour for T} and is chosen to be symmetric with respect to $i \mathbb{R}$.
\item[(b)] It satisfies the jumps $T_{+}(\zeta) = T_{-}(\zeta) J_{T}(\zeta)$ for $\zeta \in \bigcup_{j=1}^{7} \Sigma_{j}$, where
\begin{align}\label{jumps J_T}
J_{T}(\zeta) = \begin{cases}
\begin{pmatrix}
1 & -\sqrt{1-t}e^{-s^{\rho}(ih(\zeta)-\ell)} \widetilde{\mathcal{G}}(\zeta;s) \\
0 & 1
\end{pmatrix}, & \mbox{if } \zeta \in \Sigma_{1} \cup \Sigma_{2}, \\
\begin{pmatrix}
1 & 0 \\
\sqrt{1-t}e^{s^{\rho}(ih(\zeta)-\ell)} \widetilde{\mathcal{G}}(\zeta;s)^{-1} & 1
\end{pmatrix}, & \mbox{if } \zeta \in \Sigma_{3} \cup \Sigma_{4}, \\
\begin{pmatrix}
\frac{1}{t} & 0 \\
0 & t
\end{pmatrix}, & \mbox{if } \zeta \in \Sigma_{5}, \\
\begin{pmatrix}
1 & 0 \\
\frac{\sqrt{1-t}}{t} e^{s^{\rho}(ih(\zeta)-\ell)} \widetilde{\mathcal{G}}(\zeta;s)^{-1} & 1 
\end{pmatrix}, & \mbox{if } \zeta \in \Sigma_{6}, \\
\begin{pmatrix}
1 & - \frac{\sqrt{1-t}}{t}e^{-s^{\rho}(ih(\zeta)-\ell)} \widetilde{\mathcal{G}}(\zeta;s) \\
0 & 1
\end{pmatrix}, & \mbox{if } \zeta \in \Sigma_{7},
\end{cases}
\end{align}
where
\begin{align}\label{def of G tilde}
\widetilde{\mathcal{G}}(\zeta;s) = \mathcal{G}(\zeta;s)s^{-c_{4}}.
\end{align}
\item[(c)] As $\zeta \to \infty$, we have
\begin{align}\label{asymp of T at inf}
T(\zeta) = I + \frac{T_{1}}{\zeta} + \bigO(\zeta^{-2}),
\end{align}
where
\begin{align}\label{relation T1 to Y1}
T_{1} = s^{-\frac{c_{4}}{2}\sigma_{3}}e^{\frac{s^{\rho} \ell}{2}\sigma_{3}}U_{1}e^{-\frac{s^{\rho} \ell}{2}\sigma_{3}}s^{\frac{c_{4}}{2}\sigma_{3}} = \frac{1}{is^{\rho}} s^{-\frac{c_{4}}{2}\sigma_{3}}e^{\frac{s^{\rho} \ell}{2}\sigma_{3}}s^{\frac{\tau}{2}\sigma_{3}}Y_{1}s^{-\frac{\tau}{2}\sigma_{3}}e^{-\frac{s^{\rho} \ell}{2}\sigma_{3}}s^{\frac{c_{4}}{2}\sigma_{3}}.
\end{align}
As $\zeta \to b_{1}$ and as $\zeta \to b_{2}$, we have $T(\zeta) = \bigO(1)$.
\end{itemize}
\begin{remark}\label{remark: symmetry for T}
We choose the jump contour for $T$ to be symmetric with respect to $i \mathbb{R}$ for later use (it will make the analysis simpler). Using this symmetry, we show in a similar way as in Remark \ref{remark: symmetry for U} that $J_{T}(\zeta) = \overline{J_{T}(-\overline{\zeta})}$ for $\zeta \in \bigcup_{j=1}^{7} \Sigma_{j}$. By uniqueness of the solution to the RH problem for $T$, this implies the symmetry
\begin{align}\label{symmetry for T}
T(\zeta) = \overline{T(-\overline{\zeta})}, \qquad \zeta \in \mathbb{C} \setminus \bigcup_{j=1}^{7} \Sigma_{j}.
\end{align}
\end{remark}
By Lemma \ref{lemma: real part of h}, the jumps for $T$ tends to $I$ exponentially fast as $s \to + \infty$ on $\big( \bigcup_{j=1}^{7} \Sigma_{j} \big) \setminus \Sigma_{5}$, and this convergence is uniform outside neighborhoods of $b_{1}$ and $b_{2}$.

For convenience, we use the notation
\begin{align}\label{def of ell tilde}
\widetilde{\ell} := \im (ih(b_{2})) = - \im(i h(b_{1})) = (c_{1}+c_{2})\exp \left( - \frac{c_{1}+c_{2}+c_{3}}{c_{1}+c_{2}} \right) \cos \left( \frac{\pi}{2} \frac{c_{2}-c_{1}}{c_{1}+c_{2}} \right).
\end{align}
\subsection{Global parametrix}\label{subsection: global param}
In this subsection we construct the global parametrix $P^{(\infty)}$. We will show in Section \ref{subsection: small norm} that $P^{(\infty)}$ approximates $T$ outside of neighborhoods of $b_{1}$ and $b_{2}$.
\subsubsection*{RH problem for $P^{(\infty)}$}
\begin{itemize}
\item[(a)] $P^{(\infty)} : \mathbb{C}\setminus \Sigma_{5} \to \mathbb{C}^{2 \times 2}$ is analytic.
\item[(b)] It satisfies the jumps
\begin{align*}
& P^{(\infty)}_{+}(\zeta) = P^{(\infty)}_{-}(\zeta)\begin{pmatrix}
\frac{1}{t} & 0 \\
0 & t
\end{pmatrix}, & & \zeta \in \Sigma_{5}.
\end{align*}
\item[(c)] As $\zeta \to \infty$, we have
\begin{align}\label{asymp of Pinf at inf}
P^{(\infty)}(\zeta) = I + \frac{P^{(\infty)}_{1}}{\zeta} + \bigO(\zeta^{-2}).
\end{align}
\item[(d)] As $\zeta$ tends to $b_{1}$ or $b_{2}$, $P^{(\infty)}(\zeta)$ remains bounded.
\end{itemize}
Conditions (a)-(c) for the RH problem for $P^{(\infty)}$ are obtained from the RH problem for $T$ by ignoring the jumps on $\big( \bigcup_{j=1}^{7} \Sigma_{j} \big) \setminus \Sigma_{5}$. Condition (d) has been added to ensure uniqueness of the solution of the RH problem for $P^{(\infty)}$. This solution can be easily obtained by using Cauchy's formula and is given by
\begin{align}\label{def of Pinf}
P^{(\infty)}(\zeta) = D(\zeta)^{-\sigma_{3}},
\end{align}
where
\begin{align}\label{def of nu and D}
D(\zeta) = \exp \left( i\nu\int_{\Sigma_{5}} \frac{d\xi}{\xi-\zeta} \right) = \exp \left( i\nu \log \left[ \frac{\zeta - b_{2}}{\zeta - b_{1}} \right] \right), \qquad \nu := - \frac{1}{2\pi} \log t \in \mathbb{R},
\end{align}
where the branch for the log is taken along $\Sigma_{5}$. The function $D$ satisfies 
\begin{align*}
& D_{+}(\zeta) = D_{-}(\zeta)t, & & \zeta \in \Sigma_{5}, \\
& D(\zeta) = 1+ \frac{D_{1}}{\zeta} + \bigO(\zeta^{-2}), & & \mbox{as } \zeta \to \infty,
\end{align*}
where $D_{1} = -i \nu (b_{2}-b_{1}) = -2i\nu \re b_{2}$. From \eqref{asymp of Pinf at inf} and \eqref{def of Pinf}, we obtain
\begin{align}\label{Pinf1}
P_{1}^{(\infty)} = - D_{1} \sigma_{3}.
\end{align}
We will also need asymptotics for $P^{(\infty)}(\zeta)$ as $\zeta \to b_{2}$. From \eqref{def of nu and D}, as $\zeta \to b_{2}$ we have
\begin{align*}
D(\zeta) = \left( \frac{\zeta -b_{2}}{b_{2}-b_{1}} \right)^{i \nu}\left( 1-i \nu \frac{\zeta - b_{2}}{b_{2}-b_{1}} + \bigO((\zeta - b_{2})^{2}) \right),
\end{align*}
which implies by \eqref{def of Pinf} that
\begin{align}\label{asymp Pinf near b2}
P^{(\infty)}(\zeta) =  \left( \frac{\zeta -b_{2}}{b_{2}-b_{1}} \right)^{-i \nu \sigma_{3}}\left(I+i \nu \frac{\zeta - b_{2}}{b_{2}-b_{1}} \sigma_{3} + \bigO((\zeta - b_{2})^{2}) \right), \qquad \mbox{as } \zeta \to b_{2}.
\end{align}
it is also direct to verify from \eqref{def of Pinf} and \eqref{def of nu and D} that $P^{(\infty)}$ satisfies the symmetry relation
\begin{align}\label{symmetry for Pinf}
P^{(\infty)}(\zeta) = \overline{P^{(\infty)}(-\overline{\zeta})}, \qquad \zeta \in \mathbb{C}\setminus \Sigma_{5}.
\end{align}
\subsection{Local parametrix at $b_{2}$}\label{subsection: local param at b2}
We construct the local parametrix $P^{(b_{2})}$ in a small disk $\mathcal{D}_{b_{2}}$ around $b_{2}$ with radius independent of $s$. We require $P^{(b_{2})}$ to satisfy the same jumps as $T$ inside $\mathcal{D}_{b_{2}}$, to remain bounded as $\zeta \to b_{2}$, and to match with $P^{(\infty)}$ on the boundary of $\mathcal{D}_{b_{2}}$, in the sense that
\begin{align*}
P^{(b_{2})}(\zeta) = (I+o(1))P^{(\infty)}(\zeta), \qquad \mbox{as } s \to + \infty,
\end{align*}
uniformly for $\zeta \in \partial\mathcal{D}_{b_{2}}$. The solution can be constructed in terms of the solution $\Phi_{\mathrm{PC}}$ to the Parabolic Cylinder model RH problem presented in Appendix \ref{appendix:PC}. This model RH problem depends on a parameter $q$; in our case we need to choose $q = \sqrt{1-t}$. Let us define
\begin{align}\label{def of fb2}
f(\zeta) = \sqrt{-2(h(\zeta)-h(b_{2}))}.
\end{align}
This is a conformal map from $\mathcal{D}_{b_{2}}$ to a neighborhood of $0$ satisfying $f(b_2)=0$ and
\begin{equation}\label{cb2 and cb2p2p}
f'(b_2)= \sqrt{\frac{c_{1}+c_{2}}{b_{2}}} = \frac{\sqrt{c_{1}+c_{2}}}{\exp\Big(\hspace{-0.1cm}-\hspace{-0.05cm}\frac{c_{1}+c_{2}+c_{3}}{2(c_{1}+c_{2})}\Big) \exp \Big(i \frac{\pi}{4}\frac{c_{2}-c_{1}}{c_{1}+c_{2}}\Big)} \quad \mbox{ and } \quad f''(b_2)=- \frac{1}{3 b_{2}}f'(b_2) .
\end{equation}
In small neighborhoods of $\mathcal{D}_{b_{2}}$ and $\mathcal{D}_{b_{1}}$, we slightly deform the contour $\bigcup_{j=1}^{7} \Sigma_{j}$ such that 
it remains symmetric with respect to $i \mathbb{R}$ and such that it satisfies
\begin{align}\label{Sigma j is well-mapped by fb2}
f\bigg(\bigcup_{j=1}^{7} \Sigma_{j} \cap \mathcal{D}_{b_{2}}\bigg) \subset \Sigma_{\mathrm{PC}},
\end{align}
where $\Sigma_{\mathrm{PC}}$ is shown in Figure \ref{fig: contour for PC}. The local parametrix is given by
\begin{align}\label{def of Pb2}
P^{({2})}(\zeta;s) = E(\zeta;s) \Phi_{\mathrm{PC}}(s^{\frac{\rho}{2}}f(\zeta);\sqrt{1-t})e^{ \frac{s^{\rho} }{2}(ih(\zeta)-\ell) \sigma_{3}} \widetilde{\mathcal{G}}(\zeta;s)^{-\frac{\sigma_{3}}{2}},
\end{align}
where $E$ is analytic in $\mathcal{D}_{b_{2}}$ and given by
\begin{equation}\label{def of Eb2}
E(\zeta;s) = P^{(\infty)}(\zeta)\widetilde{\mathcal{G}}(\zeta;s)^{\frac{\sigma_{3}}{2}}e^{-\frac{s^{\rho}}{2}i \widetilde{\ell}\sigma_{3}}\big( s^{\frac{\rho}{2}}f(\zeta) \big)^{i\nu \sigma_{3}},
\end{equation}
where $\nu = \nu(t) \in \mathbb{R}$ is given by \eqref{def of nu and D} and the branch cut for $\big( s^{\frac{\rho}{2}}f(\zeta) \big)^{i\nu \sigma_{3}}$ is taken along $\Sigma_{5} \cap \mathcal{D}_{b_{2}}$. Note that $\widetilde{\mathcal{G}}(\zeta;s)$ depends on $s$, but by \eqref{asymp for G} and \eqref{def of G tilde} it is bounded as $s \to + \infty$ uniformly for $\zeta \in \mathcal{D}_{b_{2}}$. Since $\nu \in \mathbb{R}$ and $\widetilde{\ell} \in \mathbb{R}$ (see \eqref{def of ell tilde}), we thus have $E(\zeta;s) = \bigO(1)$ as $s \to + \infty$, uniformly for $\zeta \in \mathcal{D}_{b_{2}}$. Using \eqref{asymp Pinf near b2} and \eqref{cb2 and cb2p2p}, we infer that
\begin{align}
& E(\zeta;s) = \alpha(s)^{\sigma_3}\Big( I + \beta(s)\sigma_{3}(\zeta - b_{2}) + \bigO\big( (\zeta-b_{2})^{2} \big) \Big), \qquad \zeta \to b_{2}, \label{Eb2 at b2 asymptotics} \\
& \alpha(s) = \Big[ (b_{2}-b_{1})f'(b_{2})s^{\frac{\rho}{2}} \Big]^{i\nu} \widetilde{\mathcal{G}}(b_{2};s)^{\frac{1}{2}} e^{- \frac{i s^{\rho}}{2}\widetilde{\ell}}, \label{Eb2 at b2} \\
& \beta(s) = \frac{i\nu}{b_{2}-b_{1}} + \frac{1}{2} (\log \widetilde{\mathcal{G}})'(b_{2};s) - \frac{i\nu}{6b_2} . \label{eb2 at b2}
\end{align}
As $s \to + \infty$, for any $N \in \mathbb{N}$, we have
\begin{align}\label{matching param b2}
P^{(b_{2})}(\zeta)P^{(\infty)}(\zeta)^{-1} = I + E(\zeta;s)\bigg(  \sum_{j=1}^{N}\frac{\Phi_{\mathrm{PC},j}}{s^{\frac{j\rho}{2}}f(\zeta)^{j}} \bigg) E(\zeta;s)^{-1} +  \bigO(s^{-\frac{(N+1)\rho}{2}}),
\end{align}
uniformly for $\zeta \in \partial \mathcal{D}_{b_{2}}$, where the matrices $\Phi_{\mathrm{PC},1}$ and $\Phi_{\mathrm{PC},2}$ are given by \eqref{PhiPC1 and beta def}. In particular, the matrix $\Phi_{\mathrm{PC},1}$ is expressed in terms of the quantities $\beta_{12}$ and $\beta_{21}$ defined in \eqref{def of beta 12 and beta 21}. Furthermore, the matrices $\Phi_{\mathrm{PC},2k}$ are diagonal for every $k \geq 1$ and the matrices $\Phi_{\mathrm{PC},2k-1}$ are off-diagonal for every $k \geq 1$, see again \eqref{PhiPC1 and beta def}. 

\medskip

We need to expand $E(\zeta;s)$ as $s\to\infty$. By the expansion \eqref{asymp for G} of $\mathcal G$ and the definition \eqref{def of G tilde} of $\widetilde{\mathcal G}$, we obtain
\begin{align}\label{asymp for ln G tilde}
\log \widetilde{\mathcal{G}}(\zeta;s) = c_{5} \log(i\zeta) + c_{6} \log(-i\zeta) + c_{7} + \frac{c_{8}}{is^{\rho}\zeta} + \bigO ( s^{-2\rho} ) \qquad \mbox{as } s \to + \infty,
\end{align}
uniformly for $\zeta \in \mathcal{D}_{b_{2}}$, where the error term can be expanded in a full asymptotic series in integer powers of $s^{-\rho}$.
We deduce from this that
\begin{align}
& \widetilde{\mathcal{G}}(b_{2};s) = (ib_{2})^{c_{5}}(-ib_{2})^{c_{6}}e^{c_{7}}\Big(1+\frac{c_{8}}{is^{\rho}b_{2}} + \bigO(s^{-2\rho}) \Big), \label{asymp of G at b2} \\
& (\log \widetilde{\mathcal{G}})'(b_{2};s) = \frac{c_{5}+c_{6}}{b_{2}} - \frac{c_{8}}{i s^{\rho} b_{2}^{2}} + \bigO(s^{-2\rho}), \label{asymp of log G prim at b2}
\end{align}
as $s\to\infty$,
and by \eqref{def of Eb2}, we can write
\begin{align}\label{expansion of Eb2 in s}
E(\zeta;s) = \sum_{j=0}^{N} E_j(\zeta;s)s^{-\rho j} + \bigO(s^{-(N+1)\rho}), \qquad \mbox{as } s \to + \infty,
\end{align}
for any $N \in \mathbb{N}$, uniformly for $\zeta \in \mathcal{D}_{b_{2}}$, and where the diagonal matrices $E_j(\zeta;s)$ depend on $s$ but are bounded; in particular
\begin{align}
& E_0(b_{2};s) = \Big[ (b_{2}-b_{1})f'(b_2)s^{\frac{\rho}{2}} \Big]^{i\nu \sigma_{3}} \big((ib_{2})^{c_{5}}(-ib_{2})^{c_{6}}e^{c_{7}}\big)^{\frac{\sigma_{3}}{2}} e^{- \frac{s^{\rho}}{2}i \widetilde{\ell}\sigma_{3}}, \label{Eb2p0p at b2} \\
& E_0'(b_{2};s) = \beta_0(s)E_{{0}}(b_{2};s)\sigma_{3}, \label{Eb2'p0p at b2} \\
& \beta_0(s) = \frac{i\nu}{b_{2}-b_{1}} +  \frac{c_{5}+c_{6}}{2b_{2}} -\frac{i\nu}{6b_2}. \label{eb2p0p at b2}
\end{align}
For later use, we note that this implies
\begin{align}
\mathfrak{e}(s) & := \frac{E_0(b_{2};s)_{11}}{\overline{E_0(b_{2};s)_{11}}} = \exp\Big( 2i \arg(E_0(b_{2};s)_{11}) \Big) \nonumber \\
& = \Big[ (b_{2}-b_{1})|f'({b_{2}})|s^{\frac{\rho}{2}} \Big]^{2i\nu } \exp \Big( i c_{5} \arg(ib_{2}) + i c_{6} \arg(-ib_{2}) \Big) e^{-i \widetilde{\ell} s^{\rho}}. \label{def of efrak}
\end{align}
\subsection{Local parametrix at $b_{1}$}\label{subsection: local param at b1}
We construct the local parametrix $P^{(b_{1})}$ in a small disk $\mathcal{D}_{b_{1}}$ around $b_{1}$ in a similar way as we defined $P^{(b_{2})}$ in $\mathcal{D}_{b_{2}}$. More precisely, we require $P^{(b_{1})}$ to satisfy the same jumps as $T$ inside $\mathcal{D}_{b_{1}}$, to remain bounded as $\zeta \to b_{1}$, and to satisfy the matching condition
\begin{align*}
P^{(b_{1})}(\zeta) = (I+o(1))P^{(\infty)}(\zeta), \qquad \mbox{as } s \to + \infty,
\end{align*}
uniformly for $\zeta \in \partial\mathcal{D}_{b_{1}}$. It is possible to construct $P^{(b_{1})}(\zeta)$ in a similar way as $P^{(b_{2})}(\zeta)$ in terms of parabolic cylinder functions. To avoid unnecessary analysis and computations, we choose $\mathcal{D}_{b_{1}}=-\overline{\mathcal{D}_{b_{2}}}$, and we rely on the symmetry $J_{T}(\zeta) = \overline{J_{T}(-\overline{\zeta})}$ for $\zeta \in \bigcup_{j=1}^{7} \Sigma_{j}$ (see Remark \ref{remark: symmetry for T}) to conclude directly that the function
\begin{align}\label{symmetry local param}
P^{(b_{1})}(\zeta) = \overline{P^{(b_{2})}(-\overline{\zeta})}, \qquad \zeta \in \mathcal{D}_{b_{1}} \setminus \bigcup_{j=1}^{7} \Sigma_{j}
\end{align}
satisfies the required conditions for the local parametrix.

\subsection{Small norm RH problem}\label{subsection: small norm}
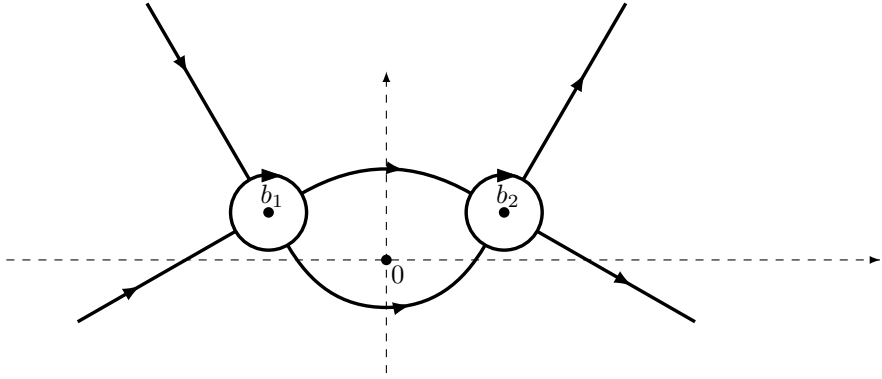
\begin{figure}
\begin{center}
\begin{tikzpicture}
\node at (0,0) {};
\fill (0,0) circle (0.07cm);
\node at (0.15,-0.2) {$0$};
\fill (1.55,0.63) circle (0.07cm);
\draw[line width=0.45 mm] (1.55,0.63) circle (0.5cm);
\draw[black,arrows={-Triangle[length=0.27cm,width=0.18cm]}]
($(1.55,0.63)+(70:0.52)$) --  ++(+0.0001,0);
\node at (1.6,0.85) {$b_{2}$};
\fill (-1.55,0.63) circle (0.07cm);
\draw[line width=0.45 mm] (-1.55,0.63) circle (0.5cm);
\draw[black,arrows={-Triangle[length=0.27cm,width=0.18cm]}]
($(-1.55,0.63)+(70:0.52)$) --  ++(+0.0001,0);
\node at (-1.5,0.85) {$b_{1}$};

\draw[dashed,->-=1,black] (0,-1.5) to [out=90, in=-90] (0,2.5);
\draw[dashed,->-=1,black] (-5,0) to [out=0, in=-180] (6.5,0);


\draw[->-=0.6,black,line width=0.45 mm] ($(1.55,0.63)+(15+45:0.5)$)--($(1.55,0.63)+(15+45:3.2)$);
\draw[->-=0.6,black,line width=0.45 mm] ($(1.55,0.63)+(15-45:0.5)$)--($(1.55,0.63)+(15-45:2.9)$);

\draw[-<-=0.6,black,line width=0.45 mm] ($(-1.55,0.63)+(180-15-45:0.5)$)--($(-1.55,0.63)+(180-15-45:3.2)$);
\draw[-<-=0.6,black,line width=0.45 mm] ($(-1.55,0.63)+(180-15+45:0.5)$)--($(-1.55,0.63)+(180-15+45:2.9)$);

\draw[->-=0.6,line width=0.45 mm,black]
($(-1.55,0.63)+(-15+45:0.5)$) to [out=-15+45, in=180+15-45] ($(1.55,0.63)+(15-45+180:0.5)$);
\draw[->-=0.6,line width=0.45 mm,black]
($(-1.55,0.63)+(-15-45:0.5)$) to [out=-15-45, in=180] (0,-0.63)
to [out=0, in=180+15+45] ($(1.55,0.63)+(15+45+180:0.5)$);
\end{tikzpicture}
\end{center}
\caption{\label{fig: contour for R}\textit{The jump contour $\Sigma_{R}$ in the RH problem for $R$.}}
\end{figure}
In this section we show that, as $s$ becomes large, $P^{(\infty)}(z)$ approximates $T(z)$ for $z \in \mathbb{C}\setminus \cup_{j=1}^2\mathcal{D}_{b_j}$ and $P^{(b_j)}(z)$ approximates $T(z)$ for $z \in \mathcal{D}_{b_j}$, $j=1,2$.
We define
\begin{align}\label{def of R}
R(\zeta) = \begin{cases}
T(\zeta)P^{(\infty)}(\zeta)^{-1}, & \mbox{if } \zeta \in \mathbb{C}\setminus (\overline{\mathcal{D}_{b_{1}} \cup \mathcal{D}_{b_{2}}}), \\
T(\zeta)P^{(b_{1})}(\zeta)^{-1}, & \mbox{if } \zeta \in \mathcal{D}_{b_{1}}, \\
T(\zeta)P^{(b_{2})}(\zeta)^{-1}, & \mbox{if } \zeta \in \mathcal{D}_{b_{2}}.
\end{cases}
\end{align}
Since $P^{(b_j)}$, $j=1,2$, have the exact same jumps as $T$ inside the disks, $R$ is analytic in $\cup_{j=1}^{2}\mathcal{D}_{b_j}\setminus\{b_{j}\}$. Furthermore, since $S(z)$ and $P^{(b_j)}(z)^{-1}$ remain bounded as $z \to b_{j}$, $j=1,2$, we conclude that $R(z)$ is also bounded as $z \to b_{j}$, $j=1,2$. Thus the singularities of $R$ at $b_{1}$ and $b_{2}$ are removable and $R$ is analytic in the entire open disks. $R$ satisfies the following RH problem.
\subsubsection*{RH problem for $R$}
\begin{itemize}
\item[(a)] $R : \mathbb{C}\setminus  \Sigma_{R} \to \mathbb{C}^{2 \times 2}$ is analytic, where
\begin{align*}
\Sigma_{R} = \partial \mathcal{D}_{b_{1}} \cup \partial \mathcal{D}_{b_{2}} \cup \bigcup_{j=1}^{7} \Sigma_{j} \setminus (\mathcal{D}_{b_{1}} \cup \mathcal{D}_{b_{2}} \cup \Sigma_{5}).
\end{align*}
The contour $\Sigma_{R}$ is oriented as shown in Figure \ref{fig: contour for R}. In particular, we orient the circles $\partial \mathcal{D}_{b_{1}}$ and $\partial \mathcal{D}_{b_{2}}$ in the clockwise direction.
\item[(b)] For $\zeta \in \Sigma_{R}$, $R$ satisfies the jumps $R_{+}(\zeta) = R_{-}(\zeta)J_{R}(\zeta)$, where
\begin{align*}
& J_{R}(\zeta) = P^{(b_{1})}(\zeta)P^{(\infty)}(\zeta)^{-1}, & & \zeta \in \partial \mathcal{D}_{b_{1}}, \\
& J_{R}(\zeta) = P^{(b_{2})}(\zeta)P^{(\infty)}(\zeta)^{-1}, & & \zeta \in \partial \mathcal{D}_{b_{2}}, \\
& J_{R}(\zeta) = P^{(\infty)}(\zeta)J_{T}(\zeta)P^{(\infty)}(\zeta)^{-1}, & & \zeta \in \Sigma_{R}\setminus (\partial \mathcal{D}_{b_{1}} \cup \partial \mathcal{D}_{b_{2}}).
\end{align*}
\item[(c)] $R(\zeta)$ remains bounded as $\zeta$ tends to the points of self-intersection of $\Sigma_{R}$. 

As $\zeta \to \infty$, there exists $R_1=R_1(s)$ such that
\begin{align}\label{asymp of R at inf}
R(\zeta) = I + \frac{R_{1}}{\zeta} + \bigO(\zeta^{-2}).
\end{align}
\end{itemize}
\begin{remark}\label{remark: symmetry for R}
The contour $\Sigma_{R}$ is symmetric with respect to $i \mathbb{R}$. Furthermore, by \eqref{symmetry for Pinf} and \eqref{symmetry local param}, the jumps $J_{R}$ satisfy the symmetry relation $J_{R}(\zeta) = \overline{J_{R}(-\overline{\zeta})}$ for $\zeta \in \Sigma_{R}$. Hence, by uniqueness of the solution to the RH problem for $R$, we conclude that
\begin{align}\label{symmetry for R}
R(\zeta) = \overline{R(-\overline{\zeta})}, \qquad \zeta \in \mathbb{C}\setminus \Sigma_{R}.
\end{align}
\end{remark}
From Lemma \ref{lemma: real part of h} and the fact that $P^{(\infty)}$ is independent of $s$ and uniformly bounded outside $\mathcal{D}_{b_{1}} \cup \mathcal{D}_{b_{2}}$, we have
\begin{align}\label{estimate jumps exp small}
J_{R}(\zeta) = I+ \bigO(e^{-c s^{\rho} |\zeta|}), \qquad \mbox{as } s \to +\infty
\end{align}
uniformly for $\zeta \in \Sigma_{R}\setminus (\partial \mathcal{D}_{b_{1}} \cup \partial \mathcal{D}_{b_{2}})$, and uniformly for $t$ in compact subsets of $(0,\infty)$. By substituting the expansion \eqref{expansion of Eb2 in s} in \eqref{matching param b2}, we infer that, for any $N \in \mathbb{N}$, $J_R$ has an expansion in the form
\begin{align}\label{estimate jumps on the disks}
J_{R}(\zeta)=J_{R}(\zeta;s) = I+ \sum_{j=1}^{N}J_{R}^{(j)}(\zeta;s)s^{-\frac{j\rho}{2}}+\bigO(s^{-\frac{(N+1)\rho}{2}}), \qquad \mbox{as } s \to +\infty,
\end{align}
where all coefficients $J_{R}^{(j)}(\zeta;s)$ satisfy the symmetry $J_{R}^{(j)}(\zeta;s) = \overline{J_{R}^{(j)}(-\overline{\zeta};s)}$ and are bounded as $s \to + \infty$, uniformly for $\zeta \in \partial \mathcal{D}_{b_{1}} \cup \partial \mathcal{D}_{b_{2}}$ and for $t$ in compact subsets of $(0,\infty)$. The first two coefficients $J_{R}^{(j)}(\zeta;s)$, for $j=1,2$ are given by
\begin{equation}\label{jumps for R first order on the boundary Db2}
J_R^{(j)}(\zeta;s)=E_0(\zeta;s)\frac{\Phi_{\mathrm{PC},j}}{f(\zeta)^{j}} E_0(\zeta;s)^{-1}, \quad j=1,2.
\end{equation}

\medskip

The jump relation for $R$ can also be written in the additive form $R_+=R_- + R_- J_R$, and together with the asymptotis for $R$, this implies the integral equation
\begin{equation}\label{eq:integraleq}R(\zeta)=R(\zeta;s)=I+\frac{1}{2\pi i}\int_{\Sigma_R}\frac{R_-(\xi;s)J_R(\xi;s)}{\xi-\zeta}d\xi.\end{equation}
We conclude from \eqref{estimate jumps exp small} and \eqref{estimate jumps on the disks} that $R$ satisfies a small norm RH problem as $s\to +\infty$, and by standard theory \cite{Deiftetal}, it follows that $R$ exists for sufficiently large $s$. Moreover, substituting \eqref{estimate jumps exp small} and \eqref{estimate jumps on the disks} in \eqref{eq:integraleq} and expanding as $s\to +\infty$, we obtain
\begin{align}
& R(\zeta;s) = I + \sum_{j=1}^{N} \frac{R^{(j)}(\zeta;s)}{s^{\frac{j\rho}{2}}} + \bigO(s^{-\frac{(N+1)\rho}{2}}), \qquad \mbox{as } s \to + \infty, \label{estimate for R} \\
& R'(\zeta;s) = \sum_{j=1}^{N} \frac{R^{(j)\prime}(\zeta;s)}{s^{\frac{j\rho}{2}}} + \bigO(s^{-\frac{(N+1)\rho}{2}}), \qquad \mbox{as } s \to + \infty, \nonumber
\end{align}
uniformly for $\zeta \in \mathbb{C}\setminus \Sigma_{R}$, and uniformly for $t$ in compact subsets of $(0,\infty)$. All the coefficients $R^{(j)}$ can in principle be computed iteratively.
In particular, 
\begin{align}\label{Rp1p integral form}
R^{(1)}(\zeta;s) = \frac{1}{2\pi i} \int_{\partial \mathcal{D}_{b_{1}}} \frac{J_{R}^{(1)}(\xi;s)}{\xi-\zeta}d\xi + \frac{1}{2\pi i} \int_{\partial \mathcal{D}_{b_{2}}} \frac{J_{R}^{(1)}(\xi;s)}{\xi-\zeta}d\xi,
\end{align}
and
\begin{multline}\label{Rp2p integral representation}
R^{(2)}(\zeta;s) = \frac{1}{2\pi i} \int_{\partial \mathcal{D}_{b_{1}}} \frac{R^{(1)}_{-}(\xi;s)J_{R}^{(1)}(\xi;s)+J_{R}^{(2)}(\xi;s)}{\xi-\zeta}d\xi \\+ \frac{1}{2\pi i} \int_{\partial \mathcal{D}_{b_{2}}} \frac{R^{(1)}_{-}(\xi;s)J_{R}^{(1)}(\xi;s)+J_{R}^{(2)}(\xi;s)}{\xi-\zeta}d\xi.
\end{multline}
where we recall that $\partial \mathcal{D}_{b_{1}}$ and $\partial \mathcal{D}_{b_{2}}$ are oriented clockwise.

\medskip

In the rest of this section, we evaluate $R^{(1)}(\zeta;s)$ and $R^{(2)}(\zeta;s)$ explicitly for $\zeta \in \mathbb{C}\setminus (\mathcal{D}_{b_{1}} \cup \mathcal{D}_{b_{2}})$, and we prove that $R^{(k)}(\zeta;s)$ can be chosen diagonal for $k$ even and off-diagonal for $k$ odd.


\medskip

The expression \eqref{jumps for R first order on the boundary Db2} for $J_R^{(1)}$ can be analytically continued from $\partial \mathcal{D}_{b_{2}}$ to the punctured disk $\mathcal{D}_{b_{2}} \setminus \{b_{2}\}$, and we note that $J_{R}^{(1)}(\zeta;s)$ has a simple pole at $\zeta = b_{2}$. Therefore, for $\zeta$ outside the disks, a residue calculation gives
\begin{align}\label{res JRp1p at b2}
\frac{1}{2\pi i} \int_{\partial \mathcal{D}_{b_{2}}} \frac{J_{R}^{(1)}(\xi;s)}{\xi-\zeta}d\xi = \frac{A^{(1)}(s)}{\zeta - b_{2}}, \quad \mbox{with} \quad
A^{(1)}(s) = \mbox{Res} \Big( J_{R}^{(1)}(\xi;s),\xi = b_{2} \Big).
\end{align}
To evaluate the first integral that appears at the right-hand-side of \eqref{Rp1p integral form}, we appeal to the symmetries of Remark \ref{remark: symmetry for R} to write
\begin{align}\label{res JRp1p at b1}
\frac{1}{2\pi i} \int_{\partial \mathcal{D}_{b_{1}}} \frac{J_{R}^{(1)}(\xi)}{\xi-\zeta}d\xi = \overline{\frac{1}{2\pi i} \int_{\partial \mathcal{D}_{b_{2}}} \frac{J_{R}^{(1)}(\xi)}{\xi+\overline{\zeta}}d\xi} = \frac{-\overline{A^{(1)}(s)}}{\zeta - b_{1}}.
\end{align}
Therefore, it remains to evaluate $A^{(1)}(s)=\mbox{Res} \Big( J_{R}^{(1)}(\xi;s),\xi = b_{2} \Big)$. From \eqref{Eb2p0p at b2}, \eqref{jumps for R first order on the boundary Db2}, and \eqref{PhiPC1 and beta def}, we immediately obtain that
\begin{align}\label{explicit expression for Ap1p}
A^{(1)}(s) = \frac{1}{f'(b_{2})}E_0(b_{2};s) \Phi_{\mathrm{PC},1}E_0(b_{2};s)^{-1} = \frac{1}{f'(b_2)}\begin{pmatrix}
0 & \beta_{12} E_0(b_{2};s)_{11}^{2} \\
\beta_{21} E_0(b_{2};s)_{11}^{-2} & 0
\end{pmatrix},
\end{align}
which is an off-diagonal matrix, and $E_0(b_{2};s)$ has been explicitly evaluated in \eqref{Eb2p0p at b2}. Combining \eqref{Rp1p integral form} with \eqref{res JRp1p at b2} and \eqref{res JRp1p at b1}, we obtain 
\begin{align}\label{Rp1p explicit}
R^{(1)}(\zeta;s) = \frac{A^{(1)}}{\zeta - b_{2}}-\frac{\overline{A^{(1)}}}{\zeta - b_{1}}, \qquad \mbox{for } \zeta \in \mathbb{C}\setminus (\mathcal{D}_{b_{1}} \cup \mathcal{D}_{b_{2}}),
\end{align}
where $A^{(1)}$ is given by \eqref{explicit expression for Ap1p}.


\medskip

For the computation of $R^{(2)}$, we recall that $J_{R}^{(1)}$ and $J_{R}^{(2)}$ are given by \eqref{jumps for R first order on the boundary Db2}, and we note that $J_{R}^{(2)}$ can be simplified as follows
\begin{align}\label{jumps for R second order on the boundary Db2}
J_{R}^{(2)}(\zeta) = \frac{1}{f(\zeta)^{2}}
E_0(\zeta;s) \Phi_{\mathrm{PC},2} E_0(\zeta;s)^{-1} = \frac{1}{f(\zeta)^{2}}
\Phi_{\mathrm{PC},2}, \qquad \zeta \in \partial \mathcal{D}_{b_{2}},
\end{align}
where we have used that both $E_0(\zeta;s)$ and $\Phi_{\mathrm{PC},2}$ are diagonal matrices. We note that $J_{R}^{(2)}$ can also be analytically continued from $\partial \mathcal{D}_{b_{2}}$ to the punctured disk $\mathcal{D}_{b_{2}} \setminus \{b_{2}\}$. Let us start by evaluating the integral over $\partial \mathcal{D}_{b_{2}}$ which appears at the right-hand-side of \eqref{Rp2p integral representation}. For $\zeta \in \mathbb{C}\setminus (\partial \mathcal{D}_{b_{1}} \cup \partial \mathcal{D}_{b_{2}})$, since $R_{-}^{(1)}$ is analytic on $\mathcal{D}_{b_{1}} \cup \mathcal{D}_{b_{2}}$, and since $J_{R}^{(j)}$ admits a pole of order $j$ at $b_{2}$, $j=1,2$, we have
\begin{align*}
& \frac{1}{2\pi i} \int_{\partial \mathcal{D}_{b_{2}}} \frac{R^{(1)}_{-}(\xi;s)J_{R}^{(1)}(\xi;s)+J_{R}^{(2)}(\xi;s)}{\xi-\zeta}d\xi = \frac{A^{(2)}(s)}{\zeta-b_{2}} + \frac{B^{(2)}(s)}{(\zeta-b_{2})^{2}}, \\
& A^{(2)}(s) = R^{(1)}(b_{2};s)A^{(1)}(s) + \mbox{Res} \Big( J_{R}^{(2)}(\xi;s),\xi = b_{2} \Big), \\
& B^{(2)}(s) = \mbox{Res} \Big( (\xi-b_{2})J_{R}^{(2)}(\xi;s),\xi = b_{2} \Big).
\end{align*}
We again appeal to the symmetry $\zeta \mapsto - \overline{\zeta}$ of Remark \ref{remark: symmetry for R} to evaluate the integral over $\partial \mathcal{D}_{b_{1}}$:
\begin{align*}
\frac{1}{2\pi i} \int_{\partial \mathcal{D}_{b_{1}}} \frac{R^{(1)}_{-}(\xi;s)J_{R}^{(1)}(\xi;s)+J_{R}^{(2)}(\xi;s)}{\xi-\zeta}d\xi & = \overline{\frac{1}{2\pi i} \int_{\partial \mathcal{D}_{b_{2}}} \frac{R^{(1)}_{-}(\xi;s)J_{R}^{(1)}(\xi;s)+J_{R}^{(2)}(\xi;s)}{\xi+\overline{\zeta}}d\xi} \\
& = -\frac{\overline{A^{(2)}(s)}}{\zeta-b_{1}} + \frac{\overline{B^{(2)}(s)}}{(\zeta-b_{1})^{2}}.
\end{align*}
To evaluate $A^{(2)}$ and $B^{(2)}$ explicitly, it remains to compute $R^{(1)}(b_{2};s)$, and the two residues 
\begin{align}\label{residue needed for Rp2p}
\mbox{Res} \Big( J_{R}^{(2)}(\xi;s),\xi = b_{2} \Big) \quad \mbox{ and } \quad \mbox{Res} \Big( (\xi-b_{2})J_{R}^{(2)}(\xi;s),\xi = b_{2} \Big).
\end{align}
It is fairly easy to compute the residues \eqref{residue needed for Rp2p} from the expression \eqref{jumps for R first order on the boundary Db2} for $J_{R}^{(2)}(\zeta;s)$. We obtain
\begin{align*}
& B^{(2)}(s) = \mbox{Res} \Big( (\xi-b_{2})J_{R}^{(2)}(\xi;s),\xi = b_{2} \Big) = \frac{1}{f'(b_2)^{2}} \begin{pmatrix}
\frac{(1+i\nu)\nu}{2} & 0 \\
0 & \frac{(1-i\nu)\nu}{2}
\end{pmatrix}, \\
& \mbox{Res} \Big( J_{R}^{(2)}(\xi;s),\xi = b_{2} \Big) =- \frac{f''(b_2)}{2f'(b_2)^3}\begin{pmatrix}
(1+i\nu)\nu & 0 \\
0 & (1-i\nu)\nu
\end{pmatrix},
\end{align*}
where $f'(b_2)$ and $f''(b_2)$ are given by \eqref{cb2 and cb2p2p}. Since
\begin{align*}
R_{-}^{(1)}(\xi;s) = R_{+}^{(1)}(\xi;s) - J_{R}^{(1)}(\xi;s),
\end{align*}
and since $R_{+}^{(1)}(\xi;s)$ has already been computed in \eqref{Rp1p explicit}, we obtain 
\begin{align*}
& R^{(1)}(b_{2};s) = - \frac{\overline{A^{(1)}(s)}}{b_{2}-b_{1}} - \mbox{Res} \Big( \frac{J_{R}^{(1)}(\xi;s)}{\xi-b_{2}},\xi=b_{2} \Big) \\
& = - \frac{\overline{A^{(1)}(s)}}{b_{2}-b_{1}} + \frac{1}{f'(b_2)} \begin{pmatrix}
0 & (-\frac{1}{6b_2}-2\beta_0(s))\beta_{12} E_0(b_{2};s)_{11}^{2} \\
(-\frac{1}{6b_2}+2\beta_0(s)) \beta_{21} E_0(b_{2};s)_{11}^{-2} & 0
\end{pmatrix},
\end{align*}
and the constant $\beta_0(s)$ is given by \eqref{eb2p0p at b2}. Summarizing, we have
\begin{align}\label{Rp2p at infty}
R^{(2)}(\zeta;s) = \frac{A^{(2)}(s)}{\zeta-b_{2}} + \frac{B^{(2)}(s)}{(\zeta-b_{2})^{2}} + \frac{-\overline{A^{(2)}(s)}}{\zeta-b_{1}} + \frac{\overline{B^{(2)}(s)}}{(\zeta-b_{1})^{2}}, \qquad \mbox{for }\zeta \in \mathbb{C}\setminus (\mathcal{D}_{b_{1}}\cup \mathcal{D}_{b_{2}}), 
\end{align}
where $A^{(2)}$ and $B^{(2)}$ are diagonal matrices given by
\begin{align*}
& B^{(2)}_{11}(s) = \frac{(1+i\nu)\nu}{2f'(b_{2})^{2}}, \qquad \qquad  B^{(2)}_{22}(s) = \frac{(1-i\nu)\nu}{2f'(b_2)^2}, \\
& A^{(2)}_{11}(s) = \frac{1}{f'(b_2)^{2}}\left( \left(-\frac{1}{6b_2}-2\beta_0(s)\right)\beta_{12}\beta_{21} +\frac{1}{6b_2}(1+i\nu)\nu - \frac{f'(b_2)}{\overline{f'(b_2)}} 
\frac{ \beta_{21} \overline{\beta_{12}}}{(b_{2}-b_{1})\mathfrak{e}^{2}}  \right), \\
& A^{(2)}_{22}(s) = \frac{1}{f'(b_2)^{2}}\left( \left(-\frac{1}{6b_2}+2\beta_0(s)\right)\beta_{12}\beta_{21} +\frac{1}{6b_2}(1-i\nu)\nu - \frac{f'(b_2)}{\overline{f'(b_2)}} \frac{ \beta_{12} \overline{\beta_{21}}\mathfrak{e}^{2}}{(b_{2}-b_{1})}  \right),
\end{align*}
where $\mathfrak{e} = \mathfrak{e}(s)$ depends on $s$, but satisfies $|\mathfrak{e}(s)| = 1$ for all values of $s$. Its precise expression is given by \eqref{def of efrak}. The formula \eqref{beta 12 beta21 relation} allows the simplification
\begin{align}\label{explicit expression for Ap2p11}
A^{(2)}_{11}(s) =  \frac{1}{f'(b_2)^2}\left( \left(-\frac{1}{6b_2}-2\beta_0(s)\right)\nu +\frac{1}{6b_2}(1+i\nu)\nu - \frac{f'(b_2)}{\overline{f'(b_2)}} \frac{ \beta_{21} \overline{\beta_{12}}}{(b_{2}-b_{1})\mathfrak{e}^{2}}  \right),
\end{align}
and similarly for $A_{22}^{(2)}$. We end this section with a lemma about the structure of the matrices $R^{(j)}$, $j \geq 1$. 
\begin{lemma}\label{lemma: diag and off-diag}
For any $j \geq 1$, the matrix $R^{(2j-1)}$ is off-diagonal and the matrix $R^{(2j)}$ is diagonal.
\end{lemma}
\begin{proof}
By \eqref{eq:integraleq}, the matrices $R^{(j)}$ can be computed recursively as follows:
\begin{align*}
R^{(j)}(\zeta;s) = \frac{1}{2\pi i} \int_{\partial \mathcal{D}_{b_{1}} \cup \partial \mathcal{D}_{b_{2}}} \frac{\sum_{\ell = 1}^{j}R_{-}^{(j-\ell)}(\xi;s)J_{R}^{(\ell)}(\xi;s)}{\xi-\zeta}d\xi, \qquad j \geq 1.
\end{align*}
The result follows by induction, provided that the matrices $J_{R}^{(2j)}$ are diagonal and $J_{R}^{(2j-1)}$ are off-diagonal.
To prove this claim, consider \eqref{matching param b2} and \eqref{expansion of Eb2 in s}. These imply that $J_R^{(j)}(\zeta;s)$ from \eqref{jumps for R first order on the boundary Db2} is composed of terms of the form
\[\frac{1}{f^{j-2k}}E_m\Phi_{{\rm PC}, j-2k}(E^{-1})_{k-m},\]
for $m=0,\ldots, k$ and $k=0, 1, \ldots, \lfloor\frac{j-1}{2}\rfloor$, and where $E_m$ and $(E^{-1})_{k-m}$ are diagonal matrices. All these terms are diagonal if $j$ is even and off-diagonal if $j$ is odd.
\end{proof}

\section{Proofs of Theorems \ref{thm:expmoments Wrights} and \ref{thm:expmoments Meijer}: part 1}\label{Section: Differential identity in s}
In this section, we use the analysis of Section \ref{Section: steepest descent} to prove part of Theorems \ref{thm:expmoments Wrights} and \ref{thm:expmoments Meijer} via the differential identity in $s$ 
\begin{align}\label{recall the diff id in s}
& \partial_{s} \log \det \left( 1 - (1-t)\mathbb{K}^{(j)}\Big|_{[0,s]} \right) = \frac{Y_{1,11}}{s},
\end{align}
which was derived in \eqref{diff identity with s}. As mentioned in the introduction, the advantage of this differential identity is that it leads to a significantly simpler analysis than the one carried out in Section \ref{Section: Differential identity in t} and that it allows to prove the optimal bound $\bigO(s^{-\rho})$ for the error terms of \eqref{asymp gap thm explicit  Wr} and \eqref{asymp gap thm explicit  Me}. The main disadvantage is that it does not allow for the evaluation of the constants $C$ of \eqref{asymp gap thm explicit  Wr} and \eqref{asymp gap thm explicit  Me}. These constants will be obtained in Section \ref{Section: Differential identity in t}.

\medskip

By \eqref{relation T1 to Y1}, we have
\begin{align*}
T_{1,11} = \frac{1}{i s^{\rho}}Y_{1,11}.
\end{align*}
On the other hand, for $\zeta$ outside the lenses and outside the disks, we know from \eqref{def of R} that
\begin{align*}
T(\zeta) = R(\zeta)P^{(\infty)}(\zeta),
\end{align*}
from which we deduce, by \eqref{asymp of T at inf}, \eqref{asymp of Pinf at inf}, \eqref{asymp of R at inf}, and \eqref{estimate for R} that
\begin{align*}
T_{1}=T_1(s) = P^{(\infty)}_{1} + R_{1}(s) = P^{(\infty)}_{1} + \sum_{j=1}^{2N+1} R_{1}^{(j)}(s) \, s^{-\frac{j\rho}{2}} + \bigO(s^{-(N+1)\rho}), \qquad \mbox{as } s \to + \infty,
\end{align*}
uniformly for $t$ in compact subsets of $\mathbb{R}$, where $N \in \mathbb{N}$ is arbitrary, and where the coefficients $R_{1}^{(j)}(s)$, $j \geq 1$, are defined via the expansion
\begin{align*}
R^{(j)}(\zeta) = \frac{R^{(j)}_{1}(s)}{\zeta} + \bigO(\zeta^{-2}), \qquad \mbox{as } \zeta \to \infty.
\end{align*}
We know from Lemma \ref{lemma: diag and off-diag} that $R_{1}^{(2j-1)}$ is off-diagonal for all $j \geq 1$. Thus, using \eqref{recall the diff id in s}, we find
\begin{align}\label{asymp diff identity in s}
\partial_{s} \log \det \left( 1 - (1-t)\mathbb{K}^{(j)}\Big|_{[0,s]} \right) & = is^{\rho-1}\Big( P^{(\infty)}_{1} + \sum_{j=1}^{2N+1} R_{1}^{(j)}(s)s^{-\frac{j\rho}{2}} + \bigO(s^{-(N+1)\rho}) \Big)_{11} \nonumber \\
& = is^{\rho-1}\Big( P^{(\infty)}_{1,11} + \sum_{j=1}^{N} R_{1,11}^{(2j)}(s)s^{-j \rho} + \bigO(s^{-(N+1)\rho}) \Big),
\end{align}
as $s \to + \infty$. After integrating \eqref{asymp diff identity in s}, we obtain
\begin{align}\label{lol6}
\log \det \left( 1 - (1-t)\mathbb{K}^{(j)}\Big|_{[0,s]} \right) = \frac{i}{\rho}P_{1,11}^{(\infty)} s^{\rho} +  \int_{M}^{s}\frac{iR_{1,11}^{(2)}(s)}{s} ds + \log C_{1} + \bigO(s^{-\rho}),
\end{align}
as $s \to + \infty$, where $C_{1}$ is an unknown constant of integration and $M$ is a sufficiently large constant (i.e. $M$ is independent of $s$). An explicit expression for $P_{1,11}^{(\infty)}$ has been computed in \eqref{Pinf1}. Then, the leading coefficient in \eqref{lol6} is given by
\begin{align}\label{leading term}
\frac{i}{\rho} P_{1,11}^{(\infty)} = -\frac{iD_{1}}{\rho} = -\frac{2 \nu \re b_{2}}{\rho}.
\end{align}
We now turn to the computation of the second term of \eqref{lol6}. Using \eqref{explicit expression for Ap2p11} and \eqref{Rp2p at infty}, we obtain
\begin{align*}
iR_{1,11}^{(2)}(s) = i(A_{11}^{(2)}(s) - \overline{A_{11}^{(2)}}(s)) = -2 \, \im A_{11}^{(2)}(s) = \frac{\nu^{2}}{c_{1}+c_{2}} + \frac{1}{|f'(b_2)|^{2} \re b_{2}} \im \big( \beta_{21} \overline{\beta_{12}}\mathfrak{e}(s)^{-2} \big).
\end{align*}
We recall that $\mathfrak{e}(s)$ is given by 
\begin{align*}
\mathfrak{e}(s) = \Big[ (b_{2}-b_{1})|f'(b_2)|s^{\frac{\rho}{2}} \Big]^{2i\nu } \exp \Big( i c_{5} \arg(ib_{2}) + i c_{6} \arg(-ib_{2}) \Big) e^{-i \widetilde{\ell} s^{\rho}}. 
\end{align*}
In particular, it satisfies $|\mathfrak{e}(s)| = 1$ and it oscillates rapidly as $s \to + \infty$, since $\widetilde{\ell} \neq 0$, see \eqref{def of ell tilde}. Therefore, we have
\begin{align*}
\int_{M}^{s} \mathfrak{e}^{-2}(s) s^{-1} ds & = \tilde{c} \int_{M}^{s} e^{2 i \widetilde{\ell}s^{\rho}-2 \nu i \log s^{\rho}} s^{-1} ds = \frac{\tilde{c}}{\rho}\int_{M^{\rho}}^{s^{\rho}} e^{2i \widetilde{\ell} u -2 \nu i \log u} u^{-1}du \\
& = \log C_{2} + C_{3}(s).
\end{align*}
where $\tilde{c}$ and $C_{2}$ are constants whose exact values are unimportant for us, and $C_{3}(s)$ is bounded by
\begin{align*}
|C_{3}(s)| = \bigO\Big(\frac{1}{\widetilde{\ell}s^{\rho}}\Big)\qquad \mbox{as $s\to +\infty$}.
\end{align*}
We conclude that the integral in \eqref{lol6} has the following asymptotics
\begin{align}\label{integral in diff in s asymp}
\int_{M}^{s}\frac{iR_{1,11}^{(2)}(s)}{s} ds = \frac{\nu^{2}}{c_{1}+c_{2}}\log s + \log(C_{2}) + \bigO(s^{-\rho}), \qquad \mbox{as } s \to + \infty.
\end{align}
Since $\rho = \frac{1}{c_{1}+c_{2}}$, by combining \eqref{lol6}, \eqref{leading term}, and \eqref{integral in diff in s asymp}, we get
\begin{align}\label{final asymptotics}
\log \det \left( 1 - (1-t)\mathbb{K}^{(j)}\Big|_{[0,s]} \right) = -\frac{2 \nu \re b_{2}}{\rho} s^{\rho} + \nu^{2}\log s^{\rho} + \log(C) + \bigO\big(s^{-\rho}\big),
\end{align}
as $s \to + \infty$, where $C = C_{1}C_{2}$. The values of $\rho^{(1)}$ and $\rho^{(2)}$ are given by \eqref{tau rho 1} and \eqref{tau rho 2}, respectively, and $\re b_{2}^{(j)}$, $j=1,2$, can be evaluated by using \eqref{def of b2} together with the coefficients $c_{1}$, $c_{2}$ and $c_{3}$ (given above \eqref{def of h}). Recalling also that 
\[\log \det \left( 1 - (1-t)\mathbb{K}^{(j)}\Big|_{[0,s]} \right)=\mathbb E\left[e^{-2\pi\nu N(s)}\right],\]
we have now completed the proofs of Theorems \ref{thm:expmoments Wrights} and \ref{thm:expmoments Meijer}, up to the determination of the constants $C=C^{(j)}$, $j=1,2$.

\section{Proofs of Theorems \ref{thm:expmoments Wrights} and \ref{thm:expmoments Meijer}: part 2}\label{Section: Differential identity in t}
In this section, we will compute $C$ via the differential identity in $t$
\begin{align}\label{recall the diff id in t}
& \partial_{t} \log \det \left( 1 - (1-t)\mathbb{K}^{(j)}\Big|_{[0,s]} \right) = \frac{-1}{2(1-t)} \int_{\gamma \cup \tilde{\gamma}} {\rm Tr} \Big[ Y^{-1}(z)Y'(z)(J(z)-I) \Big] \frac{dz}{2\pi i},
\end{align}
which was derived in Lemma \ref{lemma: differential identities}. 

\medskip

We divide the proofs is a series of lemmas. First, we use the analysis of Section \ref{Section: steepest descent} to expand the right-hand side of \eqref{recall the diff id in t} as $s \to + \infty$.
\begin{lemma}
As $s \to + \infty$, we have
\begin{align}
& \partial_{t} \log \det \left( 1 - (1-t)\mathbb{K}^{(j)}\Big|_{[0,s]} \right) = I_{1} + I_{2} + 2 \, \re I_{b_{2}} +\bigO(e^{-c s^\rho}), \label{diff identity t first lemma asymptotics}
\end{align}
where $c>0$ and 
\begin{align}
I_{1} & = - \frac{1}{t}\int_{\Sigma_{5}} \Big( \log \big( e^{-is^{\rho}h(\zeta)} \mathcal{G}(\zeta;s) \big)  \Big)' \frac{d\zeta}{2\pi i} = - \frac{2}{t}\re\bigg[ \int_{[0,b_{2}]} \Big( \hspace{-0.1cm}-\hspace{-0.05cm} i s^{\rho}h'(\zeta)+ \big(\log \mathcal{G}(\zeta;s)   \big)'\Big) \frac{d\zeta}{2\pi i}\bigg] , \label{def of I1} \\
I_{2} & = - \frac{1}{t}\int_{\Sigma_{5}} {\rm Tr}\big[ T^{-1}(\zeta)T'(\zeta)\sigma_{3} \big] \frac{d\zeta}{2\pi i} = - \frac{2}{t}\re \bigg[\int_{[0,b_{2}]} {\rm Tr}\big[ T^{-1}(\zeta)T'(\zeta)\sigma_{3} \big] \frac{d\zeta}{2\pi i}\bigg], \label{def of I2} \\
I_{b_{2}}  & = \frac{1}{2\sqrt{1-t}} \int_{\Sigma_{2} \cap \mathcal{D}_{b_{2}}} e^{-s^{\rho}(ih(\zeta)-\ell)} \widetilde{\mathcal{G}}(\zeta;s){\rm Tr}\big[T^{-1}(\zeta) T'(\zeta)\sigma_{+}\big]\frac{d\zeta}{2\pi i} \nonumber \\
& + \frac{2-t}{2t^{2}\sqrt{1-t}} \int_{\Sigma_{7} \cap \mathcal{D}_{b_{2}}} e^{-s^{\rho}(ih(\zeta)-\ell)} \widetilde{\mathcal{G}}(\zeta;s){\rm Tr}\big[T^{-1}(\zeta) T'(\zeta)\sigma_{+}\big]\frac{d\zeta}{2\pi i} \nonumber \\
& + \frac{-1}{2\sqrt{1-t}} \int_{ \Sigma_{4}\cap \mathcal{D}_{b_{2}}} e^{s^{\rho}(ih(\zeta)-\ell)} \widetilde{\mathcal{G}}(\zeta;s)^{-1}{\rm Tr}\big[T^{-1}(\zeta) T'(\zeta)\sigma_{-}\big]\frac{d\zeta}{2\pi i} \nonumber \\
& + \frac{-(2-t)}{2t^{2}\sqrt{1-t}} \int_{ \Sigma_{6}\cap \mathcal{D}_{b_{2}}} e^{s^{\rho}(ih(\zeta)-\ell)} \widetilde{\mathcal{G}}(\zeta;s)^{-1}{\rm Tr}\big[T^{-1}(\zeta) T'(\zeta)\sigma_{-}\big]\frac{d\zeta}{2\pi i}. \label{def of Ib2}
\end{align}
\end{lemma}
\begin{proof}
Using the change of variables $z = is^{\rho}\zeta + \tau$ in \eqref{recall the diff id in t}, we obtain
\begin{align}
& \partial_{t} \log \det \left( 1 - (1-t)\mathbb{K}^{(j)}\Big|_{[0,s]} \right) = \frac{-1}{2(1-t)} \int_{\gamma_{U}\cup\tilde{\gamma}_{U}} {\rm Tr}\big[U^{-1}(\zeta) U'(\zeta)(J_{U}(\zeta)-I)\big]\frac{d\zeta}{2\pi i} = I_{\gamma} + I_{\tilde{\gamma}}, \nonumber \\
& I_{\gamma} = \frac{1}{2\sqrt{1-t}} \int_{\gamma_{U}} e^{-is^{\rho}h(\zeta)} \mathcal{G}(\zeta;s){\rm Tr}\big[U^{-1}(\zeta) U'(\zeta)\sigma_{+}\big]\frac{d\zeta}{2\pi i}, \label{I gamma U} \\
& I_{\tilde{\gamma}} = \frac{-1}{2\sqrt{1-t}} \int_{\tilde{\gamma}_{U}} e^{is^{\rho}h(\zeta)} \mathcal{G}(\zeta;s)^{-1}{\rm Tr}\big[U^{-1}(\zeta) U'(\zeta)\sigma_{-}\big]\frac{d\zeta}{2\pi i}, \label{I gamma tilde U}
\end{align}
where $U$ is defined in \eqref{def of U}, $\gamma_{U}$ and $\tilde{\gamma}_{U}$ are defined in \eqref{def of gamma U and gamma tilde U}, 
\begin{align*}
\sigma_{+} = \begin{pmatrix}
0 & 1 \\
0 & 0
\end{pmatrix} \quad \mbox{ and } \quad \sigma_{-} = \begin{pmatrix}
0 & 0 \\
1 & 0
\end{pmatrix},
\end{align*}
and where we have used \eqref{JU in terms of G}. Note that we do not specify whether we take the $+$ or $-$ boundary values of $U$ in \eqref{I gamma U} and \eqref{I gamma tilde U}, which is without ambiguity, see Remark \ref{remark: no boundary values indicated}. Now, we deform the contours of integration by using the analytic continuation of $U$ (denoted $\widehat{T}$ and defined in Section \ref{subsection: U to T}). We obtain
\begin{align}
& I_{\gamma} = \frac{1}{2\sqrt{1-t}} \int_{\Sigma_{1} \cup \Sigma_{2}} e^{-is^{\rho}h(\zeta)} \mathcal{G}(\zeta;s){\rm Tr}\big[\widehat{T}^{-1}(\zeta) \widehat{T}'(\zeta)\sigma_{+}\big]\frac{d\zeta}{2\pi i} \nonumber \\
& \hspace{0.4cm} + \frac{1}{2\sqrt{1-t}} \int_{\Sigma_{5}} e^{-is^{\rho}h(\zeta)} \mathcal{G}(\zeta;s){\rm Tr}\big[\widehat{T}_{+}^{-1}(\zeta) \widehat{T}_{+}'(\zeta)\sigma_{+}\big]\frac{d\zeta}{2\pi i}, \label{Igamma in terms of T hat} \\
& I_{\tilde{\gamma}} = \frac{-1}{2\sqrt{1-t}} \int_{\Sigma_{3}\cup \Sigma_{4}} e^{is^{\rho}h(\zeta)} \mathcal{G}(\zeta;s)^{-1}{\rm Tr}\big[\widehat{T}^{-1}(\zeta) \widehat{T}'(\zeta)\sigma_{-}\big]\frac{d\zeta}{2\pi i} \nonumber \\
& \hspace{0.4cm} + \frac{-1}{2\sqrt{1-t}} \int_{\Sigma_{5}} e^{is^{\rho}h(\zeta)} \mathcal{G}(\zeta;s)^{-1}{\rm Tr}\big[\widehat{T}_{-}^{-1}(\zeta) \widehat{T}_{-}'(\zeta)\sigma_{-}\big]\frac{d\zeta}{2\pi i}, \label{Igamma tilde in terms of T hat}
\end{align}
where the contours $\Sigma_{1},\ldots,\Sigma_{5}$ are shown in Figure \ref{fig: Sigma 1,2,3,4}. Once more, we have not specified the boundary values of $\widehat{T}$ for the integrals over $\Sigma_{j}$, $j=1,...,4$ of \eqref{Igamma in terms of T hat} and \eqref{Igamma tilde in terms of T hat}; again, this is without ambiguity. Note however that this is not the case for the integrals over $\Sigma_{5}$. For $\zeta \in \Sigma_{5}$, by \eqref{def of T} we have
\begin{align*}
& {\rm Tr}\big[\widehat{T}_{\pm}^{-1} \widehat{T}_{\pm}'\sigma_{\pm}\big] = {\rm Tr}\big[ H_{\pm}(H_{\pm}^{-1})' \sigma_{\pm} \big] + {\rm Tr}\big[ T_{\pm}^{-1}T_{\pm}'s^{-\frac{c_{4}}{2}\sigma_{3}}e^{\frac{s^{\rho}\ell}{2}\sigma_{3}}H_{\pm}^{-1} \sigma_{\pm}H_{\pm}e^{-\frac{s^{\rho}\ell}{2}\sigma_{3}}s^{\frac{c_{4}}{2}\sigma_{3}} \big], \\
& e^{-is^{\rho}h(\zeta)} \mathcal{G}(\zeta;s) {\rm Tr}\big[ H_{+}(H_{+}^{-1})' \sigma_{+} \big] = - \frac{\sqrt{1-t}}{t} \Big( \log \big( e^{-is^{\rho}h(\zeta)} \mathcal{G}(\zeta;s) \big)  \Big)', \\
& e^{is^{\rho}h(\zeta)} \mathcal{G}(\zeta;s)^{-1} {\rm Tr}\big[ H_{-}(H_{-}^{-1})' \sigma_{-} \big] = \frac{\sqrt{1-t}}{t} \Big( \log \big( e^{-is^{\rho}h(\zeta)} \mathcal{G}(\zeta;s) \big)  \Big)' , \\
& e^{-is^{\rho}h(\zeta)} \mathcal{G}(\zeta;s) s^{-\frac{c_{4}}{2}\sigma_{3}}e^{\frac{s^{\rho}\ell}{2}\sigma_{3}}H_{+}^{-1}\sigma_{+}H_{+} e^{-\frac{s^{\rho}\ell}{2}\sigma_{3}}s^{\frac{c_{4}}{2}\sigma_{3}}  = \begin{pmatrix}
- \frac{\sqrt{1-t}}{t} & e^{-s^{\rho}(ih(\zeta)-\ell)} \widetilde{\mathcal{G}}(\zeta;s) \\
- \frac{1-t}{t^{2}}e^{s^{\rho}(ih(\zeta)-\ell)} \widetilde{\mathcal{G}}(\zeta;s)^{-1} & \frac{\sqrt{1-t}}{t}
\end{pmatrix}, \\
& e^{is^{\rho}h(\zeta)} \mathcal{G}(\zeta;s)^{-1} s^{-\frac{c_{4}}{2}\sigma_{3}}e^{\frac{s^{\rho}\ell}{2}\sigma_{3}}H_{-}^{-1}\sigma_{-}H_{-} e^{-\frac{s^{\rho}\ell}{2}\sigma_{3}}s^{\frac{c_{4}}{2}\sigma_{3}} = \begin{pmatrix}
\frac{\sqrt{1-t}}{t} &  - \frac{1-t}{t^{2}} e^{-s^{\rho}(ih(\zeta)-\ell)} \widetilde{\mathcal{G}}(\zeta;s) \\
e^{s^{\rho}(ih(\zeta)-\ell)} \widetilde{\mathcal{G}}(\zeta;s)^{-1} & -\frac{\sqrt{1-t}}{t}
\end{pmatrix}.
\end{align*}
Therefore, using also the jumps for $T$ given by \eqref{jumps J_T}, we obtain
\begin{align*}
& \frac{1}{2\sqrt{1-t}} \int_{\Sigma_{5}} e^{-is^{\rho}h(\zeta)} \mathcal{G}(\zeta;s){\rm Tr}\big[\widehat{T}_{+}^{-1}(\zeta) \widehat{T}_{+}'(\zeta)\sigma_{+}\big]\frac{d\zeta}{2\pi i} = - \frac{1}{2t}\int_{\Sigma_{5}} \Big( \log \big( e^{-is^{\rho}h(\zeta)} \mathcal{G}(\zeta;s) \big)  \Big)' \frac{d\zeta}{2\pi i} \\
& - \frac{1}{2t}\int_{\Sigma_{5}} {\rm Tr}\big[ T^{-1}(\zeta)T'(\zeta)\sigma_{3} \big] \frac{d\zeta}{2\pi i}- \frac{\sqrt{1-t}}{2t^{2}}\int_{\Sigma_{6}}e^{s^{\rho}(ih(\zeta)-\ell)}\widetilde{\mathcal{G}}(\zeta;s)^{-1}{\rm Tr}\big[ T^{-1}(\zeta)T'(\zeta)\sigma_{-} \big] \frac{d\zeta}{2\pi i} \\
& + \frac{1}{2t^{2}\sqrt{1-t}}\int_{\Sigma_{7}}e^{-s^{\rho}(ih(\zeta)-\ell)}\widetilde{\mathcal{G}}(\zeta;s) {\rm Tr}\big[ T^{-1}(\zeta)T'(\zeta)\sigma_{+} \big] \frac{d\zeta}{2\pi i},
\end{align*}
and
\begin{align*}
& \frac{-1}{2\sqrt{1-t}} \int_{\Sigma_{5}} e^{is^{\rho}h(\zeta)} \mathcal{G}(\zeta;s)^{-1}{\rm Tr}\big[\widehat{T}_{-}^{-1}(\zeta) \widehat{T}_{-}'(\zeta)\sigma_{-}\big]\frac{d\zeta}{2\pi i} = - \frac{1}{2t}\int_{\Sigma_{5}} \Big( \log \big( e^{-is^{\rho}h(\zeta)} \mathcal{G}(\zeta;s) \big)  \Big)' \frac{d\zeta}{2\pi i} \\
& - \frac{1}{2t}\int_{\Sigma_{5}} {\rm Tr}\big[ T^{-1}(\zeta)T'(\zeta)\sigma_{3} \big] \frac{d\zeta}{2\pi i}- \frac{1}{2t^{2}\sqrt{1-t}}\int_{\Sigma_{6}}e^{s^{\rho}(ih(\zeta)-\ell)}\widetilde{\mathcal{G}}(\zeta;s)^{-1}{\rm Tr}\big[ T^{-1}(\zeta)T'(\zeta)\sigma_{-} \big] \frac{d\zeta}{2\pi i} \\
& + \frac{\sqrt{1-t}}{2t^{2}}\int_{\Sigma_{7}}e^{-s^{\rho}(ih(\zeta)-\ell)}\widetilde{\mathcal{G}}(\zeta;s) {\rm Tr}\big[ T^{-1}(\zeta)T'(\zeta)\sigma_{+} \big] \frac{d\zeta}{2\pi i}.
\end{align*}
Thus, using again \eqref{def of T} to rewrite the integrals over $\Sigma_{j}$, $j=1,2,3,4$, in terms of $T$, and collecting the above computations, we rewrite \eqref{Igamma in terms of T hat}--\eqref{Igamma tilde in terms of T hat} as follows,
\begin{align}
& I_{\gamma} = \frac{1}{2\sqrt{1-t}} \int_{\Sigma_{1} \cup \Sigma_{2}} e^{-s^{\rho}(ih(\zeta)-\ell)} \widetilde{\mathcal{G}}(\zeta;s){\rm Tr}\big[T^{-1}(\zeta) T'(\zeta)\sigma_{+}\big]\frac{d\zeta}{2\pi i}, \nonumber \\
& - \frac{1}{2t}\int_{\Sigma_{5}} \Big( \log \big( e^{-is^{\rho}h(\zeta)} \mathcal{G}(\zeta;s) \big)  \Big)' \frac{d\zeta}{2\pi i} \nonumber \\
& - \frac{1}{2t}\int_{\Sigma_{5}} {\rm Tr}\big[ T^{-1}(\zeta)T'(\zeta)\sigma_{3} \big] \frac{d\zeta}{2\pi i}- \frac{\sqrt{1-t}}{2t^{2}}\int_{\Sigma_{6}}e^{s^{\rho}(ih(\zeta)-\ell)}\widetilde{\mathcal{G}}(\zeta;s)^{-1}{\rm Tr}\big[ T^{-1}(\zeta)T'(\zeta)\sigma_{-} \big] \frac{d\zeta}{2\pi i} \nonumber \\
& + \frac{1}{2t^{2}\sqrt{1-t}}\int_{\Sigma_{7}}e^{-s^{\rho}(ih(\zeta)-\ell)}\widetilde{\mathcal{G}}(\zeta;s) {\rm Tr}\big[ T^{-1}(\zeta)T'(\zeta)\sigma_{+} \big] \frac{d\zeta}{2\pi i}, \label{Igamma in terms of T} 
\end{align}
\begin{align}
& I_{\tilde{\gamma}} = \frac{-1}{2\sqrt{1-t}} \int_{\Sigma_{3} \cup \Sigma_{4}} e^{s^{\rho}(ih(\zeta)-\ell)} \widetilde{\mathcal{G}}(\zeta;s)^{-1}{\rm Tr}\big[T^{-1}(\zeta) T'(\zeta)\sigma_{-}\big]\frac{d\zeta}{2\pi i} \nonumber \\
& - \frac{1}{2t}\int_{\Sigma_{5}} \Big( \log \big( e^{-is^{\rho}h(\zeta)} \mathcal{G}(\zeta;s) \big)  \Big)' \frac{d\zeta}{2\pi i} \nonumber \\
& - \frac{1}{2t}\int_{\Sigma_{5}} {\rm Tr}\big[ T^{-1}(\zeta)T'(\zeta)\sigma_{3} \big] \frac{d\zeta}{2\pi i}- \frac{1}{2t^{2}\sqrt{1-t}}\int_{\Sigma_{6}}e^{s^{\rho}(ih(\zeta)-\ell)}\widetilde{\mathcal{G}}(\zeta;s)^{-1}{\rm Tr}\big[ T^{-1}(\zeta)T'(\zeta)\sigma_{-} \big] \frac{d\zeta}{2\pi i} \nonumber \\
& + \frac{\sqrt{1-t}}{2t^{2}}\int_{\Sigma_{7}}e^{-s^{\rho}(ih(\zeta)-\ell)}\widetilde{\mathcal{G}}(\zeta;s) {\rm Tr}\big[ T^{-1}(\zeta)T'(\zeta)\sigma_{+} \big] \frac{d\zeta}{2\pi i}. \label{Igamma tilde in terms of T}
\end{align}
We note from \eqref{def of Pinf} that $P^{(\infty)}$ is independent of $s$. From \eqref{def of Pb2} and \eqref{symmetry local param}, we also note that $P^{(b_{j})}(\zeta)$, $j\in \{1,2\}$, depends on $s$ but is bounded as $s \to +\infty$ uniformly for $\zeta \in \mathcal{D}_{b_{j}}$, and that $P^{(b_{j})\prime}(\zeta) = \bigO(s^{\rho})$ as $s \to + \infty$ uniformly for $\zeta \in \mathcal{D}_{b_{j}}$. Using \eqref{def of R} and \eqref{estimate for R}, we infer that
\begin{align*}
T(\zeta) = \bigO(1), \qquad T'(\zeta) = \bigO(s^{\rho}), \qquad \mbox{as } s \to + \infty,
\end{align*}
uniformly for $\zeta \in \mathbb{C}\setminus \bigcup_{j=1}^{7} \Sigma_{j}$. Since $\re (ih(\zeta)-\ell) > 0$ for $\zeta \in \Sigma_{1} \cup \Sigma_{2}$ and $\re (ih(\zeta)-\ell) < 0$ for $\zeta \in \Sigma_{3} \cup \Sigma_{4}$ (see Figures \ref{fig: real part of ih-ell} and \ref{fig: Sigma 1,2,3,4}), we have
\begin{align*}
I_{\gamma} + I_{\tilde{\gamma}} = I_{1} + I_{2} + I_{b_{2}} + I_{b_{1}} + \bigO(e^{-cs^{\rho}}), \qquad \mbox{as } s \to + \infty,
\end{align*}
where $I_{1}$, $I_{2}$ and $I_{b_{2}}$ are defined in \eqref{def of I1}, \eqref{def of I2} and \eqref{def of Ib2}, respectively, and $I_{b_{1}}$ is defined similarly as $I_{b_{2}}$. Using the symmetry $\zeta \mapsto -\overline{\zeta}$ (see in particular \eqref{symmetry for T}), we obtain 
\begin{align*}
I_{b_{1}} = \overline{I_{b_{2}}},
\end{align*}
which finishes the proof.
\end{proof}

\begin{lemma}\label{lemma:I1}
\begin{align*}
I_{1} = \frac{\re b_{2}}{\pi \, \rho \, t} s^{\rho} + \frac{c_{1}c_{6}-c_{2}c_{5}}{t(c_{1}+c_{2})} + \bigO(s^{-\rho}), \qquad \mbox{as } s \to + \infty.
\end{align*}
\end{lemma}
\begin{proof}
By the definition \eqref{def of I1} of $I_1$, we have
\begin{align*}
I_{1} = \frac{s^{\rho}}{\pi t} \im \int_{[0,b_{2}]} ih'(\zeta)d\zeta - \frac{1}{\pi t}\im \int_{[0,b_{2}]}\big( \log \mathcal{G}(\zeta;s) \big)' d\zeta.
\end{align*}
For the first integral, we use \eqref{def of h}, \eqref{def of ell tilde}, and \eqref{def of b2}, to obtain
\begin{align*}
\im \int_{[0,b_{2}]} ih'(\zeta)d\zeta = \im (ih(b_{2})) = \widetilde{\ell}= (c_{1}+c_{2})\exp \left( - \frac{c_{1}+c_{2}+c_{3}}{c_{1}+c_{2}} \right) \cos \left( \frac{\pi}{2} \frac{c_{2}-c_{1}}{c_{1}+c_{2}} \right) = \frac{\re b_{2}}{\rho}.
\end{align*}
For the second integral, we find
\begin{align*}
\int_{[0,b_{2}]}\big( \log \mathcal{G}(\zeta;s) \big)' d\zeta = \log \mathcal{G}(b_{2};s)-\log \mathcal{G}(0;s).
\end{align*}
Using \eqref{asymp for G} and \eqref{def of G}, as $s \to + \infty$ we have
\begin{align*}
\int_{[0,b_{2}]}\big( \log \mathcal{G}(\zeta;s) \big)' d\zeta = c_{4} \log s + c_{5} \log(ib_{2}) + c_{6} \log(-ib_{2}) + c_{7} -\log F(\tau) + \bigO \left( s^{-\rho} \right),
\end{align*}
and thus, by \eqref{def of b2}, we get
\begin{align*}
\im \int_{[0,b_{2}]}\big( \log \mathcal{G}(\zeta;s) \big)' d\zeta & =  c_{5} \arg(ib_{2}) + c_{6} \arg(-ib_{2}) + \bigO \left( s^{-\rho} \right), \\
& = c_{5} \frac{\pi}{2} \left( \frac{c_{2}-c_{1}}{c_{2}+c_{1}} + 1 \right) + c_{6} \frac{\pi}{2} \left( \frac{c_{2}-c_{1}}{c_{2}+c_{1}} - 1 \right) + \bigO \left( s^{-\rho} \right).
\end{align*}
\end{proof}
We split $I_{b_{2}}$ into four parts
\begin{align}\label{splitting of Ib2 first time}
I_{b_{2}} = I_{b_{2},1} + I_{b_{2},2} + I_{b_{2},3} + I_{b_{2},4},
\end{align}
where $I_{b_{2},j}$, $j=1,2,3,4$, are given by
\begin{align}
I_{b_{2},1}  & = \frac{1}{2\sqrt{1-t}} \int_{\Sigma_{2} \cap \mathcal{D}_{b_{2}}} e^{-s^{\rho}(ih(\zeta)-\ell)} \widetilde{\mathcal{G}}(\zeta;s){\rm Tr}\big[T^{-1}(\zeta) T'(\zeta)\sigma_{+}\big]\frac{d\zeta}{2\pi i}, \label{def of Ib2 1} \\
I_{b_{2},2} & = \frac{2-t}{2t^{2}\sqrt{1-t}} \int_{\Sigma_{7} \cap \mathcal{D}_{b_{2}}} e^{-s^{\rho}(ih(\zeta)-\ell)} \widetilde{\mathcal{G}}(\zeta;s){\rm Tr}\big[T^{-1}(\zeta) T'(\zeta)\sigma_{+}\big]\frac{d\zeta}{2\pi i}, \nonumber \\
I_{b_{2},3}  & = \frac{-1}{2\sqrt{1-t}} \int_{ \Sigma_{4}\cap \mathcal{D}_{b_{2}}} e^{s^{\rho}(ih(\zeta)-\ell)} \widetilde{\mathcal{G}}(\zeta;s)^{-1}{\rm Tr}\big[T^{-1}(\zeta) T'(\zeta)\sigma_{-}\big]\frac{d\zeta}{2\pi i}, \nonumber \\
I_{b_{2},4} & = \frac{-(2-t)}{2t^{2}\sqrt{1-t}} \int_{ \Sigma_{6}\cap \mathcal{D}_{b_{2}}} e^{s^{\rho}(ih(\zeta)-\ell)} \widetilde{\mathcal{G}}(\zeta;s)^{-1}{\rm Tr}\big[T^{-1}(\zeta) T'(\zeta)\sigma_{-}\big]\frac{d\zeta}{2\pi i}. \nonumber
\end{align}
\begin{lemma}\label{lemma:Ib2}
As $s \to + \infty$, we have
\begin{align*}
& I_{b_{2},1} = \frac{1}{2\sqrt{1-t}} \int_{e^{\frac{\pi i}{4}}[0,+\infty)} {\rm Tr} \big[ \Phi_{\mathrm{PC}}^{-1}(z)\Phi_{\mathrm{PC}}'(z)\sigma_{+} \big]  \frac{dz}{2\pi i} +\bigO ( s^{-\frac{\rho}{2}} ),\\
& I_{b_{2},2} = \frac{2-t}{2t^{2}\sqrt{1-t}} \int_{e^{-\frac{3\pi i}{4}}(+\infty,0]} {\rm Tr} \big[ \Phi_{\mathrm{PC}}^{-1}(z)\Phi_{\mathrm{PC}}'(z)\sigma_{+} \big] \frac{dz}{2\pi i} +\bigO ( s^{-\frac{\rho}{2}} ), \\
& I_{b_{2},3} = \frac{-1}{2\sqrt{1-t}} \int_{e^{-\frac{\pi i}{4}}[0,+\infty)} {\rm Tr} \big[ \Phi_{\mathrm{PC}}^{-1}(z)\Phi_{\mathrm{PC}}'(z)\sigma_{-} \big] \frac{dz}{2\pi i} +\bigO ( s^{-\frac{\rho}{2}} ), \\
& I_{b_{2},4} = \frac{-(2-t)}{2t^{2}\sqrt{1-t}} \int_{e^{\frac{3\pi i}{4}}(+\infty,0]} {\rm Tr} \big[ \Phi_{\mathrm{PC}}^{-1}(z)\Phi_{\mathrm{PC}}'(z)\sigma_{-} \big] \frac{dz}{2\pi i} +\bigO ( s^{-\frac{\rho}{2}} ). \\
\end{align*}
\end{lemma}
\begin{proof}
From \eqref{def of fb2} and \eqref{cb2 and cb2p2p}, we have
\begin{align*}
s^{\rho}(ih(\zeta)-\ell) = s^{\rho} \Big( i \widetilde{\ell} - i \frac{f(\zeta)^{2}}{2} \Big) = i \widetilde{\ell} s^{\rho} - \frac{is^{\rho}}{2}f'(b_2)^{2}(\zeta-b_{2})^{2}(1+\bigO(\zeta - b_{2})), \qquad \mbox{as } \zeta \to b_{2},
\end{align*}
so the main contribution as $s \to + \infty$ in the integrals for $I_{b_{2},j}$, $j=1,...,4$ comes from the integrand as $s^{\frac{\rho}{2}}(\zeta-b_{2})  = \bigO(1)$.
We first obtain an expansion for ${\rm Tr}\big[ T^{-1}T'\sigma_{\pm} \big]$ as $s^{\frac{\rho}{2}}(\zeta-b_{2}) \to 0$ and simultaneously $s \to + \infty$. For $\zeta$ inside the disk $\mathcal{D}_{b_{2}}$, by \eqref{def of R} we have
\begin{align}\label{lol1}
T(\zeta) = R(\zeta)P^{(b_{2})}(\zeta).
\end{align}
Thus for $\zeta \in \mathcal{D}_{b_{2}}$, we have
\begin{align}\label{split the trace for Ib2}
{\rm Tr}\big[ T^{-1}T'\sigma_{\pm} \big] & = {\rm Tr}\big[ (P^{(b_{2})})^{-1}(P^{(b_{2})})'\sigma_{\pm} \big] +  {\rm Tr}\big[ (P^{(b_{2})})^{-1}R^{-1}R'P^{(b_{2})}\sigma_{\pm} \big].
\end{align}
We recall from \eqref{def of Pb2} that $P^{(b_{2})}$ is given by
\begin{align*}
P^{(b_{2})}(\zeta) = E(\zeta;s) \Phi_{\mathrm{PC}}(s^{\frac{\rho}{2}}f(\zeta);\sqrt{1-t})e^{ \frac{s^{\rho} }{2}(ih(\zeta)-\ell) \sigma_{3}} \widetilde{\mathcal{G}}(\zeta;s)^{-\frac{\sigma_{3}}{2}}, \qquad \zeta \in \mathcal{D}_{b_{2}} \setminus \bigcup_{j=1}^{7} \Sigma_{j},
\end{align*}
and thus
\begin{align}
& {\rm Tr}\big[ (P^{(b_{2})})^{-1}(P^{(b_{2})})'\sigma_{\pm} \big] = e^{\pm s^{\rho}(ih(\zeta)-\ell)}\widetilde{\mathcal{G}}(\zeta;s)^{\mp 1} \Big(s^{\frac{\rho}{2}}f' {\rm Tr} \big[ \Phi_{\mathrm{PC}}^{-1}\Phi_{\mathrm{PC}}'\sigma_{\pm} \big] + {\rm Tr}\big[ \Phi_{\mathrm{PC}}^{-1}E^{-1}E' \Phi_{\mathrm{PC}}\sigma_{\pm} \big]\Big),  \nonumber \\
& {\rm Tr}\big[ (P^{(b_{2})})^{-1}R^{-1}R'P^{(b_{2})}\sigma_{\pm} \big] = e^{\pm s^{\rho}(ih(\zeta)-\ell)}\widetilde{\mathcal{G}}(\zeta;s)^{\mp 1} {\rm Tr}\big[ \Phi_{\mathrm{PC}}^{-1}E^{-1}R^{-1}R'E\Phi_{\mathrm{PC}}\sigma_{\pm} \big], \label{explicit trace for Ib2 j}
\end{align}
where $\Phi_{\mathrm{PC}}$ and $\Phi_{\mathrm{PC}}'$ are evaluated at $s^{\frac{\rho}{2}}f(\zeta)$ and the other functions are evaluated at $\zeta$. We also recall from \eqref{def of Eb2} that 
\begin{align*}
E(\zeta;s) = P^{(\infty)}(\zeta)\widetilde{\mathcal{G}}(\zeta;s)^{\frac{\sigma_{3}}{2}}e^{-\frac{s^{\rho}}{2}i \widetilde{\ell}\sigma_{3}}\big( s^{\frac{\rho}{2}}f(\zeta) \big)^{i\nu \sigma_{3}}
\end{align*}
is analytic for $\zeta \in \mathcal{D}_{b_{2}}$, and thus
\begin{align}\label{estimate Eb2}
& E^{\pm 1}(\zeta;s) = \bigO(1), \qquad  E'(\zeta;s) = \bigO(1),
\end{align}
as $s \to + \infty$ uniformly for $\zeta  \in \mathcal{D}_{b_{2}}$. By \eqref{def of Ib2 1}, \eqref{split the trace for Ib2}, and \eqref{explicit trace for Ib2 j}, we have
\begin{align*}
I_{b_{2},1}  &  = \frac{1}{2\sqrt{1-t}} \int_{\Sigma_{2} \cap \mathcal{D}_{b_{2}}} \Big(s^{\frac{\rho}{2}}f' {\rm Tr} \big[ \Phi_{\mathrm{PC}}^{-1}\Phi_{\mathrm{PC}}'\sigma_{+} \big] + {\rm Tr}\big[ \Phi_{\mathrm{PC}}^{-1}E^{-1}E' \Phi_{\mathrm{PC}}\sigma_{+} \big]\Big)\frac{d\zeta}{2\pi i} \\
& + \frac{1}{2\sqrt{1-t}} \int_{\Sigma_{2} \cap \mathcal{D}_{b_{2}}} {\rm Tr}\big[ \Phi_{\mathrm{PC}}^{-1}E^{-1}R^{-1}R'E\Phi_{\mathrm{PC}}\sigma_{+} \big] \frac{d\zeta}{2\pi i}.
\end{align*}
Let us now perform the change of variables
\begin{align}\label{change of variables Ib2 j}
z = s^{\frac{\rho}{2}} f(\zeta),
\end{align}
where we recall that $f$ is injective on $\mathcal{D}_{b_{2}}$. Then we have
\begin{align*}
e^{-s^{\rho}(ih(\zeta)-\ell)} = e^{-is^{\rho}\widetilde{\ell}}e^{\frac{iz^{2}}{2}}, \qquad \zeta = f^{-1}(s^{-\frac{\rho}{2}}z), \qquad dz = s^{\frac{\rho}{2}}f'(\zeta) d\zeta.
\end{align*}
Since $\Sigma_{2} \cap \mathcal{D}_{b_{2}}$ is mapped by $f$ to a subset of $e^{\frac{\pi i}{4}}[0,+\infty)$, see \eqref{Sigma j is well-mapped by fb2}, this change of variables allows us to rewrite $I_{b_{2},1}$ as
\begin{align*}
I_{b_{2},1}& = \frac{1}{2\sqrt{1-t}} \int_{e^{\frac{\pi i}{4}}[0,s^{\frac{\rho}{2}}r]} {\rm Tr} \big[ \Phi_{\mathrm{PC}}^{-1}(z)\Phi_{\mathrm{PC}}'(z)\sigma_{+} \big]  \frac{dz}{2\pi i} \\
& + \frac{1}{2\sqrt{1-t}} \int_{e^{\frac{\pi i}{4}}[0,s^{\frac{\rho}{2}}r]} \frac{{\rm Tr}\big[ \Phi_{\mathrm{PC}}^{-1}(z)E^{-1}(\zeta;s)E'(\zeta;s) \Phi_{\mathrm{PC}}(z)\sigma_{+} \big]}{s^{\frac{\rho}{2}}f'(\zeta)} \frac{dz}{2\pi i}\\
& + \frac{1}{2\sqrt{1-t}} \int_{e^{\frac{\pi i}{4}}[0,s^{\frac{\rho}{2}}r]} \frac{{\rm Tr}\big[ \Phi_{\mathrm{PC}}^{-1}(z)E^{-1}(\zeta;s)R^{-1}(\zeta)R'(\zeta)E(\zeta;s)\Phi_{\mathrm{PC}}(z)\sigma_{+} \big]}{s^{\frac{\rho}{2}}f'(\zeta)}  \frac{dz}{2\pi i},
\end{align*}
where $r := |f(r_{\star})|$ with $r_{\star}$ defined by $\mathcal{D}_{b_{2}} \cap \Sigma_{2} = \{r_{\star}\}$. We note from \eqref{Phi PC at inf} that
\begin{align*}
\Phi_{\mathrm{PC}}(z) \sigma_{+} \Phi_{\mathrm{PC}}^{-1}(z) = \bigO(e^{\frac{iz^{2}}{2}}) \qquad \mbox{as } z \to \infty,
\end{align*}
and we conclude from \eqref{estimate for R} and \eqref{estimate Eb2} that
\begin{align*}
& \frac{1}{2\sqrt{1-t}} \int_{e^{\frac{\pi i}{4}}[0,s^{\frac{\rho}{2}}r]} \frac{{\rm Tr}\big[ \Phi_{\mathrm{PC}}^{-1}(z)E^{-1}(\zeta;s)E'(\zeta;s) \Phi_{\mathrm{PC}}(z)\sigma_{+} \big]}{s^{\frac{\rho}{2}}f'(\zeta)} \frac{dz}{2\pi i} = \bigO(s^{-\frac{\rho}{2}}), \\
& \frac{1}{2\sqrt{1-t}} \int_{e^{\frac{\pi i}{4}}[0,s^{\frac{\rho}{2}}r]} \frac{{\rm Tr}\big[ \Phi_{\mathrm{PC}}^{-1}(z)E^{-1}(\zeta;s)R^{-1}(\zeta)R'(\zeta)E(\zeta;s)\Phi_{\mathrm{PC}}(z)\sigma_{+} \big]}{s^{\frac{\rho}{2}}f'(\zeta)}  \frac{dz}{2\pi i} = \bigO(s^{-\rho}), \\
& \frac{1}{2\sqrt{1-t}} \int_{e^{\frac{\pi i}{4}}[0,s^{\frac{\rho}{2}}r]} \hspace{-0.2cm} {\rm Tr} \big[ \Phi_{\mathrm{PC}}^{-1}(z)\Phi_{\mathrm{PC}}'(z)\sigma_{+} \big]  \frac{dz}{2\pi i} = \frac{1}{2\sqrt{1-t}} \int_{e^{\frac{\pi i}{4}}[0,+\infty)} \hspace{-0.3cm}{\rm Tr} \big[ \Phi_{\mathrm{PC}}^{-1}(z)\Phi_{\mathrm{PC}}'(z)\sigma_{+} \big]  \frac{dz}{2\pi i} + \bigO(e^{-c s^{\rho}})
\end{align*}
as $s \to + \infty$, for a certain $c>0$. 
This finishes the proof for $I_{b_{2},1}$. The proofs of the expressions for the other integrals are similar.
\end{proof}
\begin{lemma}\label{lemma: asymp for Ib2 final}
We have
\begin{align*}
I_{b_{2}} = \mathcal{I}_{b_{2}}+ \bigO(s^{-\frac{\rho}{2}}), \qquad \mbox{as } s \to + \infty,
\end{align*}
where $\mathcal{I}_{b_{2}}$ depends on $t$ but is independent of the other parameters. More precisely, for $j=1$ (the Meijer-$G$ process), $\mathcal{I}_{b_{2}}$ is independent of $r$, $q$, $\nu_{1},...,\nu_{r}$, $\mu_{1},...,\mu_{q}$, and for $j=2$   (the Wright's generalized Bessel process), $\mathcal{I}_{b_{2}}$ is independent of $\alpha$ and $\theta$.
\end{lemma}
\begin{proof}
This follows from \eqref{splitting of Ib2 first time}, Lemma \ref{lemma:Ib2}, and the fact that $\Phi_{\mathrm{PC}}$ only depends on $q = \sqrt{1-t}$.
\end{proof}
Let $b_{\star} := \Sigma_{5}\cap \partial \mathcal{D}_{b_{2}}$. We split $I_{2}$ into two parts:
\begin{align*}
& I_{2} = - \frac{2}{t}\re \bigg[\int_{[0,b_{2}]} {\rm Tr}\big[ T^{-1}(\zeta)T'(\zeta)\sigma_{3} \big] \frac{d\zeta}{2\pi i}\bigg] = I_{2,1} + I_{2,2}, \\
& I_{2,1} = - \frac{2}{t}\re \bigg[\int_{[0,b_{\star}]} {\rm Tr}\big[ T^{-1}(\zeta)T'(\zeta)\sigma_{3} \big] \frac{d\zeta}{2\pi i}\bigg], \\
& I_{2,2} = - \frac{2}{t}\re \bigg[\int_{\Sigma_{5}\cap\mathcal{D}_{b_{2}}} {\rm Tr}\big[ T^{-1}(\zeta)T'(\zeta)\sigma_{3} \big] \frac{d\zeta}{2\pi i}\bigg].
\end{align*}
\begin{lemma}\label{lemma: asymp for I21}
\begin{align*}
I_{2,1} = \frac{2\nu}{\pi t} \log \left| \frac{b_{\star}- b_{2}}{b_{\star}-b_{1}} \right|+ \bigO(s^{-\frac{\rho}{2}}), \qquad \mbox{as } s \to + \infty.
\end{align*}
\end{lemma}
\begin{proof}
For $\zeta \in [0,b_{\star}]\subset \mathbb{C}\setminus \mathcal{D}_{b_{2}}$, by \eqref{def of R} we have $T(\zeta) = R(\zeta)P^{(\infty)}(\zeta)$, and thus
\begin{align*}
{\rm Tr}\big[ T^{-1}T'\sigma_{3} \big] & = {\rm Tr}\big[ (P^{(\infty)})^{-1}(P^{(\infty)})'\sigma_{3} \big] +  {\rm Tr}\big[ (P^{(\infty)})^{-1}R^{-1}R'P^{(\infty)}\sigma_{3} \big] \\
& =  {\rm Tr}\big[ (P^{(\infty)})^{-1}(P^{(\infty)})'\sigma_{3} \big] + \bigO(s^{-\frac{\rho}{2}}), \qquad \mbox{as } s \to + \infty,
\end{align*}
uniformly for $\zeta \in [0,b_{\star}]$, where we have used \eqref{estimate for R}. We recall that $P^{(\infty)}$ is given by
\begin{align*}
P^{(\infty)}(\zeta) = D(\zeta)^{-\sigma_{3}}, \quad \mbox{ where } \quad D(\zeta) = \exp \left( i\nu\int_{\Sigma_{5}} \frac{d\xi}{\xi-\zeta} \right) = \exp \left( i\nu \log \left[ \frac{\zeta - b_{2}}{\zeta - b_{1}} \right] \right)
\end{align*}
and where the branch of the logarithm is taken along $\Sigma_{5}$. Thus
\begin{align*}
& {\rm Tr}\big[ (P^{(\infty)})^{-1}(P^{(\infty)})'\sigma_{3} \big] = -2\big( \log D\big)',
\end{align*}
and we find as $s\to +\infty$,
\begin{align*}
I_{2,1} & = - \frac{2}{t}\re \bigg[\int_{[0,b_{\star}]} {\rm Tr}\big[ (P^{(\infty)})^{-1}(P^{(\infty)})'\sigma_{3} \big] \frac{d\zeta}{2\pi i}\bigg] + \bigO(s^{-\frac{\rho}{2}}), \\
& = \frac{4}{t}\re \bigg[\int_{[0,b_{\star}]} \big( \log D\big)' \frac{d\zeta}{2\pi i}\bigg] + \bigO(s^{-\frac{\rho}{2}}) = \frac{2\nu}{\pi t} \log \left| \frac{b_{\star}- b_{2}}{b_{\star}-b_{1}} \right|+ \bigO(s^{-\frac{\rho}{2}}).
\end{align*}
\end{proof}
\begin{lemma}\label{lemma: splitting of I22}
As $s \to + \infty$, we have
\begin{align*}
& I_{2,2} = I_{2,2}^{(1)} + I_{2,2}^{(2)} + I_{2,2}^{(3)} + \bigO(s^{-\frac{\rho}{2}}), \\
& I_{2,2}^{(1)} = - \frac{2}{t}\re \bigg[\int_{\Sigma_{5}\cap\mathcal{D}_{b_{2}}} \Big( s^{\rho} i h'(\zeta) - (\log \widetilde{\mathcal{G}})'(\zeta;s) \Big) \frac{d\zeta}{2\pi i}\bigg], \\
& I_{2,2}^{(2)} = - \frac{2}{t}\re \bigg[s^{\frac{\rho}{2}}\int_{\Sigma_{5}\cap\mathcal{D}_{b_{2}}} f'(\zeta) {\rm Tr} \big[ \Phi_{\mathrm{PC}}^{-1}(s^{\frac{\rho}{2}}f_{b_{2}}(\zeta))\Phi_{\mathrm{PC}}'(s^{\frac{\rho}{2}}f_{b_{2}}(\zeta))\sigma_{3} \big] \frac{d\zeta}{2\pi i}\bigg],\\
& I_{2,2}^{(3)} = - \frac{2}{t}\re \bigg[\int_{\Sigma_{5}\cap\mathcal{D}_{b_{2}}}  {\rm Tr}\big[ \Phi_{\mathrm{PC}}^{-1}(s^{\frac{\rho}{2}}f_{b_{2}}(\zeta))E^{-1}(\zeta;s)E'(\zeta;s) \Phi_{\mathrm{PC}}(s^{\frac{\rho}{2}}f_{b_{2}}(\zeta))\sigma_{3} \big] \frac{d\zeta}{2\pi i}\bigg].
\end{align*} 
\end{lemma}
\begin{proof}
For $\zeta \in \Sigma_{5}\cap\mathcal{D}_{b_{2}}$, by \eqref{def of R} we have $T(\zeta) = R(\zeta)P^{(b_{2})}(\zeta)$, and thus
\begin{align*}
{\rm Tr}\big[ T^{-1}T'\sigma_{3} \big] & = {\rm Tr}\big[ (P^{(b_{2})})^{-1}(P^{(b_{2})})'\sigma_{3} \big] +  {\rm Tr}\big[ (P^{(b_{2})})^{-1}R^{-1}R'P^{(b_{2})}\sigma_{3} \big].
\end{align*}
We recall that $P^{(b_{2})}$ is given by
\begin{align*}
P^{(b_{2})}(\zeta) = E(\zeta;s) \Phi_{\mathrm{PC}}(z)e^{ \frac{s^{\rho} }{2}(ih(\zeta)-\ell) \sigma_{3}} \widetilde{\mathcal{G}}(\zeta;s)^{-\frac{\sigma_{3}}{2}} \quad \mbox{ with } \quad z = s^{\frac{\rho}{2}}f(\zeta),
\end{align*}
and thus
\begin{align*}
& {\rm Tr}\big[ (P^{(b_{2})})^{-1}(P^{(b_{2})})'\sigma_{3} \big] = \Big( s^{\rho} i h' - (\log \widetilde{\mathcal{G}})' \Big) + s^{\frac{\rho}{2}}f' {\rm Tr} \big[ \Phi_{\mathrm{PC}}^{-1}\Phi_{\mathrm{PC}}'\sigma_{3} \big] + {\rm Tr}\big[ \Phi_{\mathrm{PC}}^{-1}E^{-1}E' \Phi_{\mathrm{PC}}\sigma_{3} \big], \\
& {\rm Tr}\big[ (P^{(b_{2})})^{-1}R^{-1}R'P^{(b_{2})}\sigma_{3} \big] =  {\rm Tr}\big[ \Phi_{\mathrm{PC}}^{-1}E^{-1}R^{-1}R'E\Phi_{\mathrm{PC}}\sigma_{3} \big],
\end{align*}
where $\Phi_{\mathrm{PC}}$ and $\Phi_{\mathrm{PC}}'$ are evaluated at $z=s^{\frac{\rho}{2}}f_{b_{2}}(\zeta)$ and the other functions are evaluated at $\zeta$. Thus
\begin{align*}
& \int_{\Sigma_{5}\cap\mathcal{D}_{b_{2}}} {\rm Tr}\big[ T^{-1}T'\sigma_{3} \big] \frac{d\zeta}{2\pi i} = \int_{\Sigma_{5}\cap\mathcal{D}_{b_{2}}} \Big( s^{\rho} i h' - (\log \widetilde{\mathcal{G}})' \Big) \frac{d\zeta}{2\pi i} + s^{\frac{\rho}{2}}\int_{\Sigma_{5}\cap\mathcal{D}_{b_{2}}} f' {\rm Tr} \big[ \Phi_{\mathrm{PC}}^{-1}\Phi_{\mathrm{PC}}'\sigma_{3} \big] \frac{d\zeta}{2\pi i} \\
& + \int_{\Sigma_{5}\cap\mathcal{D}_{b_{2}}}  {\rm Tr}\big[ \Phi_{\mathrm{PC}}^{-1}E^{-1}E' \Phi_{\mathrm{PC}}\sigma_{3} \big] \frac{d\zeta}{2\pi i} + \int_{\Sigma_{5}\cap\mathcal{D}_{b_{2}}}  {\rm Tr}\big[ \Phi_{\mathrm{PC}}^{-1}E_{{2}}^{-1}R^{-1}R'E_{{2}}\Phi_{\mathrm{PC}}\sigma_{3} \big] \frac{d\zeta}{2\pi i}.
\end{align*}
From Appendix \ref{appendix:PC}, we know that
\begin{align}\label{estimate on Phi PC sigma3 Phi inv}
\Phi_{\mathrm{PC}}(z)\sigma_{3}\Phi_{\mathrm{PC}}(z)^{-1} = \bigO(1),
\end{align}
uniformly for $z \in \mathbb{C}$. Therefore, by the cyclic property of the trace, and using also the estimates \eqref{estimate for R} and \eqref{estimate Eb2}, we conclude that
\begin{align*}
\int_{\Sigma_{5}\cap\mathcal{D}_{b_{2}}}  {\rm Tr}\big[ \Phi_{\mathrm{PC}}^{-1}E^{-1}R^{-1}R'E\Phi_{\mathrm{PC}}\sigma_{3} \big] \frac{d\zeta}{2\pi i} = \bigO(s^{-\frac{\rho}{2}}), \qquad \mbox{as } s \to + \infty.
\end{align*}
\end{proof}
\begin{lemma}\label{lemma: asymp for I22p1p}
As $s \to + \infty$, we have
\begin{align*}
I_{2,2}^{(1)} = - \frac{\im (ih(b_{2})) - \im (ih(b_{\star}))}{\pi t}s^{\rho} + \frac{c_{5}+c_{6}}{\pi t} \arg \left(\frac{b_{2}}{b_{\star}}\right) + \bigO \left( s^{-\rho} \right).
\end{align*}
\end{lemma}
\begin{proof}
By definition of $I_{2,2}^{(1)}$, we have
\begin{align*}
I_{2,2}^{(1)} = -\frac{s^{\rho}}{\pi t} \im \int_{\Sigma_{5}\cap\mathcal{D}_{b_{2}}} ih'(\zeta)d\zeta + \frac{1}{\pi t}\im \int_{\Sigma_{5}\cap\mathcal{D}_{b_{2}}}\big( \log \widetilde{\mathcal{G}}(\zeta;s) \big)' d\zeta,
\end{align*}
and
\begin{align}
& \im \int_{\Sigma_{5}\cap\mathcal{D}_{b_{2}}} ih'(\zeta)d\zeta = \im (ih(b_{2})) - \im (ih(b_{\star})), \nonumber \\
& \int_{\Sigma_{5}\cap\mathcal{D}_{b_{2}}}\big( \log \widetilde{\mathcal{G}}(\zeta;s) \big)' d\zeta = \log \widetilde{\mathcal{G}}(b_{2};s)-\log \widetilde{\mathcal{G}}(b_{\star};0). \label{second trivial primitive}
\end{align}
The right-hand-side of \eqref{second trivial primitive} can be expanded as $s \to + \infty$ using \eqref{asymp for G}, and we find
\begin{align*}
\int_{\Sigma_{5}\cap\mathcal{D}_{b_{2}}}\big( \log \widetilde{\mathcal{G}}(\zeta;s) \big)' d\zeta = (c_{5}+c_{6}) \log\left(\frac{b_{2}}{b_{\star}}\right) + \bigO \left( s^{-\rho} \right),
\end{align*}
and the result follows.
\end{proof}
\begin{lemma}\label{lemma: asymp for I22p2p}
Let $m \in \mathbb{C}\setminus \mathbb{R}^{-}$. As $s \to + \infty$, we have
\begin{align*}
I_{2,2}^{(2)} = - \frac{s^{\rho}}{\pi t}\big[ \im(ih(b_{\star}))-\im(ih(b_{2})) \big] - \frac{\nu  \rho}{\pi t}\log s - \frac{2\nu}{\pi t}\log r + \frac{2 \nu}{\pi t} \log |m| + \mathcal{I}_{2,2}^{(2)}(m) + \bigO(s^{-\frac{\rho}{2}}),
\end{align*}
where $r = |f(b_{\star})| = - f(b_{\star})$ and
\begin{align*}
\mathcal{I}_{2,2}^{(2)}(m) = -\frac{2}{t} \re\bigg[\int_{(-\infty,0]} \left({\rm Tr} \big[ \Phi_{\mathrm{PC}}^{-1}(z)\Phi_{\mathrm{PC}}'(z)\sigma_{3} \big] - \left( i z - \frac{2i \nu}{z-m} \right) \right)  \frac{dz}{2\pi i}\bigg].
\end{align*}
\end{lemma}
\begin{proof}
Using the change of variables $z = s^{\frac{\rho}{2}} f_{b_{2}}(\zeta)$ and denoting $r = |f_{b_{2}}(b_{\star})| = - f_{b_{2}}(b_{\star})$, we rewrite $I_{2,2}^{(2)}$ as
\begin{align}\label{divergent integral for I22p2p}
I_{2,2}^{(2)}& = - \frac{2}{t}\re \bigg[ \int_{[-s^{\frac{\rho}{2}}r,0]} {\rm Tr} \big[ \Phi_{\mathrm{PC}}^{-1}(z)\Phi_{\mathrm{PC}}'(z)\sigma_{3} \big]  \frac{dz}{2\pi i}\bigg].
\end{align}
From the expansion \eqref{Phi PC at inf}, we get
\begin{align*}
{\rm Tr} \big[ \Phi_{\mathrm{PC}}^{-1}(z)\Phi_{\mathrm{PC}}'(z)\sigma_{3} \big] & = iz - \frac{2i \nu}{z} + \bigO(z^{-2}), \qquad \mbox{as } z \to -\infty.
\end{align*}
Let $m \in \mathbb{C}\setminus \mathbb{R}^{-}$. We have
\begin{align}
\int_{[-s^{\frac{\rho}{2}}r,0]} {\rm Tr} \big[ \Phi_{\mathrm{PC}}^{-1}(z)\Phi_{\mathrm{PC}}'(z)\sigma_{3} \big]  \frac{dz}{2\pi i} & = \int_{[-s^{\frac{\rho}{2}}r,0]} \left( i z - \frac{2i \nu}{z-m} \right)  \frac{dz}{2\pi i} \label{explicit integral for I22p2p} \\
& + \int_{[-s^{\frac{\rho}{2}}r,0]} \left({\rm Tr} \big[ \Phi_{\mathrm{PC}}^{-1}(z)\Phi_{\mathrm{PC}}'(z)\sigma_{3} \big] - \left( i z - \frac{2i \nu}{z-m} \right) \right)  \frac{dz}{2\pi i}. \nonumber
\end{align}
Since 
\begin{align*}
& {\rm Tr} \big[ \Phi_{\mathrm{PC}}^{-1}(z)\Phi_{\mathrm{PC}}'(z)\sigma_{3} \big] - \left( i z - \frac{2i \nu}{z-m} \right) = \bigO(z^{-2}), \qquad \mbox{as } z \to - \infty,
\end{align*}
we have
\begin{align*}
& \int_{[-s^{\frac{\rho}{2}}r,0]} \left({\rm Tr} \big[ \Phi_{\mathrm{PC}}^{-1}(z)\Phi_{\mathrm{PC}}'(z)\sigma_{3} \big] - \left( i z - \frac{2i \nu}{z-m} \right) \right)  \frac{dz}{2\pi i}   \\
& =\int_{[-\infty,0]} \left({\rm Tr} \big[ \Phi_{\mathrm{PC}}^{-1}(z)\Phi_{\mathrm{PC}}'(z)\sigma_{3} \big] - \left( i z - \frac{2i \nu}{z-m} \right) \right)  \frac{dz}{2\pi i} + \bigO(s^{- \frac{\rho}{2}}), \qquad \mbox{as } s \to + \infty.
\end{align*}
On the other hand, the first integral on the right-hand-side of \eqref{explicit integral for I22p2p} can be easily expanded as follows:
\begin{align*}
\int_{[-s^{\frac{\rho}{2}}r,0]} \left( i z - \frac{2i \nu}{z-m} \right)  \frac{dz}{2\pi i} & = - \frac{s^{\rho}r^{2}}{4\pi} + \frac{\nu}{\pi}\log \Big( \frac{s^{\frac{\rho}{2}}r}{m}+1 \Big)  \\
& = - \frac{s^{\rho}r^{2}}{4\pi} + \frac{\nu \rho}{2\pi} \log s + \frac{\nu}{\pi}\log \frac{r}{m} + \bigO(s^{-\frac{\rho}{2}}), \qquad \mbox{as } s \to + \infty.
\end{align*}
Therefore, as $s \to + \infty$ we have
\begin{align*}
I_{2,2}^{(2)} = & - \frac{2}{t}\re \Big[ - \frac{s^{\rho}r^{2}}{4\pi} + \frac{\nu \rho}{2\pi} \log s + \frac{\nu}{\pi}\log \frac{r}{m}  \Big]  \\
& - \frac{2}{t}\re \bigg[ \int_{(-\infty,0]} \left({\rm Tr} \big[ \Phi_{\mathrm{PC}}^{-1}(z)\Phi_{\mathrm{PC}}'(z)\sigma_{3} \big] - \left( i z - \frac{2i \nu}{z-m} \right) \right)  \frac{dz}{2\pi i} \bigg] + \bigO(s^{-\frac{\rho}{2}}),
\end{align*}
and the claim follows by noticing that
\begin{align*}
& r^{2} = f(b_{\star})^{2} = -2 \big( h(b_{\star})-h(b_{2}) \big) = -2\big( \im(ih(b_{\star}))-\im(ih(b_{2})) \big).
\end{align*}
\end{proof}
\begin{lemma}\label{lemma: I22p3p}
\begin{align*}
& I_{2,2}^{(3)} = \bigO(1), & & \mbox{as } s \to + \infty, \\
& I_{2,2}^{(3)} = \bigO(b_{\star}-b_{2}), & & \mbox{as } b_{\star} \to b_{2}.
\end{align*}
\end{lemma}
\begin{proof}
This follows from the previous estimates \eqref{estimate Eb2} and \eqref{estimate on Phi PC sigma3 Phi inv}, and the cyclic property of the trace. 
\end{proof}
\begin{lemma}\label{lemma: asymp for I2}
As $s \to + \infty$, we have
\begin{align*}
I_{2} = - \frac{\nu  \rho}{\pi t}\log s - \frac{2\nu}{\pi t} \log \left| (b_{2}-b_{1})f'(b_2) \right| + \mathcal{I}_{2,2}^{(2)}(1) + \bigO(s^{-\frac{\rho}{2}}).
\end{align*}
\end{lemma}
\begin{proof}
By combining Lemmas \ref{lemma: asymp for I21}, \ref{lemma: splitting of I22}, \ref{lemma: asymp for I22p1p}, \ref{lemma: asymp for I22p2p} and \ref{lemma: I22p3p}, as $s \to + \infty$ we have 
\begin{align}
 I_{2} = & \; I_{2,1} + I_{2,2} = I_{2,1} + I_{2,2}^{(1)} + I_{2,2}^{(2)} + I_{2,2}^{(3)} + \bigO(s^{-\frac{\rho}{2}}) \label{I2 asymp in proof mess} \\
= & \; \frac{2\nu}{\pi t} \log \left| \frac{b_{\star}- b_{2}}{b_{\star}-b_{1}} \right| - \frac{\im (ih(b_{2})) - \im (ih(b_{\star}))}{\pi t}s^{\rho} + \frac{c_{5}+c_{6}}{\pi t} \arg \left(\frac{b_{2}}{b_{\star}}\right) \nonumber \\
& \; - \frac{s^{\rho}}{\pi t}\big[ \im(ih(b_{\star}))-\im(ih(b_{2})) \big] - \frac{\nu  \rho}{\pi t}\log s - \frac{2\nu}{\pi t}\log |f(b_{\star})| + \frac{2 \nu}{\pi t} \log |m| + \mathcal{I}_{2,2}^{(2)}(m)  + I_{2,2}^{(3)} + \bigO(s^{-\frac{\rho}{2}}) \nonumber \\
= & \; - \frac{\nu  \rho}{\pi t}\log s + \frac{2\nu}{\pi t} \log \left| \frac{b_{\star}- b_{2}}{(b_{\star}-b_{1})f(b_{\star})} \right| + \frac{c_{5}+c_{6}}{\pi t} \arg \left(\frac{b_{2}}{b_{\star}}\right) + \frac{2 \nu}{\pi t} \log |m| + \mathcal{I}_{2,2}^{(2)}(m) + I_{2,2}^{(3)} + \bigO(s^{-\frac{\rho}{2}}). \nonumber
\end{align}
The term of order $\mathcal O(1)$ as $s\to +\infty$ in this expansion is given by
\begin{align}\label{term of order 1 what a mess}
\frac{2\nu}{\pi t} \log \left| \frac{b_{\star}- b_{2}}{(b_{\star}-b_{1})f(b_{\star})} \right| + \frac{c_{5}+c_{6}}{\pi t} \arg \left(\frac{b_{2}}{b_{\star}}\right) + \frac{2 \nu}{\pi t} \log |m| + \mathcal{I}_{2,2}^{(2)}(m) + \mathcal{I}_{2,2}^{(3)}.
\end{align}
We simplify this term by noticing that the disks can be chosen arbitrarily small (though independent of $s$). Therefore it is possible to evaluate \eqref{term of order 1 what a mess} simply by taking the limit $b_{\star} \to b_{2}$. As $b_{\star} \to b_{2}$, we have
\begin{align*}
& \frac{b_{\star}- b_{2}}{(b_{\star}-b_{1})f(b_{\star})} = \frac{1}{(b_{2}-b_{1})f'(b_2)} + \bigO(b_{\star}-b_{2}), \qquad \arg \left(\frac{b_{2}}{b_{\star}}\right) = \bigO(b_{\star}-b_{2}), \qquad I_{2,2}^{(3)} = \bigO(b_{\star}-b_{2}),
\end{align*}
where we have used Lemma \ref{lemma: I22p3p}. Therefore, taking the limit $b_{\star} \to b_{2}$ in \eqref{term of order 1 what a mess} and then substituting in \eqref{I2 asymp in proof mess}, we obtain
\begin{align}\label{lol7}
I_{2} = - \frac{\nu  \rho}{\pi t}\log s - \frac{2\nu}{\pi t} \log \left| (b_{2}-b_{1})f'(b_2) \right| + \frac{2 \nu}{\pi t} \log |m| + \mathcal{I}_{2,2}^{(2)}(m) + \bigO(s^{-\frac{\rho}{2}}), \qquad \mbox{ as } s \to + \infty.
\end{align}
We have the freedom to choose $m \in \mathbb{C}\setminus \mathbb{R}^{-}$. The claim follows after setting $m=1$ in \eqref{lol7}.
\end{proof}
\begin{lemma}
For $j=1,2$, we have
\begin{align}
\partial_{t} \log \det \left( 1 - (1-t)\mathbb{K}^{(j)}\Big|_{[0,s]} \right) = & \; \frac{\re b_{2}}{\pi \, \rho \, t} s^{\rho} - \frac{\nu  \rho}{\pi t}\log s + \frac{c_{1}c_{6}-c_{2}c_{5}}{t(c_{1}+c_{2})} - \frac{2\nu}{\pi t} \log \left| (b_{2}-b_{1})f'(b_2) \right| \nonumber \\
&  + \partial_{t}\big[ \log \big( G(1+i\nu)G(1-i\nu) \big) \big] + \bigO(s^{-\frac{\rho}{2}}), \qquad \mbox{as }s \to + \infty, \label{asymp der log Det t}
\end{align}
where $G$ is Barnes' $G$-function.
\end{lemma}
\begin{proof}
It follows from Lemmas \ref{lemma:I1}, \ref{lemma: asymp for Ib2 final} and \ref{lemma: asymp for I2} that
\begin{align}
& \partial_{t} \log \det \left( 1 - (1-t)\mathbb{K}^{(j)}\big|_{[0,s]} \right) =  I_{1} + I_{2} + 2 \, \re I_{b_{2}} +\bigO(e^{-c\rho}) \nonumber \\
& = \frac{\re b_{2}}{\pi \, \rho \, t} s^{\rho} - \frac{\nu  \rho}{\pi t}\log s + \frac{c_{1}c_{6}-c_{2}c_{5}}{t(c_{1}+c_{2})} - \frac{2\nu}{\pi t} \log \left| (b_{2}-b_{1})f'(b_2) \right| + \chi(t) + \bigO(s^{-\frac{\rho}{2}}) \label{lol3}
\end{align}
where
\begin{align*}
\chi(t) := 2 \, \re \mathcal{I}_{b_{2}} + \mathcal{I}_{2,2}^{(2)}(1).
\end{align*}
It is rather difficult to obtain an explicit expression for $\chi(t)$ from a direct analysis. However, it follows from Lemmas \ref{lemma: asymp for Ib2 final} and \ref{lemma: asymp for I2} that $\chi(t)$ depends on $t$ but is independent of the other parameters. More precisely, for $j=1$, $\chi(t)$ is independent of $r$, $q$, $\nu_{1},...,\nu_{r}$, $\mu_{1},...,\mu_{q}$, and for $j=2$, $\chi(t)$ is independent of $\alpha$ and $\theta$. We will take advantage of that by using the known result from \cite{Charlier} for the Bessel point process given by \eqref{moment asymp Bessel}. If $j=1$, then we set $r=1$, $q=0$ and $\nu_{1} = 0$, and if $j=2$, we set $\theta = 1$ and $\alpha = 0$. In these cases, $\re b_{2} = 1$, $\rho = \frac{1}{2}$, $c_{1} = c_{2} = 1$, $c_{5} = c_{6} = 0$, $f'(b_{2}) = \sqrt{2}$ and \eqref{lol3} becomes
\begin{align}\label{lol8}
\partial_{t} \log \det \left( 1 - (1-t)\mathbb{K}^{(j)}\big|_{[0,s]} \right) =  \frac{2}{\pi \, t} \sqrt{s} - \frac{\nu}{2\pi t}\log s - \frac{\nu}{\pi t} \log 8 + \chi(t) + \bigO(s^{-\frac{1}{4}}).
\end{align}
On the other hand, the asymptotics \eqref{moment asymp Bessel} can be differentiated with respect to $t$ (this follows from the analysis done in \cite{Charlier}), and we get as $s\to +\infty$,
\begin{align}
& \partial_{t} \log \det \left( 1 - (1-t)\mathbb{K}_{\mathrm{Be}}\Big|_{[0,4s]} \right) = \partial_{t}\left(-4\nu \sqrt{s} + \nu^{2} \log(8 \sqrt{s}) + \log \big( G ( 1+i\nu ) G ( 1-i\nu ) \big) + \bigO\big( \tfrac{\log s}{\sqrt{s}} \big)\right) \label{lol9} \\
& = \frac{2}{\pi t}\sqrt{s} - \frac{\nu}{2\pi t}\log s - \frac{\nu}{\pi t}\log 8 + \partial_{t}\big( \log \big( G(1+i\nu)G(1-i\nu) \big) + \bigO\big( \tfrac{\log s}{\sqrt{s}} \big). \nonumber
\end{align}
By \eqref{Bessel and Wright kernel} and \eqref{Mejer G Wright relation kernel}, the left-hand sides of \eqref{lol8} and \eqref{lol9} are equal, and this yields the relation
\begin{align*}
\chi(t) = \partial_{t}\big( \log \big( G(1+i\nu)G(1-i\nu) \big)\big).
\end{align*}
\end{proof}
\begin{lemma}
As $s \to + \infty$, we have
\begin{align*}
\log \det \left( 1 - (1-t)\mathbb{K}^{(j)}\Big|_{[0,s]} \right) = & \; -\frac{2 \nu \re b_{2}}{\rho} s^{\rho} + \nu^{2}\rho \log s -2\pi \nu \frac{c_{1}c_{6}-c_{2}c_{5}}{c_{1}+c_{2}} + 2 \nu^{2} \log \left| (b_{2}-b_{1})f'(b_2) \right| \\
& + \log \big( G(1+i\nu)G(1-i\nu) \big) + \bigO(s^{-\frac{\rho}{2}}) 
\end{align*}
\end{lemma}
\begin{proof}
It suffices to integrate \eqref{asymp der log Det t} in $t$.
\end{proof}
Thus the constants $C = C^{(j)}$, $j=1,2$ of Theorems \ref{thm:expmoments Wrights} and \ref{thm:expmoments Meijer} are given by
\begin{align*}
\log C = -2\pi \nu \frac{c_{1}c_{6}-c_{2}c_{5}}{c_{1}+c_{2}} + 2 \nu^{2} \log \left| (b_{2}-b_{1})f'(b_2)\right| + \log \big( G(1+i\nu)G(1-i\nu) \big).
\end{align*}
This expression can be computed more explicitly by substituting the values for the constants $c_{1},c_{2},c_{5},c_{6}$ given at the beginning of Section \ref{subsection: saddles point analysis}, and the values \eqref{def of b2} and \eqref{cb2 and cb2p2p} for $b_{2}$, $b_{1}$, and $f'(b_2)$.
\appendix
\section{Parabolic Cylinder model RH problem}\label{appendix:PC}

\begin{figure}
\begin{center}
\begin{tikzpicture}
\node at (0,0) {};
\fill (0,0) circle (0.1cm);
\node at (0,-0.3) {$0$};

\draw[line width = 0.45mm,->-=0.55,black] (0,0) to [out=45, in=180+45] (2,2);
\draw[line width = 0.45mm,->-=0.55,black] (0,0) to [out=135, in=180+135] (-2,2);
\draw[line width = 0.45mm,->-=0.55,black] (0,0) to [out=-45, in=180-45] (2,-2);
\draw[line width = 0.45mm,->-=0.55,black] (0,0) to [out=-135, in=180-135] (-2,-2);
\draw[line width = 0.45mm,->-=0.55,black] (-3,0) to [out=0, in=-180] (0,0);

\end{tikzpicture}
\end{center}
\caption{\label{fig: contour for PC}\textit{The jump contour $\Sigma_{\mathrm{PC}}$ for the model RH problem for $\Phi_{\mathrm{PC}}$.}}
\end{figure}
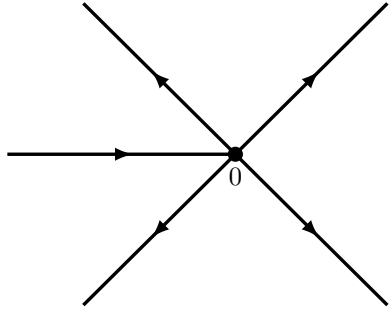
Let $q \in \mathbb{T} = [0,1) \cup i[0,+\infty)$ and let 
\begin{align*}
\nu := - \frac{1}{2\pi} \log(1-q^{2}) \in \mathbb{R}.
\end{align*}
Consider the following model RH problem.
\subsubsection*{RH problem for $\Phi_{\rm PC}$}
\begin{itemize}
\item[(a)] $\Phi_{\mathrm{PC}} : \mathbb{C}\setminus \Sigma_{\mathrm{PC}} \to \mathbb{C}^{2 \times 2}$ is analytic, where
\begin{align*}
\Sigma_{\mathrm{PC}} = \mathbb{R}^{-} \cup \bigcup_{j=0}^{3} e^{\frac{\pi i}{4} + j \frac{\pi i}{2}} \mathbb{R}^{+},
\end{align*}
as shown in Figure \ref{fig: contour for PC}.
\item[(b)] With the contour $\Sigma_{\mathrm{PC}}$ oriented as in Figure \ref{fig: contour for PC}, $\Phi_{\mathrm{PC}}$ satisfies the jumps
\begin{align*}
& \Phi_{\mathrm{PC},+}(z) = \Phi_{\mathrm{PC},-}(z) \begin{pmatrix}
1 & -q \\
0 & 1
\end{pmatrix}, & & z \in e^{\frac{\pi i}{4}}\mathbb{R}^{+}, \\
& \Phi_{\mathrm{PC},+}(z) = \Phi_{\mathrm{PC},-}(z) \begin{pmatrix}
1 & 0 \\
- \frac{q}{1-q^{2}} & 1
\end{pmatrix}, & & z \in e^{\frac{3\pi i}{4}}\mathbb{R}^{+}, \\
& \Phi_{\mathrm{PC},+}(z) = \Phi_{\mathrm{PC},-}(z) \begin{pmatrix}
1 & \frac{q}{1-q^{2}} \\
0 & 1
\end{pmatrix}, & & z \in e^{-\frac{3\pi i}{4}}\mathbb{R}^{+}, \\
& \Phi_{\mathrm{PC},+}(z) = \Phi_{\mathrm{PC},-}(z) \begin{pmatrix}
1 & 0 \\
q & 1
\end{pmatrix}, & & z \in e^{-\frac{\pi i}{4}}\mathbb{R}^{+}, \\
& \Phi_{\mathrm{PC},+}(z) = \Phi_{\mathrm{PC},-}(z) \begin{pmatrix}
\frac{1}{1-q^{2}} & 0 \\
0 & 1-q^{2}
\end{pmatrix}, & & z \in \mathbb{R}^{-}.
\end{align*}
\item[(c)] As $z \to 0$, we have $\Phi_{\mathrm{PC}}(z) = \bigO(1)$.

As $z \to \infty$, $\Phi_{\mathrm{PC}}$ admits an asymptotic series of the form
\begin{align}\label{Phi PC at inf}
\Phi_{\mathrm{PC}}(z) \sim \bigg( I + \sum_{k=1}^{\infty} \frac{\Phi_{\mathrm{PC},k}(q)}{z^{k}} \bigg) z^{-i\nu \sigma_{3}} e^{\frac{i z^{2}}{4}\sigma_{3}},
\end{align}
where the principal branch is taken for $z^{\pm i\nu}$, and where
\begin{align}
& \Phi_{\mathrm{PC},1}(q) = \begin{pmatrix}
0 & \beta_{12}(q) \\
\beta_{21}(q) & 0
\end{pmatrix},  \label{PhiPC1 and beta def} \\
& \Phi_{\mathrm{PC},2}(q) = \begin{pmatrix}
\frac{(1+i\nu)\nu}{2} & 0 \\
0 & \frac{(1-i\nu)\nu}{2}
\end{pmatrix}, \nonumber \\
& \Phi_{\mathrm{PC},2k-1}(q) = \begin{pmatrix}
0 & \Phi_{\mathrm{PC},2k-1}(q)_{12} \\
\Phi_{\mathrm{PC},2k-1}(q)_{21} & 0
\end{pmatrix}, \qquad k \geq 2, \nonumber \\
& \Phi_{\mathrm{PC},2k}(q) = \begin{pmatrix}
\Phi_{\mathrm{PC},2k}(q)_{11} & 0 \\
0 & \Phi_{\mathrm{PC},2k}(q)_{22}
\end{pmatrix}, \hspace{1.8cm} k \geq 2, \nonumber
\end{align}
where
\begin{align}\label{def of beta 12 and beta 21}
& \beta_{12}(q) = \frac{e^{-\frac{3\pi i}{4} }e^{-\frac{\pi \nu}{2}}\sqrt{2\pi}}{q \Gamma(i \nu)} \qquad \mbox{ and } \qquad \beta_{21}(q) = \frac{e^{\frac{3\pi i}{4} }e^{-\frac{\pi \nu}{2}}\sqrt{2\pi}}{q \Gamma(-i \nu)}.
\end{align}
\end{itemize}
The solution $\Phi_{\mathrm{PC}}(z) = \Phi_{\mathrm{PC}}(z;q)$ can be expressed in terms of the parabolic cylinder function $D_{a}(z)$ (see \cite[Chapter 12]{NIST} for a definition). RH problems related to parabolic cylinder functions were first studied in \cite{I1981}, and first used in a steepest descent analysis in \cite{DeiftZhou}. The solution to the above RH problem for $q \in [0,1)$ is known and can be found in e.g. \cite[Appendix B]{Lenells}. However, for $q \in i(0,+\infty)$, the RH problem for $\Phi_{\mathrm{PC}}$ differs from the one of \cite{Lenells} and, to the best of our knowledge, has not appeared before. Therefore, we construct its explicit solution here.
\begin{lemma}\label{lemma:sol mod Phi PC}
The unique solution to the model RH problem for $\Phi_{\mathrm{PC}}$ is given by 
\begin{align}\label{def of Phi PC in terms of Psi}
\Phi_{\mathrm{PC}}(z) = \Psi(z)B(z)^{-1},
\end{align}
where
\begin{align*}
B(z) = \begin{cases} 
\begin{pmatrix} 1 & -q \\ 0  & 1 \end{pmatrix}, & \arg z \in (0, \frac{\pi}{4}), 	
	\\
\begin{pmatrix} 1 & 0 \\  \frac{q}{1 - q^2}  & 1 \end{pmatrix}, & \arg z \in (\frac{3\pi}{4}, \pi), 
	\\
\begin{pmatrix} 1 & \frac{q}{1 - q^2}  \\ 0 & 1 \end{pmatrix}, & \arg z \in (-\pi, -\frac{3\pi}{4}), 
	\\
\begin{pmatrix} 1 & 0 \\ -q & 1 \end{pmatrix}, & \arg z \in (-\frac{\pi}{4},0),
	\\
I, & \mbox{elsewhere},
\end{cases}
\end{align*}
and
\begin{align}\label{psiqzdef}
  \Psi(z) = \begin{pmatrix} \psi_{11}(z) & \frac{\bigl(-i\frac{d}{dz} + \frac{z}{2}\bigr)\psi_{22}(z)}{ \beta_{21}(q)} \\
\frac{\bigl(i\frac{d}{dz} + \frac{z}{2}\bigr)\psi_{11}(z)}{ \beta_{12}(q)}  &  \psi_{22}(z) \end{pmatrix}, \qquad q \in \mathbb{T},\ z \in \C\setminus \R,  
\end{align}
where the functions $\psi_{11}$ and $\psi_{22}$ are defined by
\begin{subequations}
\begin{align}\label{psi11def}
& \psi_{11}(z) = \begin{cases}
e^{\frac{\pi \nu}{4}} D_{-i\nu}(e^{-\frac{\pi i}{4}} z), & \im z > 0, \\
e^{-\frac{3\pi \nu}{4}} D_{-i\nu}(e^{\frac{3\pi i}{4}} z), & \im z < 0,
\end{cases}
	\\ \label{psi22def}
& \psi_{22}(z) = \begin{cases}
e^{-\frac{3\pi \nu}{4}} D_{i\nu}(e^{-\frac{3\pi i}{4}} z), & \im z > 0, \\
e^{\frac{\pi \nu}{4}} D_{i\nu}(e^{\frac{\pi i}{4}} z), & \im z < 0.
\end{cases}
\end{align}
\end{subequations}
\end{lemma}
\begin{proof}
It is a classical fact (see e.g. \cite[Chapter 12]{NIST}) that $D_{a}(z)$ is an entire functions in both $a$ and $z$, which satisfies the second order ODE in $z$
\begin{align*}
D_{a}''(z) = \left(\frac{z^{2}}{4}-\frac{1}{2}-a\right) D_{a}(z).
\end{align*}
Therefore, we verify that the function $\Psi$ defined by \eqref{psiqzdef} satisfies the first order matrix differential equation
\begin{align*}
\Psi'(z) = \frac{iz}{2}\sigma_{3} \Psi(z) - i \begin{pmatrix}
0 & \beta_{12} \\
-\beta_{21} & 0
\end{pmatrix} \Psi(z), \qquad z \in \mathbb{C}\setminus \mathbb{R}.
\end{align*}
Since $\Psi_{+}$ and $\Psi_{-}$ satisfy the same linear differential equation, there exists $J_{\Psi}$ independent of $z$ such that $\Psi_{+}(z) = \Psi_{-}(z) J_{\Psi}$ for $z \in \mathbb{R}$. Using
\begin{align*}
D_{i\nu}(0) = \frac{2^{\frac{i\nu}{2}}\sqrt{\pi}}{\Gamma(\frac{1-i\nu}{2})}, \qquad D_{i\nu}'(0) = - \frac{2^{\frac{1+i\nu}{2}}\sqrt{\pi}}{\Gamma(-\frac{i\nu}{2})},
\end{align*}
we obtain after a computation that
\begin{align*}
J_{\Psi} = \Psi_{-}(0)^{-1}\Psi_{+}(0) = \begin{pmatrix}
1 & -q \\ q & 1-q^{2}
\end{pmatrix}
\end{align*}
where we have also used \eqref{def of beta 12 and beta 21}. Using the jumps for $\Psi$, it is easy to verify that $\Phi_{\mathrm{PC}}$ defined by \eqref{def of Phi PC in terms of Psi} satisfies the jumps of the RH problem for $\Phi_{\mathrm{PC}}$. For each $\delta > 0$, the parabolic cylinder function satisfies the asymptotic formula
\begin{align*}
D_a(z) = & \; z^a e^{-\frac{z^2}{4}}\biggl(1 - \frac{a(a-1)}{2z^2} + O\bigl(z^{-4}\bigr)\biggr)
	\\
& - \widehat{\mathfrak{e}}(z) \frac{\sqrt{2\pi}e^{\frac{z^2}{4}} z^{-a-1}}{\Gamma(-a)} \biggl(1 + \frac{(a+1)(a+2)}{2z^2} + O\bigl(z^{-4}\bigr)\biggr), \quad z \to \infty, \quad a \in \C,
	\\
\widehat{\mathfrak{e}}(z) = & \begin{cases} 0 , & \arg z \in [-\frac{3\pi}{4} + \delta, \frac{3\pi}{4} - \delta], \\
e^{i\pi a}, & \arg z \in [\frac{\pi}{4} + \delta, \frac{5\pi}{4} - \delta], \\
e^{-i\pi a}, & \arg z \in [-\frac{5\pi}{4} + \delta, -\frac{\pi}{4} - \delta], \\
\end{cases} 
\end{align*}
where the error terms are uniform with respect to $a$ in compact subsets and $\arg z$ in the given ranges. Using this formula and the identity
$$D_a'(z) = \frac{z}{2}D_a(z) - D_{a+1}(z),$$
the asymptotic equation \eqref{Phi PC at inf} follows from a tedious but straightforward computation. This shows that $\Phi_{\mathrm{PC}}$ given by \eqref{lemma:sol mod Phi PC} satisfies all the conditions of the RH problem for $\Phi_{\mathrm{PC}}$.
\end{proof}

\paragraph{Formula for $\beta_{12}\beta_{21}$.} Since $\nu \in \mathbb{R}$, we note from \cite[formula 5.4.3]{NIST} that
\begin{align*}
|\Gamma(i\nu)| = \frac{\sqrt{2\pi}}{\sqrt{\nu(e^{\pi \nu}-e^{-\pi \nu})}} = \frac{\sqrt{2\pi}}{\sqrt{\nu \, q^{2} \, e^{\pi \nu}}},
\end{align*}
from which we deduce the identity
\begin{align}\label{beta 12 beta21 relation}
\beta_{12} \beta_{21} = \nu.
\end{align}

\paragraph{Acknowledgements.}
C.C. was supported by the European Research Council, Grant Agreement No. 682537.
T.C. was supported by 
 the Fonds de la Recherche Scientifique-FNRS under EOS project O013018F. The authors are grateful to Gaultier Lambert and Christian Webb for useful discussions, in particular related to Remark \ref{remark:existence}.

\end{document}